\pdfoutput=1
\documentclass[11pt,letterpaper]{article}

\usepackage[margin=1in]{geometry}
\usepackage{amsmath,amssymb,amsthm}
\usepackage{mathtools}
\usepackage{mathrsfs,dsfont}
\usepackage{braket}
\usepackage{thmtools,thm-restate}

\usepackage{graphicx}
\usepackage{subcaption}
\usepackage{tikz}
\usetikzlibrary{positioning,arrows.meta,patterns,calc,shapes}
\usepackage[x11names]{xcolor}
\usepackage{bm}
\usepackage{pifont}

\usepackage{booktabs}
\usepackage{makecell}

\usepackage[inline]{enumitem}
\setlist[itemize]{topsep=2pt,itemsep=1pt,parsep=0pt,partopsep=0pt}
\setlist[enumerate]{topsep=2pt,itemsep=1pt,parsep=0pt,partopsep=0pt}
\makeatletter
\renewcommand\paragraph{\@startsection{paragraph}{4}{\z@}%
  {1.25ex plus .2ex minus .1ex}{-1em}{\normalfont\normalsize\bfseries}}
\makeatother
\usepackage{nicefrac}
\usepackage{csquotes}
\usepackage{soul}
\usepackage{comment, verbatim}
\usepackage{etoolbox, tabto}
\usepackage{xspace, xparse}

\usepackage[x11names]{xcolor}
\usepackage{nicematrix}
\usepackage{tikz}
\usepackage{subcaption}
\usepackage{tcolorbox}
\usepackage{caption}

\usepackage[ruled,linesnumbered]{algorithm2e}

\SetAlFnt{\small}
\SetAlCapFnt{\small}
\SetAlCapNameFnt{\small}
\SetAlCapHSkip{0pt}
\IncMargin{-\parindent}

\usepackage[numbers]{natbib}
\usepackage[
  colorlinks=true,
  linkcolor=blue,
  citecolor=blue,
  urlcolor=blue,
  filecolor=blue,
  pdfborder={0 0 0}
]{hyperref}
\usepackage[capitalize]{cleveref}
\usepackage{fancyhdr}

\DeclarePairedDelimiter{\crl}{\{}{\}}

\newcommand{\R}{\mathbb{R}}
\newcommand{\N}{\mathbb{N}}

\newcommand{\ind}{\mathds{1}}

\DeclareMathOperator*{\E}{\mathbb{E}}

\newcommand{\reals}{\mathbb{R}}

\DeclareMathOperator*{\argmax}{arg\,max}

\ifdefined\sodaproceedingsformat
\newsiamthm{claim}{Claim}
\newsiamthm{conjecture}{Conjecture}
\newsiamthm{observation}{Observation}
\newsiamthm{property}{Property}
\newsiamremark{remark}{Remark}
\newsiamremark{example}{Example}
\newsiamremark{question}{Question}
\newsiamthm{mainthm}{Main Theorem}
\newsiamthm{case}{Case}
\newsiamthm{cond}{Condition}
\newenvironment{sproof}{\begin{proof}[Proof Sketch]}{\end{proof}}
\newenvironment{proofsketch}{\begin{proof}[Proof sketch]}{\end{proof}}
\else

\newtheoremstyle{coloredthm}%
  {3pt}
  {3pt}
  {\itshape\color{blue}}
  {}
  {\bfseries\color{blue}}
  {.}
  { }
  {}

\theoremstyle{plain}
\newtheorem{theorem}{Theorem}[section]
\newtheorem{claim}{Claim}
\newtheorem{lemma}[theorem]{Lemma}
\newtheorem{corollary}[theorem]{Corollary}
\newtheorem{proposition}[theorem]{Proposition}

\newtheorem{definition}{Definition}
\newtheorem{remark}[theorem]{Remark}
\newtheorem*{lemma*}{Lemma}
\newtheorem*{claim*}{Claim}
\newtheorem*{theorem*}{Theorem}
\newtheorem*{corollary*}{Corollary}
\newtheorem*{proposition*}{Proposition}
\newtheorem{observation}{Observation}

\newtheorem{example}[theorem]{Example}

\newtheorem{mainthm}{Main Theorem}

\newtheoremstyle{case}{}{}{}{}{}{:}{ }{}
\theoremstyle{case}

\fi

\DeclareMathOperator{\A}{\auc{A}}
\newcommand{\auc}[1]{\mathsf{#1}}

\newcommand\numberthis{\addtocounter{equation}{1}\tag{\theequation}}

\newcommand{\sqbr}[1]{\left[ #1 \right]}
\newcommand{\brackets}[1]{\left( #1 \right)}
\DeclareMathOperator{\irn}{\Bar{\varphi}}

\renewcommand{\vec}[1]{\bm{#1}}
\newcommand{\bv}{\vec{v}}
\newcommand{\pred}{\mathcal{G}}
\newcommand{\true}{\mathsf{True}}
\newcommand{\false}{\mathsf{False}}
\newcommand{\Uni}{\mathsf{Uni}}

\newcommand{\bid}{\beta}
\newcommand{\vbid}{\vec{\bid}}

\newcommand{\bottomV}{\underline{v}}
\newcommand{\topV}{\overline{v}}
\newcommand{\payoffPay}{\rho}
\newcommand{\conPayoffPay}{\kappa}
\newcommand{\payoffValue}{\omega}
\newcommand{\expect}[1]{\overline{#1}}
\newcommand{\util}{u}
\newcommand{\unconstrained}{\util^{\mathrm{uc}}}
\newcommand{\revMaximizer}{\textsf{IVW-maximizer}\xspace}
\newcommand{\consMaximizer}[1][\delta]{\ifmmode\textsf{IVW-thresh-maximizer}_{#1}\else\textsf{IVW-thresh-maximizer}\textsubscript{$#1$}\fi\xspace}

\newcommand{\aee}{a.e.\xspace}

\newif\ifanonymous

\title{\huge Optimal Auction Design for Constrained Buyers}

\ifanonymous
  \author{Anonymous Submission}
  \date{}
\else
  \author{%
    Batya Berzack\\
    Tel Aviv University\\
    \texttt{batyaberzack@mail.tau.ac.il}
    \and
    Rotem Oshman\\
    Tel Aviv University\\
    \texttt{roshman@tau.ac.il}
    \and
    Inbal Talgam-Cohen\\
    Tel Aviv University\\
    \texttt{inbaltalgam@gmail.com}
  }
  \date{}
\fi

\begin{document}
\setlength{\abovedisplayskip}{4pt}
\setlength{\belowdisplayskip}{4pt}
\setlength{\abovedisplayshortskip}{2pt}
    \setlength{\belowdisplayshortskip}{2pt}

\begin{titlepage}

\pagenumbering{gobble}

\maketitle
\begin{abstract}
We study single-parameter, multi-buyer auctions in which buyers are subject to constraints that affect their bidding strategy. Such constraints arise in many real-world auction settings and fundamentally alter the auction design space: the revelation principle, envelope theorem, and Myerson's lemma no longer hold.
Our focus is on a large family of constraints where the buyers do not have a hard budget-cap, but are restricted in how they can bid or spend their budget --- including no-overbidding, ex-post individual rationality (IR), and stagewise IR.
We ask whether the seller can leverage the buyers' constraints to obtain improved payoff.

We separate between \emph{revenue-aligned} seller objectives
(revenue, welfare, or any linear combination), and \emph{consumer-aligned} seller objectives (e.g.\ consumer surplus).
For the former, we establish a unified theory for all constraints in the family,
which parallels Myerson’s theory of optimal auctions for unconstrained buyers. We develop a
new measure-theoretic technique to show that Myerson-style auctions remain optimal, despite
the altered design space and failure of the classical theory’s central tenets.
In contrast, we show a consumer-aligned seller can leverage the buyers' strategic limitations to strictly outperform mechanisms that are classically incentive compatible. We also design an auction arbitrarily close to optimal for buyers who cannot tolerate temporary debt.

\end{abstract}

\end{titlepage}

\hypersetup{linkcolor=black}
\vspace*{0.5cm}
\setcounter{tocdepth}{2}
\tableofcontents
\hypersetup{linkcolor=blue}
\newpage

\pagenumbering{arabic}
\setcounter{page}{1}

\section{Introduction}\label{sec:intro}

Classical auction theory often makes the idealized assumption that buyers enjoy complete strategic freedom, allowing them to take any action
that maximizes their utility.
However, actual markets often present a different picture:
in many real-world auction settings, buyers face strict limitations that constrain their behavior. For example, buyers may be unable
to risk loss with any probability (\emph{ex-post individual rationality (IR)} constraint~\cite[e.g.,][]{christodoulou2008bayesian,lucier2010price}); they may be unable to go into temporary debt in auctions that unfold over time (\emph{stagewise IR} constraint~\cite[e.g.,][]{berzack2025dynamic,krahmer2025dynamic});
or they may be unable to bid above their true valuation due to bid verification (\emph{no-overbidding} constraint~\cite[e.g.,][]{green1986partially,krahmer2025unidirectional}).
All of these are examples of \emph{buyer constraints},
which may prevent buyers from placing bids or making choices
that would otherwise yield
them positive utility.
While
these examples and others have been studied in an ad-hoc manner (see Related Work),
to date there has not been an attempt to develop a general theory of single-parameter mechanism design with constrained buyers,
and this is the goal of the current paper.

In a world \emph{without} buyer constraints,
single-parameter auction design has a
profound and elegant solution: Myerson's lemma~\cite{myerson1981optimal}, together with the revelation principle~\cite{krishna2009auction},
provides a complete characterization of any truthful auction, essentially solving the auction design problem in this setting.
Although Myerson originally formulated his results for revenue-optimal mechanisms,
the theory is powerful enough to
extend to a wide range of other objectives.
We call auctions that are based on this theory \emph{Myerson-like} auctions.

Interestingly,
the classical theory of single-parameter mechanism design breaks down when buyers are constrained.
Central tenets such as the revelation principle, Myerson's lemma and the envelope theorem~\cite{milgrom2002envelope} no longer hold: not only do their standard proofs not go through, but the results themselves are no longer true.
This is because buyer constraints change the \emph{space of implementable allocation rules};
since buyers are constrained, the auction no longer has to be truthful against all bids,
only against those that buyers are actually able to place.

To illustrate this,
consider a simple setting with two units of an item
and a single ex-post IR constrained buyer,
whose valuation per unit is supported over
$[5,7]$.
If the buyer bids in the range $[5,6)$,
they are allocated both units, and charged $10$;
if they bid in the range $[6,7]$, they are allocated
one unit, and are charged $6$ w.p.\ $1/2$ and nothing w.p.\ $1/2$.
The allocation rule of this auction
is \emph{non-monotone} --- a lower bid results in a  higher allocation
--- and in the unconstrained setting,
the auction would not be truthful.
For example, a buyer with valuation $5$ would prefer to bid $6$,
for expected utility $1\cdot 5 - (1/2) \cdot 6 = 2$,
which is better than their expected utility of $2 \cdot 5 - 10 = 0$
if they truthfully report their valuation.
However, when the buyer is ex-post IR,
they \emph{cannot} place this bid:
with probability $1/2$, their utility will be negative,
$1 \cdot 5 - 6 = -1$.
It is easy to verify that the auction is indeed truthful, given ex-post IR buyers.

The example above is perhaps odd,
in that it is clearly \emph{non-optimal}
if our goal is to maximize either welfare or revenue:
to maximize welfare we should just give both units to the buyer, no matter the bid,
and to maximize revenue a posted-price auction is optimal.
Thus, at this point one might still hope that \emph{optimal} auctions are Myerson-like, even though the example shows that not \emph{all} implementable auctions are.
But this is not the case, at least for some seller objectives:
suppose we have a single item and three ex-post IR buyers with valuations
drawn uniformly from $[0,1]$,
and the seller wishes to maximize \emph{consumer surplus}
(the buyers' total value
minus their payments).
The optimal Myerson-like auction for consumer surplus is to allocate the item to a \emph{random} buyer and charge nothing~\cite{HartlineR08}; this results in an expected consumer surplus of $0.5$.
However, by leveraging the ex-post IR constraint,
the seller can obtain higher payoff:
for example, let $b_{(1)}$ and $b_{(2)}$ denote the highest and second-highest bids.
The seller allocates the item to the highest bidder, breaking ties by a fixed rule, and charges their full bid $b_{(1)}$ with probability $\frac{1}{10}\frac{b_{(2)}}{b_{(1)}}$ (and otherwise charges nothing), where the payment is zero if $b_{(1)}=0$.
Thus the winner's expected payment is $b_{(2)}/10$, which is independent of their bid as long as they remain the winner.
A losing buyer can win only by bidding above their value, in which case there is a positive probability that they are charged this bid, so the deviation violates ex-post IR.
Hence ex-post IR buyers are truthful, while unconstrained buyers need not be.
Since the expected highest and second-highest values are $3/4$ and $1/2$, respectively, the expected consumer surplus is $3/4-(1/10)(1/2)=0.7>0.5$.

\paragraph{The general picture.}
Fig.~\ref{fig:rev-aligned-usual-flow} below
depicts
the flow of the classical theory of optimal mechanism design with \emph{unconstrained} buyers,
and points out the components that no longer hold when buyers are constrained (these are marked with \ding{55} in the figure).
The examples above already demonstrate
that Myerson's lemma does not hold in the constrained setting (i.e., truthful auctions are not necessarily monotone,
and do not necessarily have Myerson payments).
The examples can easily be extended to show
the failure of the relationship between expected payment and virtual value and between virtual value and ironed virtual value in all truthful auctions.
We demonstrate the failure of the revelation principle in Example~\ref{ex:failure-of-revelation-principle}.

\begin{figure}[h]

\centering
\resizebox{\textwidth}{!}{%
\begin{tikzpicture}[
    node distance=1.35cm and 0.25cm,
    box/.style={
        rectangle,
        draw=none,
        very thick,
        align=center,
        minimum height=1.35cm,
        text width=3.4cm,
        inner sep=3pt,
        font=\footnotesize
    },
      redbox/.style={
        box,
        draw=black,
        fill=red!20!,
    line width=1.5pt,
    append after command={
      node[anchor=north east, inner sep=2pt, font=\large]
      at (\tikzlastnode.north east) {\ding{55}}
    }
  },
    smalllabel/.style={
    rectangle,
    draw=black,
    fill=white,
    line width=0.6pt,
    rounded corners=1pt,
    inner xsep=5pt,
    inner ysep=2pt,
    font=\footnotesize\sffamily,
    align=center
},
  whitebox/.style={
    box,
    fill=white,
    draw=black,
    line width=0.8pt,
    append after command={
      node[anchor=north east, inner sep=1.2pt, font=\large]
      at (\tikzlastnode.north east) {\ding{52}}
    }
    },
  greyboxwrong/.style={
    box,
    dashed,
        draw=black,
    fill=gray!15,
    append after command={
      node[anchor=north east, inner sep=4pt, font=\large]
      at (\tikzlastnode.north east) {\ding{55}}
    }
  },
  greyboxquestion/.style={
    box,
    dashed,
        draw=black,
    fill=gray!15,
    append after command={
      node[anchor=north east, inner sep=7pt, font=\large]
      at (\tikzlastnode.north east) {\Huge\sffamily\bfseries \textit{?}}
    }
  },
    line/.style={
        draw,
        thick,
        rounded corners,
        line cap=round
    },
    arrow/.style={
        line,
        -{Stealth[scale=1.2]}
    }
]

    \node[redbox, text width=2.8cm] (TC) {Truthful auctions\\are monotone};
    \node[redbox, text width=3.7cm, left=0.2cm of TC] (TL) {Truthful auctions have Myerson payments};
    \node[whitebox, text width=4.2cm, right=0.2cm of TC] (TR) {In monotone auctions,\\${\E[\varphi(v)x(v)] \leq \E[\bar{\varphi}(v)x(v)]}$};

    \node[redbox, minimum height=1.55cm, text width=2.8cm, below=0.35cm of TC] (MC) {Truthfulness is\\w.l.o.g.};
    \node[greyboxwrong, minimum height=1.55cm, text width=3.7cm, left=0.2cm of MC] (ML) {In truthful auctions,\\$\E[p(v)] = \E[\varphi(v)x(v)]$};
    \node[greyboxwrong, minimum height=1.55cm, text width=4.2cm, below=0.35cm of TR] (MR) {In truthful auctions,\\${\E[\varphi(v)x(v)] \leq \E[\bar{\varphi}(v)x(v)]}$};
    \node[whitebox, minimum height=1.55cm, text width=3.9cm, left=0.2cm of ML] (MLL) {A maximizer of ironed virtual welfare has\\
    ${\E[p(v)]=\E[\bar{\varphi}(v)x(v)]}$};

    \path (MLL.west) -- (MR.east) coordinate[midway] (CENTER_X);

    \node[greyboxquestion, minimum height=1.1cm, text width=12.6cm, anchor=north] (Bottom)
        at (CENTER_X |- MC.south) [yshift=-0.55cm]
        {A maximizer of ironed virtual welfare, with Myerson payments, is optimal};

    \draw[arrow,dashed] (TL.south) -- (ML.north);
    \draw[arrow,dashed] (TC.south) to(MR.north);

    \draw[arrow] (TR.south) -- (MR.north);
            \draw[arrow] (MLL.south) -- ([xshift=-4.8cm]Bottom.north);
    \draw[arrow,dashed] (ML.south) -- ([xshift=-1.6cm]Bottom.north);
    \draw[arrow,dashed] (MC.south) -- ([xshift=1.4cm]Bottom.north);
    \draw[arrow,dashed] (MR.south) -- ([xshift=4.8cm]Bottom.north);

    \node[smalllabel, anchor=west] at ([xshift=0.12cm]TL.north west) {Payments};
    \node[smalllabel, anchor=west] at ([xshift=0.12cm]TC.north west) {Monotonicity};
    \node[smalllabel, anchor=west] at ([xshift=0.12cm]TR.north west) {Ironing};
    \node[smalllabel, anchor=west] at ([xshift=0.12cm]MLL.north west) {Candidate};
    \node[smalllabel, anchor=west] at ([xshift=0.12cm]ML.north west) {Identity};
    \node[smalllabel, anchor=west] at ([xshift=0.12cm]MC.north west) {Revelation};
    \node[smalllabel, anchor=west] at ([xshift=0.12cm]MR.north west) {Bound};
    \node[smalllabel, anchor=west] at ([xshift=0.12cm]Bottom.north west) {Optimality};
\end{tikzpicture}
}
\captionsetup{font=small}
\caption{The structure of Myerson's theory, with arrows indicating logical implications.
The theory characterizes the allocation and payment rules for all truthful auctions. It introduces virtual values $\varphi$ and ironed virtual values $\bar{\varphi}$ to characterize the expected revenue.
Statements shown in red boxes with bold borders are \emph{not true} with constrained buyers, and so the dashed arrows represent implications that can no longer be deduced. Thus, statements in grey boxes with dashed borders are not known to be true a-priori
(in fact, those marked with an \ding{55} turn out to be incorrect in our setting). Statements in white boxes with thin borders still hold in our setting.
}
    \label{fig:rev-aligned-usual-flow}
\end{figure}

Without Myerson's characterization,
what is the optimal auction for
a given seller objective and buyer constraint?

\subsection{Our Contributions}\label{sec:results}

In this paper we initiate a systematic study of auctions with constrained buyers,
which encompasses all
of the examples above and more.
We introduce a wide class
of constraints --- those that are \emph{monotone} and \emph{uncapped} (UM constraints),
i.e., they
constrain high-valuation buyers no more than low-valuation ones, and do not impose a firm cap on any buyer's budget.
The constraints are incorporated into the buyers' \emph{utility} by assigning a buyer utility of $-\infty$
whenever an outcome and/or bid violates their constraints.
We remark that several past works have incorporated specific buyer constraints into the utility function (see Appendix~\ref{app:prior-papers-that failed}), and some have implicitly assumed that classic Myerson theory still applies to the resulting \emph{constrained} utility function,
which is incorrect.

Our main finding for UM-constraints is a separation
according to the seller's objective (``payoff''):
for sellers who want to charge as much as possible (e.g., revenue-maximizing) or are agnostic to payments (e.g., welfare-maximizing),
the buyers' strategic constraints cannot be leveraged to improve the seller's payoffs.
On the other hand, sellers who want to \emph{minimize} payments (e.g., consumer surplus maximizers~\cite{ezra2025multi,goldner2025multidimensional,berzack2025dynamic,HartlineR08}) \emph{can} leverage the constraints, sometimes leading to large improvements in the seller's payoff.
One implication of these results is that Myerson's
optimal
auction is robustly optimal for revenue-aligned objectives, even if buyers are UM-constrained.
However, establishing its optimality requires new machinery.

\paragraph{A parallel theory to Myerson's for revenue-aligned objectives and UM-constraints.}

A \emph{revenue-aligned} seller's payoff is roughly composed of a welfare term combined with a revenue term with a positive coefficient (see Section~\ref{sec:auction-prelims} for the formal definition).
In Section~\ref{sec:overview-rev-aligned} we establish a theory parallel to Myerson's for any revenue-aligned objective and any UM-constraints:

\begin{mainthm}[``Parallel Myerson'', informal; see Theorem~\ref{thm:rev-myerson}]
\label{thm:informal-rev-myerson}
    Consider a single-parameter, Bayesian auction setting where each buyer's valuation is drawn from a known (possibly irregular) distribution.
    Assume that each buyer is UM-constrained and that the seller is revenue-aligned.%

        \begin{enumerate}
        \item There exists an optimal \emph{truthful} auction.
        \item Every optimal
         auction is monotone.
        \item In every optimal
        truthful auction,
        there is a unique payment rule (Myerson's payment).
    \end{enumerate}
\end{mainthm}

Note that by necessity, Theorem~\ref{thm:informal-rev-myerson} applies only to \emph{optimal} auctions,
unlike Myerson's theory which applies to all auctions (as we pointed out above, central ingredients such as truthfulness and monotonicity do not hold for all auctions in the constrained setting).
To establish this result, we introduce a technique of \emph{measure-preserving rearrangements}, which we demonstrate in Section~\ref{sec:toy-example}.

\paragraph{Consumer-aligned objectives.}

Despite the failure of Myerson's theory,
the results above show that
\emph{optimal} auctions for revenue-aligned objectives do obey the tenets of classical single-parameter design.
However, this is not the case for \emph{consumer-aligned objectives},
where the seller's payoff assigns negative weight to the revenue term (i.e., it is in the seller's interest to minimize payments; see Section~\ref{sec:auction-prelims} for the formal definition).
For this class of objectives, the seller \emph{can} leverage the buyers' constraints to improve their payoff compared to Myerson-like auctions.
At the same time, there is no universal way to do so:
each constraint rules out a different set of buyer deviations, and an auction
can exploit only the deviations that are no longer available to the buyers.
As a result,
in contrast to revenue-aligned objectives (where we have shown that the seller cannot benefit from the buyers' constraints),
for consumer-aligned objectives there is no single optimal auction structure.

In this paper we focus
on
one UM-constraint:
\emph{stagewise IR} \cite{berzack2025dynamic,krahmer2025dynamic},
where the auction unfolds over time,
and buyers cannot incur debt even temporarily at any point in the auction.
Stagewise IR auctions are only partially understood: recently,~\cite{berzack2025dynamic} gave optimal mechanisms for \emph{single-buyer} auctions.
To demonstrate the applicability of our measure-preserving rearrangement technique,
we develop the following positive result.

\begin{mainthm}[Monotonicity for consumer-aligned auctions with stagewise-IR buyers, informal; see Theorems~\ref{thm:cons-surplus-mono-wlog}-\ref{thm:cons-aligned-mono-wlog}]
\label{thm:informal-consumer-surplus}
Consider a single-parameter, Bayesian auction setting where each buyer's valuation is drawn from a known (possibly irregular) distribution.
    Assume that each buyer is subject to the stagewise-IR constraint and that the seller is consumer-aligned.
    Then in each of the following two cases, any non-monotone auction can be transformed into a monotone auction with at least as high expected payoff:
    \begin{enumerate}
        \item The seller's objective is consumer surplus, or
        \item The seller has any consumer-aligned objective,
        and the auction is \emph{single winner},
        meaning that at most one buyer may receive units (any number of units up to $m$).
    \end{enumerate}

\end{mainthm}

Thus, in these cases we may restrict attention to monotone auctions.
Using this characterization,
we develop an explicit almost-optimal auction for each of the two cases handled in the theorem, when the seller's cost function is concave. Interestingly, the optimal outcome is not always attainable by a DSIC auction.
We remark that single-winner environments arise naturally in the context of rental games~\cite{berzack2025dynamic}. As a consequence, we solve the rental game of~\cite{berzack2025dynamic} when \emph{multiple} agents can arrive each day, where~\cite{berzack2025dynamic} solved only the single-agent case.

\subsection{Our Technique: Measure-Preserving Bid Relabeling}\label{sec:toy-example}

 Our main technique of measure-preserving bid relabeling is a technique that we use in order to transform an arbitrary auction into an auction with specific desired properties, while preserving the payoff of the seller. This allows us to restrict attention to auctions with these properties, which are easier to analyze and optimize over.

At a high level, a measure-preserving bid relabeling \emph{relabels} the bids of a single buyer $i$:
using a mapping $\sigma_i$ from bids to bids,
the new auction allocates the units on bid $b_i$ as the original auction would have on bid $\sigma_i(b_i)$ (keeping the other buyers' bids unchanged).
The mapping $\sigma_i$ is chosen to have some desirable properties depending on the context (e.g., monotonicity), but it is always \emph{measure-preserving}: if $b_i$ is distributed like buyer $i$'s valuation, then so is $\sigma_i(b_i)$. This causes the change to be effectively \emph{invisible} to the other buyers, so that their incentives do not change compared to the original auction.
 Since the new auction only changes which bids lead to which joint allocations, every allocation it makes is also made by the original auction, so feasibility is preserved.

We now give a toy example that illustrates how we use this technique to transform a non-monotone auction into a monotone one, while preserving the seller's payoff, which we do as part of the proof of Theorem~\ref{thm:informal-rev-myerson}.
Here, an auction is \emph{(ex-interim) monotone} if each buyer's expected allocation,
taken over the other buyers' valuations,
is nondecreasing in the buyer's own bid.
As we saw in Section~\ref{sec:intro}, in this constrained-buyers setting, truthful auctions are not necessarily monotone.

We transform the auction one buyer at a time:
for each buyer $i$,
we make buyer $i$'s allocation monotone,
while keeping the outcome distribution, and hence the incentives, of every other buyer unchanged.
A big difficulty is preserving the allocation feasibility of the auction, since feasibility is a joint property of all buyers' allocations. Thus,
changing only buyer $i$'s allocation can allocate more units than are available or lower the seller's payoff.
The following example demonstrates this difficulty and how we overcome it, while many details are omitted, and we focus here on the intuition.

Consider a simple auction $\A$,
depicted in Fig.~\ref{fig:rearrange_a} below,
with two units and two buyers whose valuations are distributed uniformly in $[1,2]$.
We assume that the seller's objective is a linear combination of welfare and revenue.
Each entry $(\textcolor{blue}{x_1},\textcolor{green!50!black}{x_2})$ of the table represents the number of units allocated to buyer $1$ and buyer $2$, respectively,
given their bids $b_1, b_2$.
The auction can be implemented in a truthful manner under,
for example, no-overbidding constraints
(see Appendix~\ref{app:warmup-no-overbidding-payments}).

\begin{figure}[h]
    \begin{center}

    \subcaptionbox{Original auction $\A$\label{fig:rearrange_a}}{%
        $\begin{NiceArray}{c|c|c}[cell-space-top-limit=5pt, cell-space-bottom-limit=5pt]
         & \textcolor{blue}{b_1<1.9} & \textcolor{blue}{b_1\geq 1.9} \\ \hline
        \textcolor{green!50!black}{b_2<1.6}
        & (\textcolor{blue}{2},\textcolor{green!50!black}{0})
        & (\textcolor{blue}{1},\textcolor{green!50!black}{1}) \\
        \textcolor{green!50!black}{b_2\geq 1.6}
        & (\textcolor{blue}{0},\textcolor{green!50!black}{2})
        & (\textcolor{blue}{1},\textcolor{green!50!black}{1})
        \CodeAfter
          \tikz \draw [red, thick, rounded corners=4pt] ([xshift=-6pt]2-|1) rectangle ([xshift=6pt]3-|4);
        \end{NiceArray}$%
    }
    \qquad \raisebox{5ex}{$\Rightarrow$} \qquad
    \subcaptionbox{Rearranged auction $\A'$\label{fig:rearrange_b}}{%
        $\begin{NiceArray}{c|c|c}[cell-space-top-limit=5pt, cell-space-bottom-limit=5pt]
         & \textcolor{blue}{b_1<1.1} & \textcolor{blue}{b_1\geq 1.1} \\ \hline
        \textcolor{green!50!black}{b_2<1.6}
        & (\textcolor{blue}{1},\textcolor{green!50!black}{1})
        & (\textcolor{blue}{2},\textcolor{green!50!black}{0}) \\
        \textcolor{green!50!black}{b_2\geq 1.6}
        & (\textcolor{blue}{1},\textcolor{green!50!black}{1})
        & (\textcolor{blue}{0},\textcolor{green!50!black}{2})
        \end{NiceArray}$%
    }
    \end{center}
    \caption{Toy example of a measure-preserving rearrangement.}
    \label{fig:rearrange}
\end{figure}

Auction $\A$ is not monotonic:
when buyer $1$ bids $b_1 < 1.9$,
their expected allocation is $0.6 \cdot 2 + 0.4 \cdot 0 = 1.2$,
but when they bid higher, $b_1 \geq 1.9$,
their expected allocation is only $0.6 \cdot 1 + 0.4 \cdot 1 = 1$.
The issue arises from
the top row
of $\A$
(i.e., bid $b_2 < 1.6$ of buyer $2$,
marked in red).

How can we fix this violation of monotonicity?
Decreasing buyer $1$'s allocation to $1$ in the top-left cell
hurts the seller's welfare,
and increasing buyer $1$'s allocation to $2$ in the top-right cell results in an infeasible allocation,
because we have only two units.
Swapping between
the left and right cells of the top row
could cause buyer $2$ to bid differently and yield lower payoff,
as it changes the probability with which buyer $2$ receives 1 or 0 units when they bid $b_2 < 1.6$.

The solution,
depicted in Fig.~\ref{fig:rearrange_b},
is to rearrange and relabel the \emph{bids of buyer $1$} (i.e., the columns of the table).
Observe that
each column of joint outcomes still occurs with the same probability as before ($0.1$ and $0.9$, respectively),
so
from the perspective of buyer $2$, nothing has changed, guaranteeing that they will act the same.
In fact, the entire
 distribution of outcomes is preserved.
But the new auction is \emph{monotone},
because for buyer $1$, the ``low-expected-allocation column'' now occurs on lower bids ($b_1 < 1.1$) compared to the ``high-allocation column'' ($b_1 \geq 1.1$).
Buyer $2$'s allocation was already monotone; rearranging buyer $1$'s bids did not alter that.

In terms of the properties above,
the example uses the relabeling $\sigma_1$ that maps buyer $1$'s bids in $[1,1.1)$ to $[1.9,2]$ and bids in $[1.1,2]$ to $[1,1.9)$ (e.g., by shifting).
Since valuations in $[1,1.1)$ and $[1.9,2]$ both have probability $0.1$, $\sigma_1(b_1)$ is distributed like $b_1$, so the relabeling is measure-preserving;
it makes buyer $1$'s allocation monotone;
and since it only moves entire columns of the table, every allocation remains feasible.

The example omits two major challenges, which we address in Section~\ref{sec:overview-rev-aligned}.
First, in general there is no \emph{deterministic} relabeling with the properties we need, so we construct randomized relabelings.
Second, we must redesign the payments of the buyer whose bids are relabeled:
unlike the other buyers, this buyer \emph{does} see a different auction,
so we must ensure that they still bid truthfully, that their payments respect their constraint,
and that the seller's payoff from them does not decrease.
We also use measure-preserving bid relabeling, with a different payment design, to prove Theorem~\ref{thm:informal-consumer-surplus} for consumer-aligned objectives.

\subsection{Related Work}\label{sec:related-work}
Prior work has studied different buyer constraints in a ``pointwise manner'', focusing on one constraint at a time rather than attempting a unified treatment.
We give a brief overview of the most directly relevant prior work here,
specifically
on buyers subject to \emph{no overbidding} or to \emph{stagewise individual rationality (stagewise-IR)}, and on techniques related to our measure-preserving rearrangements;
we discuss additional related work in Appendix~\ref{sec:additional-related}.
Interestingly, some works implicitly rely on Myerson's framework,
which does not hold in constrained-buyer settings (see Appendix~\ref{app:prior-papers-that failed}).

\paragraph{The no-overbidding constraint.}
For buyers that \emph{cannot overbid} (also known as \emph{unidirectional incentive compatibility}),
several prior works establish special cases of our results here,
restricting the number of buyers, the distributions of valuations, or both.
To our knowledge, all prior work deals only with \emph{revenue maximization},
and asks whether the seller can extract more revenue when the buyers cannot overbid.
The question was first studied by Moore~\cite{moore1984global},
who settled it (in the negative) for a \emph{single buyer} with a \emph{finite} valuation distribution, and posed the question of whether the result holds for multiple buyers.
Moore's result was extended to continuous \emph{regular distributions}
in~\cite{krahmer2025unidirectional} (under mild continuity assumptions);
this class of distributions admits a significantly simpler proof compared to the general case we study here
(see Appendix~\ref{sec:additional-related} for details).
The argument of~\cite{krahmer2025unidirectional} for regular distributions appears to extend to multi-buyer auctions (though this is not explicitly claimed).
Finally, independently of our work, \cite{yoon2026sufficiency}
extended Moore's result to the \emph{multi-buyer} setting
for \emph{discrete and finitely-supported} distributions.
Our work removes all of these restrictions and provides a unified answer for any multi-buyer auction,
addressing not only the constraints of no overbidding and revenue maximization but the entire class of UM-constraints and a wide range of objective functions.

\paragraph{The stagewise-IR constraint.}
For buyers subject to \emph{stagewise-IR},
the closest work to ours is~\cite{berzack2025dynamic}, which introduces
a single-buyer auction that has a temporal element, where the buyer is subject to stagewise-IR.
Our results strictly generalize those of~\cite{berzack2025dynamic}, even when applied only in the single-buyer setting.

\paragraph{Related techniques.}
Our measure-preserving rearrangement technique (Section~\ref{sec:overview-rev-aligned}) is related at a high level to black-box reductions that convert an algorithm into a Bayesian incentive compatible mechanism~\cite{hartline2015non,bei2011bayesian,hartline2015bayesian}.
Like our technique, these reductions replace each agent's type by another type drawn from the same distribution, so that the outcome distributions of the other agents are not affected, and choose this replacement so as to obtain incentive compatibility.
However, there are significant differences.
First, these reductions take only an allocation rule as input.
In contrast, we start from an auction with both an allocation rule and a payment rule, where the payments are not necessarily Myerson payments, and we must compete with its equilibrium payoff, including its revenue.
As we discuss before Theorem~\ref{thm:informal-truthification}, in our setting we cannot restrict attention to Myerson payments, and this is a central challenge.
Second, our transformation implements the exact nondecreasing rearrangement of each buyer's interim allocation rule; for continuous valuation distributions, this requires tools from measure theory (Corollary~\ref{cor:random-relabeling}) that are not used in~\cite{hartline2015non,bei2011bayesian,hartline2015bayesian}.
Third, combining the nondecreasing rearrangement with our payment bounds, our transformation applies directly to irregular distributions, without decreasing revenue, welfare, or any revenue-aligned objective in between.
In contrast, the extension of~\cite{hartline2015non} to irregular distributions preserves revenue, but not necessarily welfare.

\subsection{Paper Organization}\label{sec:organization}
In Section~\ref{sec:model} we formally define the auction setting and our model for constrained buyers, including the class of UM constraints that we focus on (Section~\ref{sec:um-constraints}).
In Section~\ref{sec:overview} we give a technical overview of our results for revenue-aligned objectives (Section~\ref{sec:overview-rev-aligned}) and for consumer-aligned objectives (Section~\ref{sec:overview-cons-aligned}).
We then give the full results and proofs in Section~\ref{sec:revenue} for revenue-aligned objectives and in Section~\ref{sec:consumer} for consumer-aligned objectives.
We conclude in Section~\ref{sec:conclusion} with open questions and directions for future research.
We discuss additional related work and give deferred proofs in the appendix.

\section{Our Model}
\label{sec:model}

We begin by formally defining the auction setting which we consider, and then introduce our model for constrained buyers.

\paragraph{Notation.}
We use boldface to denote vectors (e.g., $\bv$),
and
we let $\bv_{-i}$ denote the vector obtained from $\bv$ by removing its $i$'th component.
We abuse notation slightly by letting $(v_i, \bv_{-i})$ denote the vector obtained by slotting value $v_i$ back into position $i$ of $\bv_{-i}$.
Throughout, we let $[n]=\set{1,\dots,n}$.
When working with an auction $\A$,
we use a superscript to indicate components and values associated with $\A$ (e.g., $\vec{x}^{\A}$ is the allocation rule of $\A$).
We may also sometimes use $\vec{x}^{s}$ to denote the allocation rule associated with an auction $\A^{s}$, where $s$ is some superscript,
and similarly for other values associated with $\A^s$.
We omit these superscripts when the auction is obvious from the context.

\subsection{Auction Preliminaries}\label{sec:auction-prelims}

At a high level, we consider a standard Bayesian auction setting with $n$ buyers and $m$ identical units of a single item. Each buyer has a private valuation for the item drawn independently from a known distribution. The seller's objective is to design an auction that maximizes their expected payoff, which is composed of a welfare term, a revenue term, and an allocation cost term.

In general, we also consider auctions that take place over time, which may differ from the standard one-shot auction setting. This is defined formally in Section~\ref{sec:stagewise-outcomes}.

We now describe the setting formally.

\paragraph{Auction settings.}
The setting is single-parameter: each buyer $i \in [n]$ has a private per-item valuation $v_i$ drawn independently from a distribution $F_i$ with support $[\bottomV_i, \topV_i]$,
where $\bottomV_i \ge 0$.
The distributions $\vec{F} = (F_1, \dots, F_n)$ are common knowledge, continuous and atomless,
and admit densities $f_i$ that are strictly positive almost everywhere on $[\bottomV_i, \topV_i]$. The set of feasible allocations, denoted $\Omega$, includes all allocations of at most $m$ units in total.
The only exception is the \emph{single-winner} environment, defined in Section~\ref{sec:stagewise-outcomes}, in which $\Omega$ includes all allocations of at most $m$ units in which at most one buyer receives units.

\paragraph{Auction mechanisms.}
We restrict attention to direct mechanisms,
which are mechanisms in which the bid space  of each buyer is the same as their valuation space, i.e., $[ \bottomV_i, \topV_i]$ for each buyer $i$, but the buyers do not necessarily bid truthfully.
Let $\mathcal{B} = [\bottomV_1, \topV_1] \times \ldots \times [\bottomV_n, \topV_n]$ be the space of all possible bid profiles,
and let $\mathcal{S} = \Omega \times (\reals_{\geq 0})^n$ be the set of outcomes.%
\footnote{This is the outcome space of a one-shot auction. For auctions that unfold over time, an outcome is instead a sequence of per-step allocations and payments for each buyer; see Section~\ref{sec:stagewise-outcomes}.}
An \emph{auction mechanism} (or \emph{auction} for short)
is a mapping
$\A : \mathcal{B} \rightarrow \mathcal{D}(\mathcal{S})$,
where $\mathcal{D}(\mathcal{S})$ is the set of all distributions over outcomes from $\mathcal{S}$.
Note that throughout, we abuse the terminology and notation slightly by conflating a random variable $S_i$, representing an auction outcome for buyer $i$, with its distribution.
The mechanism itself is deterministic,%
\footnote{This is without loss of generality: a randomized mechanism is a distribution over deterministic mechanisms, and by the law of total expectation,
at least one deterministic mechanism in the support has expected payoff at least that of the randomized auction.
Furthermore,
we assume throughout that the buyers' bids are deterministic, which is again w.l.o.g.\
when the mechanism is deterministic.}
but given a bid profile, the result of the auction is a distribution, i.e., a lottery.
Hence, we let $\vec{S}$ denote the random variable representing the result of the mechanism,
which we refer to as the \emph{random outcome};
for a specific realization, we use the term \emph{realized outcome}
and the notation $\vec{s}$. We assume that the buyers all bid according to a deterministic strategy $\vbid=\bid_1\times\dots\times\bid_n$, where $\bid_i$ is a mapping from the valuation of buyer $i$ to their bid.

The \emph{allocation rule} $\vec{x}^{\A} : \mathcal{B} \rightarrow \mathcal{D}(\Omega)$ and \emph{payment rule}
$\vec{p}^{\A} : \mathcal{B} \rightarrow \mathcal{D}(( \reals_{\geq 0})^n)$
of $\A$ are the projections of $\A$ to the allocation and payment, respectively (i.e., a mapping from the bids to the marginal distribution of the allocation/payment).
The auctions that we study can potentially involve a temporal element, where units are allocated over time and the payment is charged over time as well.
Thus, the \emph{outcome} of the auction consists of a \emph{schedule} specifying the units allocated to each buyer and the payments charged at each step.
However, in the technical overview we focus for the most part on two types of auctions:
In Section~\ref{sec:overview-rev-aligned} we work with
\emph{one-shot auctions}, where there is no temporal element --- all units are allocated and all payments are charged at once;
and in Section~\ref{sec:overview-cons-aligned} we introduce and design \emph{threshold auctions},
which involve two timesteps (we prove that more timesteps are not necessary).
Thus, we focus here on one-shot auctions.
In the case of a one-shot auction,
the realized outcome of the auction
is a pair $(\vec{x}, \vec{p})$,
where $\vec{x} \in \set{0,1,\dots,m}^n$ is a feasible allocation
and $\vec{p} \in (\reals_{\geq 0})^n$ is a vector of nonnegative payments.

\paragraph{Auction notations.}
We denote by $x_i^{\A}(\vec{b}), p_i^{\A}(\vec{b})$ the random variables denoting the  allocation and payment  (resp.) of buyer $i$ resulting from applying $\A$ to bid profile $\vec{b}$,
and by $x_i^{\A}(\vec{b},r), p_i^{\A}(\vec{b},r)$  the realized allocation and payment (resp.) for buyer $i$ when the bid profile is $\vec{b}$ and the randomness is $r$.
We let $\expect{x}_i^{\A}(\vec{b}) = \mathbb{E}_{r\sim\Uni[0,1]}[x_i^{\A}(\vec{b},r)]$ denote the expected allocation to buyer $i$ for a specific bid profile $\vec{b}$.
The \emph{interim} expected allocation for buyer $i$ with bid $b_i$ is
 $\expect{x}_i^{\A}(b_i;\vbid) = \E_{\bv_{-i}\sim\Pi_{j\neq i}F_j}[\expect{x}_i^{\A}(b_i, \vbid_{-i}(\bv_{-i}))]$,
 where $\vbid_{-i}$ is the strategy played by other buyers.%
 \footnote{We use this
	 notation even when $\vbid$ is not the bidding strategy for a measure zero set of valuation vectors. However, since the set is of measure zero, this makes no difference on the expected interim outcome.}
	 Additionally, when the strategy played is truthful for buyer $j\neq i$, we may write $\expect{x}_i^{\A}(b_i)$ and omit the $\vbid$ from the notation.
	 When discussing the interim outcome of a type \(v_i\) under a nontruthful strategy profile, we write
	 \(\expect{x}_i^{\A}(\bid_i(v_i);\vbid)\). When buyer \(i\)'s strategy is truthful, or when the argument is explicitly interpreted as a bid, we may write \(\expect{x}_i^{\A}(v_i;\vbid)\) or \(\expect{x}_i^{\A}(v_i)\).
We define $\expect{p}_i^{\A}(\vec{b})$, $\expect{p}_i^{\A}(b_i;\vbid)$ and $\expect{p}_i^{\A}(b_i)$ analogously. Throughout, for any strategy profile $\vbid$ played in $\A$, we assume that the induced interim allocation curve
$v_i\mapsto \expect{x}_i^{\A}(\bid_i(v_i);\vbid)$
is Riemann integrable; this assumption is purely technical and can be replaced by Lebesgue integrability without affecting any of our results.
When the auction mechanism is clear from the context, we abuse the notation by removing the superscript $\A$ from the above notations.

\paragraph{Seller's objective: Revenue-aligned and consumer-aligned.}

The seller aims to maximize in expectation
a \emph{payoff}, which is
a
combination
of a \emph{welfare term}, a linear \emph{revenue term}, and a \emph{cost}.
The first two terms are given by a pair
$\Pi=(\payoffValue,\payoffPay)$, where \(\payoffValue:\reals_{\ge 0}\to\reals_{\ge 0}\)
is a nondecreasing function representing the welfare term,
and $\payoffPay\in\reals$ represents the weight of the revenue in the objective.
The cost function \(c : \N \to \R_{\geq 0}\) is nondecreasing and normalized to $c(0)=0$.
Given an auction $\A$, if a strategy $\vbid$ is played, the expected payoff of the seller is given by
\begin{equation*}
    \E[\A;\vbid,\Pi]=\E_{\bv,r}\sqbr{\sum_{i=1}^n\brackets{\payoffValue(v_i)x^{\A}_i(\vbid(\bv),r)+\payoffPay p^{\A}_i(\vbid(\bv),r)}-c\brackets{\sum_{i=1}^n x^{\A}_i(\vbid(\bv),r)}},
\end{equation*}
where the expectation is over the buyer valuations and the realizations of outcomes.
If it is clear from context what the payoff $\Pi$ is, we may omit it from the notation and simply write $\E[\A;\vbid]$.
Similarly, when it is clear from context which
strategy is played (e.g.\ if an auction is truthful), then we may omit the $\vbid$ from the expectation and simply denote by $\E[\A]$ the expected payoff of $\A$ when $\vbid$ is played.

We consider two classes of payoff functions.
The first one is the class of \emph{revenue-aligned} payoffs where the seller (perhaps weakly) prefers higher payments, so for these, $\payoffPay\geq 0$. Examples of such objectives are:
\begin{itemize}
    \item \textbf{Revenue Maximization:} ($\payoffValue\equiv0, \payoffPay=1$), where the seller cares solely about profit.
    \item \textbf{Welfare Maximization:} ($\payoffValue(v)=v, \payoffPay=0$), where the seller cares solely about welfare.
    \item \textbf{Welfare-Revenue Tradeoff:} ($\payoffValue(v)=(1-\payoffPay)v$), where the seller seeks to maximize both revenue and welfare in a weighted manner.
\end{itemize}
The second class covers \emph{consumer-aligned} payoffs where the seller prefers \emph{lower} payments. For example \textbf{consumer surplus:} ($\payoffValue(v)=v,\payoffPay=-1$), where the seller wants to maximize the total buyers' utility.
 For consumer-aligned payoff, we also need to assume ${\payoffValue(v)}+
 \payoffPay{v}\geq0$ for every $v\geq 0$,
because otherwise the objective is moot.

\paragraph{$\varepsilon$-Optimality.}
For $\varepsilon\geq0$, we say that an auction $\A$ is \emph{$\varepsilon$-optimal}  if $\E[\A]\geq\sup_{\A'}\E[\A']-\varepsilon$, where the supremum is taken over all auctions in the class of interest (e.g., all DSIC and IR auctions).

\subsection{A Model for Constrained Buyers}\label{sec:constrained-buyers}

We now introduce our simple unified model for auctions with constrained buyers.

\paragraph{The acceptability predicate.}
Each buyer's constraint is captured by an \emph{acceptability predicate}
$\pred_i(S_i,b_i;v_i)$ (possibly a different predicate for each buyer), where $S_i$
is a distribution of outcomes for buyer $i$%
\footnote{The outcome of an auction includes an allocation and a payment. However,
when we consider auctions that unfold over time,
the outcome is more complex; see Section~\ref{sec:stagewise-outcomes}.}
and
$b_i, v_i$ are the bid and valuation (resp.) of buyer~$i$.
When $\pred_i(S_i,b_i;v_i) = \true$,
buyer $i$ with valuation $v_i$
is able to place bid $b_i$ and accept outcome distribution $S_i$ after doing so.%
\footnote{All constraints of which we are aware in the literature restrict \emph{either} the bid (e.g., no overbidding) \emph{or} the outcome distribution (e.g., ex-post IR), but not both at the same time. However, for modeling purposes it is cleaner to have one single constraint that deals with both.}
Conceptually, $\pred_i$ encodes hard constraints faced by the buyer rather
than preferences.
For instance, if the constraint is that a buyer cannot overbid, then
$\pred_i(S_i,b_i;v_i)=\true$ if and only if $b_i \le v_i$,
independently of the random outcome.
Another example,
which constrains the outcome this time,
is ex-post IR, which is the constraint
$\pred_i(S_i, b_i ; v_i) = \true$ if and only if $\Pr[ v_i \cdot x_i - p_i < 0] = 0$.
(Here, $x_i, p_i$ are the allocation and payment of buyer $i$;
they are
random variables
drawn from the outcome distribution $S_i$.)
We mention that the general definition of our model,
which includes auctions that unfold over time,
can also express constraints that include a temporal element:
for example, a buyer may be willing to pay more if they receive an item \emph{quickly}.

We require that truthful bidding is not explicitly ruled out by the constraint:
 if an outcome distribution $S_i$ is acceptable with \emph{some} bid $b_i$
 (i.e., $\pred_i(S_i,b_i;v_i)=\true$),
 then it is also acceptable with a truthful bid (i.e., $\pred_i(S_i,v_i;v_i)=\true$).
 We impose no separability assumption on \(\pred_i\): the predicate may
depend jointly on the submitted bid and the induced outcome
distribution. We call \(\pred_i\) a \emph{bid constraint} if it is
independent of \(S_i\), and an \emph{outcome constraint} if it is
independent of \(b_i\).
We require truthful-report compatibility, i.e.\ for every \(S_i,b_i,v_i\), if
$\pred_i(S_i,b_i;v_i)=\true$
then
$\pred_i(S_i,v_i;v_i)=\true$.

\paragraph{Constrained utility.}

The \emph{unconstrained} (or \emph{quasilinear}) utility of buyer $i$ with valuation $v_i$ is, as usual, given by
$\unconstrained_i(S_i;v_i)=x_i(S_i)v_i-p_i(S_i)$.
We define the \emph{constrained} utility of a buyer based on their acceptability
predicate as follows: $u_i(S_i,b_i;v_i)=\unconstrained_i(S_i;v_i)$ if $\pred_i(S_i,b_i;v_i)=\true$, and $u_i(S_i,b_i;v_i)=-\infty$ otherwise.
Unless otherwise specified, the term \emph{utility} refers to the \emph{constrained}
utility function $u_i$.
The predicate $\pred_i$ is evaluated on the full induced outcome distribution, including both the other buyers' randomness and the mechanism's randomness.

When clear from context, we may omit either the outcome or the bid from the utility
function, depending on the constraint.
When considering a specific auction, we may abuse notation and write
$u_i(\vec{b};v_i)$ instead of $u_i(S_i,b_i;v_i)$, where $S_i$ is the random outcome for buyer $i$
of $\A$ given bid vector $\vec{b}$.
As with the allocation and payment, we let $u_i^{\A}(\vec{b};v_i)$ be the realized utility for buyer~$i$ in auction $\A$ given $\vec{b}$,
and define $\expect u_i(\vec{b};v_i)=\E_{r\sim\Uni[0,1]}[u_i(\vec{b},r;v_i)]$ as the expected utility for buyer~$i$, given bid profile $\vec{b}$. We also define $\expect u_i(b_i;v_i,\vbid)=\E_{\bv_{-i}}[\expect u_i(b_i,\vbid_{-i}(\bv_{-i});v_i)]$ as the interim expected utility for buyer~$i$, assuming that the other buyers play according to strategy $\vbid$.

\paragraph{Auction properties in the constrained setting.}
We use the standard notions of BNE, BIC, DSIC, and IR, always with respect to the \emph{constrained} utility. In particular, in a truthful auction, it is possible that there are deviations (to bids or outcomes that are non-acceptable) that can \emph{increase} the buyers' unconstrained utility, but not the constrained utility.
We define these notions here for completeness.
We say that $\vbid$ is a Bayesian Nash equilibrium (BNE) if for every $i$ and $v_i,b_i$,
\begin{align*}
    \expect{\util}_i(\bid_i(v_i);v_i,\vbid)\geq \expect{\util}_i(b_i;v_i,\vbid).
\end{align*}

An auction is \emph{Bayesian incentive-compatible} (BIC) if truthful bidding is a Bayesian Nash equilibrium, i.e., if $\vbid(\bv)=\bv$ is a BNE.
We say that an auction is \emph{dominant-strategy incentive-compatible} (DSIC) if bidding truthfully is a best response to any fixed profile of other bids $\vec{b}_{-i}$. Formally, if for every buyer $i$ and every valuation $v_i$ and bids $b_i, \vec{b}_{-i}$,
\begin{align*}
    \expect{\util}_i((v_i,\vec{b}_{-i});v_i)\geq \expect{\util}_i((b_i,\vec{b}_{-i});v_i).
\end{align*}
An auction is \emph{individually rational} (IR) if the Bayesian Nash equilibrium $\vbid$ played results in nonnegative \emph{interim} expected utility for all of the buyers. Formally, for IR we require that the auction admits a BNE $\vbid$ such that for every $v_i$,
\begin{align*}
    \expect{\util}_i(\bid_i(v_i);v_i,\vbid)\geq0.
\end{align*}
We say that a truthful auction is \emph{ex-post individually rational} (ex-post IR) if truthful bidding gives every buyer nonnegative expected utility for every fixed profile of the other buyers' bids. Formally, for every $i$, $v_i$, and $\vec b_{-i}$,
\begin{align*}
    \expect{\util}_i((v_i,\vec b_{-i});v_i)\geq0.
\end{align*}
This is a property of the mechanism and is distinct from a buyer being subject to the ex-post-IR constraint.
An auction is \emph{allocation-feasible} if in every realized outcome, the resulting allocation lies in $\Omega$ (i.e., the allocation is feasible in the ex-post sense).

We call an auction \emph{U-BIC} or \emph{U-DSIC} if truthful bidding is optimal even with respect to the buyer's \emph{unconstrained} utility.

Unless otherwise stated, all results in this paper apply to auctions that are individually rational, allocation-feasible, admit a Bayesian Nash equilibrium, and have nonnegative payments.
For auctions constructed in the paper, we explicitly verify these properties.

\paragraph{Auction monotonicity.}
    An auction with allocation rule $\vec{x}$ is \emph{ex-interim monotone} if for every $i$, the allocation $\expect{x}_i(\bid_i(v_i);\vbid)$ is nondecreasing in $v_i$, where $\vbid$ is the BNE played.
    Similarly, the auction is \emph{ex-post monotone} if for every $i$ and bid profile $\vec{b}_{-i}$, the allocation function $\expect{x}_i(\bid_i(v_i),\vec{b}_{-i})$ is nondecreasing in $v_i$, where $\bid_i(\cdot)$ is buyer $i$'s strategy.

\paragraph{Measure zero exceptions.}
Due to the nature of our constraints,
allowing utility functions to have discontinuities,
we often have to
deal with exceptions on measure zero sets.
Thus we add the suffix ``a.e.'' (for ``almost everywhere'')
to the names of properties that hold except on $F_i$-measure-zero sets of valuations.
Formally, a property $P$ holds \aee for buyer $i$
if there exists a set $V_i \subseteq [\bottomV_i,\topV_i]$
with $\Pr_{F_i}[V_i]=1$ such that $P$ holds over the domain $V_i$;
similarly a property $P$ holds \aee for an auction if it holds
 \aee for every buyer $i\in[n]$.

 For example, an auction is BIC-\aee for buyer $i$ if
 there exists a set $V_i \subseteq [\bottomV_i,\topV_i]$
 with $\Pr_{F_i}[V_i]=1$ such that for every $v_i \in V_i$,
 bidding truthfully is a Bayesian best response for buyer $i$ with valuation $v_i$.
We note that if some strategy profile $\vbid$ is a BNE-\aee, then
there is a measure zero set of valuation profiles in which case we do
 not know what the buyers bid; however, as mentioned earlier, we assume w.l.o.g.\ that they use a deterministic strategy.

For a set of buyers \(I\subseteq[n]\), we say that \(\vbid\) is a
BNE-\(\aee^I\) if \(\vbid\) is a BNE-\aee, and for every buyer
\(i\in I\), the strategy \(\bid_i\) is a Bayesian best response for
every valuation in \([\bottomV_i,\topV_i]\), rather than only
\(F_i\)-almost everywhere. For singleton sets we write
BNE-\(\aee^{(i)}\) instead of BNE-\(\aee^{\{i\}}\).

For truthful auctions, ex-post monotonicity a.e. means that,
for every buyer $i$ and $F_{-i}$-almost every opponents'
profile $\bv_{-i}$, there exists a set $V_i(\bv_{-i})\subseteq[\bottomV_i,\topV_i]$ of full
$F_i$-measure such that $x_i(\cdot,\bv_{-i})$ is nondecreasing
on $V_i(\bv_{-i})$.

The reason we have these measure-zero exceptions is the constraints; they can cause the buyer's utility function to have discontinuities, which can in turn cause the allocation rule to be non-monotone at a measure-zero set of valuations.

\subsection{Uncapped Monotone Constraints}\label{sec:um-constraints}

In this paper we focus on a class of constraints that we call \emph{uncapped-payment monotone constraints} (UM constraints).
These are defined as follows.

A constraint is \emph{monotone} if
\emph{high-type} buyers can do whatever
\emph{low-type} buyers can. Formally,
if for any $v_i \le v_i'$, whenever
$\pred_i(S_i,b_i;v_i)=\true$, it also holds that
$\pred_i(S_i,b_i;v_i')=\true$.
That is, if a low-type buyer can submit a given bid and accept the induced random outcome without violating the constraint, then a higher-type buyer can do so as well. It follows that in equilibrium, mimicking a low type \emph{never increases} the unconstrained utility, because this would also increase the constrained utility; however, mimicking a high type might \emph{increase} the \emph{unconstrained} utility, but \emph{decrease} the \emph{constrained} utility.

A constraint $\pred_i$ is \emph{uncapped-payment} (or \emph{uncapped} for short) if it never rules out an outcome, resulting from truthful bidding, that is almost surely a one-shot outcome that gives the buyer nonnegative unconstrained utility.
Formally, for every random
individual outcome \(S_i\) and valuation \(v_i\) where
$\Pr_{s_i\sim S_i}
\left[
s_i\text{ is one-shot and }
\unconstrained_i(s_i;v_i)\geq 0
\right]=1$, we require that
$\pred_i(S_i,v_i;v_i)=\true$.
Thus, an uncapped-payment constraint always permits a buyer to
\emph{truthfully} accept a one-shot outcome that gives them nonnegative
unconstrained utility in every realization.
The constraint may restrict how payments are made, or the risk/timing of payments, but it does not impose a hard cap on total payment in such one-shot nonnegative-utility outcomes.
For example, the restrictions may be on the \emph{distribution} of the sum paid, or on the schedule of payments over time (given the allocation).
For a constraint to be uncapped-payment, we require that the bid constraint
allows truthful bids,
and that the outcome constraint must allow any outcome that is almost surely one-shot and
 almost surely results in nonnegative unconstrained utility, i.e.\ $\Pr[\unconstrained_i(S_i;v_i)\geq0]=1$.
 This latter requirement is our way of requiring that the total sum paid is not restricted.

A constraint is an \emph{uncapped-payment monotone constraint}, or a
\emph{UM constraint}, if it is both monotone and uncapped-payment.

Some specific UM constraints that we study in this paper are:
\begin{enumerate}
    \item No-overbidding buyers: buyers who will not report a bid higher than their true valuation.
    This is sometimes due to behavioral reasons, and sometimes due to partial external validation that's available to the seller. For this constraint
    \begin{align*}
        \pred_i(S_i,b_i;v_i)=\begin{cases}
            \true, & b_i\leq v_i\\
            \false, &\text{otherwise}
        \end{cases}.
    \end{align*}
    \item Ex-post individually rational buyers: buyers who will not bid in a way that may, with any positive probability, result in negative utility. For this constraint
    \begin{align*}
        \pred_i(S_i,b_i;v_i)=\begin{cases}
            \true, & \Pr[\unconstrained_i(S_i;v_i)<0]=0\\
            \false, &\text{otherwise}
        \end{cases}.
    \end{align*}
    \item Stagewise individually rational buyers: buyers who cannot go into temporary debt, i.e., their cumulative utility at all times must be nonnegative with probability 1. The formal multi-step definition appears in Section~\ref{sec:stagewise-outcomes}.

\end{enumerate}

\section{Technical Overview}\label{sec:overview}
In this section we give a high-level overview of the main ideas behind our results, with a strong focus on the case of revenue-aligned objectives.

\subsection{Revenue-Aligned Objectives}\label{sec:overview-rev-aligned}

We begin by describing our generalization of Myerson's lemma, which ultimately shows that Myerson-style auctions are still optimal when the seller aims to maximize a \emph{revenue-aligned} objective.
In this section, unless otherwise specified, when we say an auction is ``monotone'', we mean ``ex-interim monotone'', and when we say that an auction is ``truthful'', we mean that it is BIC.

To simplify the exposition,
in this technical overview we
restrict attention to
auctions with no seller cost,
and we consider only
truthful auctions, despite the failure of the revelation principle (in Section~\ref{sec:revenue} we handle the general case of possibly non-truthful auctions, with arbitrary BNEs).
In addition, we only sketch here the \emph{existence} of an optimal auction that satisfies the properties from Theorem~\ref{thm:informal-rev-myerson}; proving these for \emph{every} optimal auction requires a refinement of the arguments, which we do in the full proof.
Finally, we gloss over many technical details,
and ignore all violations of measure 0,
which in general occur throughout the full proof and require  delicate care.

\paragraph{A Myerson-like optimal auction.}
Myerson proves that an auction that pointwise-maximizes the ironed virtual welfare is optimal for revenue.
Here we generalize Myerson's auction to general revenue-aligned objectives:
we first generalize the definition of virtual values
to arbitrary revenue-aligned objectives,
and then define a Myerson-like deterministic auction that we call \revMaximizer,
which pointwise-maximizes the generalized ironed virtual value
and has Myerson payments
(see Section~\ref{sec:revenue} for the formal definition).

In the remainder of this section we sketch our proof that \revMaximizer is \emph{optimal}.
The high-level structure of our proof is depicted in Figure~\ref{fig:rev-aligned-no-overbidding} (for simplicity,
the figure shows the proof only for revenue maximization, but for general revenue-aligned objectives the structure is similar).
As discussed in Section~\ref{sec:intro}, the revelation principle, the envelope theorem, and Myerson's lemma do not hold for constrained buyers. Our alternative proof of optimality therefore focuses on establishing the same \emph{bottom-line statement} as Myerson's theory in all \emph{optimal} auctions,
but along the way we cannot rely on properties of \emph{all} auctions (or of all truthful auctions).

\begin{figure}[h!]
\centering

\resizebox{0.96\textwidth}{!}{%
\begin{tikzpicture}[
    node distance=0.85cm and 0.35cm,
    box/.style={
        rectangle,
        draw=none,
        very thick,
        align=center,
        minimum height=1.35cm,
        text width=4.35cm,
        inner sep=4pt,
        font=\footnotesize
    },
  whitebox/.style={
    box,
    fill=white,
    draw=black,
    line width=0.8pt,
    },
    greenbox/.style={
        box,
        draw=black,
        fill=green!20!,
    },
    line/.style={
        draw,
        thick,
        rounded corners,
        line cap=round,
        shorten <=2pt,
        shorten >=2pt
    },
    arrow/.style={
        line,
        -{Stealth[scale=1.2]}
    }
]
    \node[greenbox] (TL) {Truthful payments are at most Myerson payments (Lem.~\ref{lem:rev-aligned-payment-bound})};
    \node[greenbox, right=of TL] (TC) {\textbf{W.l.o.g., an \emph{optimal} auction is BIC and monotone (Thm.~\ref{thm:rev-aligned-bid-constraint-mono-wlog})}};
    \node[whitebox, right=of TC] (TR) {In monotone auctions,\\${\E[\varphi(v)x(v)] \leq \E[\bar{\varphi}(v)x(v)]}$\\\cite{myerson1981optimal}$^*$};

    \node[greenbox, below=of TL, minimum height=1.7cm] (ML) {In truthful auctions,\\$\E[p(v)]\leq \E[\varphi(v)x(v)]$\\(Cor.~\ref{cor:pay-bound-by-varphi})};
    \node[whitebox, below=of TC] (MLL) {In a maximizer of ironed virtual welfare,\\
    ${\E[p(v)]=\E[\bar{\varphi}(v)x(v)]}$\\
    \cite{myerson1981optimal}$^*$};
    \node[greenbox, below=of TR] (MR) {W.l.o.g., in an\\\emph{optimal} auction,\\${\E[\varphi(v)x(v)] \leq \E[\bar{\varphi}(v)x(v)]}$\\
    (immediate)};

    \path (ML.west) -- (MR.east) coordinate[midway] (CENTER_X);

    \node[greenbox, minimum height=1cm, text width=13.4cm, double, anchor=north] (Bottom)
        at (CENTER_X |- ML.south) [yshift=-0.8cm]
        {A maximizer of ironed virtual welfare, with Myerson payments, is optimal (Thm.~\ref{thm:rev-myerson})};

    \draw[arrow] (TL.south) -- (ML.north);
    \draw[arrow] (TC.south) -- (ML.north);
    \draw[arrow] (TC.south) -- (MR.north);
    \draw[arrow] (TR.south) -- (MR.north);

    \draw[arrow] (ML.south) -- ($(Bottom.north west)!0.22!(Bottom.north east)$);
    \draw[arrow] (MLL.south) -- (Bottom.north);
    \draw[arrow] (MR.south) -- ($(Bottom.north west)!0.78!(Bottom.north east)$);

\end{tikzpicture}
}
\captionsetup{font=small}
\caption{Flowchart showing the structure of the proof of Theorem~\ref{thm:rev-myerson} for the revenue objective, with $\varphi$ denoting the virtual value for revenue. The green boxes with thick borders are new results for this specific setting, while the white boxes with thin borders follow from Myerson's proofs. The box with bold text highlights the centerpiece technical result.}%
$^*$ These results are generalized in the appendix to any revenue-aligned objective, not just revenue.
    \label{fig:rev-aligned-no-overbidding}
\end{figure}

\paragraph{Myerson payments are an upper bound.}
In truthful auctions in unconstrained settings,
buyers have no incentive to either underbid or overbid.
When the buyers are constrained,
this is no longer necessarily true:
buyers might have an \emph{incentive} to deviate
but be unable to do so due to their constraints (as in the examples from Section~\ref{sec:intro}).
However, monotone buyer constraints \emph{always allow underbidding} (this is the meaning of monotonicity for constraints).
Therefore, given an auction that is truthful under monotone buyer constraints,
we know that buyers have no incentive to underbid,
and this suffices to establish a \emph{one-sided envelope theorem},
and prove that Myerson payments are an upper bound on the payments of any truthful auction.%
\footnote{A similar argument is made independently in~\cite{krahmer2025unidirectional}.}

Before proceeding, let us dispense with a na\"ive proof idea
for showing the optimality of \revMaximizer.
It is tricky to directly compare an arbitrary truthful auction $\A$  to  \revMaximizer, because $\A$ can have
both an arbitrary allocation rule
and an arbitrary payment rule
(unlike the unconstrained setting,
where
Myerson's lemma asserts that all truthful auctions are monotone and have Myerson payments).
However, having now shown that Myerson payments are an \emph{upper bound} on the payment for any allocation rule,
could we perhaps argue that Myerson payments are \emph{optimal} for any allocation rule,
simplifying the comparison to \revMaximizer (which has Myerson payments)?
After all,
with a revenue-aligned objective,
it is in the seller's interest to extract the maximum possible payment from the buyers (while still extracting enough welfare),
so why not simply increase the payments until they match the upper bound?
The answer, of course, is that increasing payments can twist the buyers' incentives,
causing them to bid differently and potentially decreasing
the seller's expected payoff despite the improved per-outcome revenue.
In fact,
this is not solely a failure of a na\"ive proof idea:
there exist
allocation rules
that simply cannot be implemented using Myerson payments (see Example~\ref{ex:binding-payment-fail}).
Thus, when we wish to argue that a given auction $\A$ is no better than \revMaximizer,
we must somehow wrestle with the allocation rule \emph{and} the payment rule of $\A$ at the same time.
This is done in the following theorem,
which simultaneously establishes another key property, \emph{monotonicity},
and is the main technical contribution of this paper:

\begin{theorem}[Simplified, see Theorem~\ref{thm:rev-aligned-bid-constraint-mono-wlog} for the full statement]\label{thm:informal-truthification}
Let $\A$ be a truthful auction with revenue-aligned payoff and UM-constrained buyers.
Then there exists a truthful and monotone auction $\A'$ with Myerson payments
 and payoff at least as high as $\A$.%
 \footnote{In the full version of this theorem (Theorem~\ref{thm:rev-aligned-bid-constraint-mono-wlog})
 we do not assume truthfulness,
 and we also argue a strict improvement in payoff in some cases, which is necessary for the full version of Theorem~\ref{thm:rev-myerson}.}
\end{theorem}

\begin{proof}[Proof sketch]
Given an arbitrary truthful auction $\A$, we carefully transform $\A$ into a  monotone auction $\A'$ that has payoff at least as high as $\A$, and also has Myerson payments.
We iterate over buyers $i = 1,\ldots,n$
and at each step we modify the allocation and payment rules (for all buyers) so buyer $i$'s allocation becomes monotone with Myerson payments,
while keeping the incentives and outcome distribution of all buyers $j \neq i$ the same, to preserve the seller's payoff from these buyers.
At the end we obtain an auction $\A'$ with the desired properties.
(In the full proof, we start with an arbitrary auction that may not be truthful, and in each iteration we also make buyer $i$ truthful.)

The key technical challenge in modifying $\A$ is that in each step $i$,
we cannot modify \emph{only} the allocation rule $x_i(\cdot)$ of buyer $i$,
as this risks creating infeasible allocations or hurting the seller's payoff.
As we demonstrated in the toy example of Section~\ref{sec:toy-example},
to preserve feasibility,
we must \emph{jointly} rearrange the vector of allocations of all the buyers together,
so that each allocation that is realized in the new auction could also have been realized in the original auction, perhaps for a different bid vector.
However, at the same time,
we want to keep the incentives and outcome distributions of all buyers $j \neq i$ the same,
so as not to hurt the seller's payoff,
and also so as not to violate monotonicity and truthfulness already established for buyers $j < i$.
This requires delicate handling:
we appeal to measure-theoretic rearrangement tools to
map the space of buyer $i$'s bids
onto itself
in a way that results in an allocation that is nondecreasing in player $i$'s bid,
while simultaneously preserving the ex-interim distribution of outcomes for each other player.
This is the \emph{measure-preserving rearrangement} from Section~\ref{sec:toy-example}:
as with the columns in Fig.~\ref{fig:rearrange}, the point is to preserve the relevant outcome distributions while rearranging the bids on which outcomes are realized.

Formally, we start with a truthful auction $\A^{(0)}$,
and construct a sequence of auctions
$\A^{(1)},\ldots,\A^{(n)}$
that are all truthful, and such that
in each $\A^{(i)}$,
buyers $1,\ldots,i$
have monotone interim allocations and Myerson payments,
and the final auction $\A^{(n)}$ is truthful, monotone and has Myerson payments.
We ensure
that every buyer $j\neq i$ has the same interim outcome distribution in $\A^{(i-1)}$ and in $\A^{(i)}$,
and that buyer $i$ yields at least as much payoff in $\A^{(i)}$ as it does in $\A^{(i-1)}$.
Each $\A^{(i)}$ is constructed from $\A^{(i-1)}$ using a
process which we now describe.

Given an interim allocation function $y_i(v_i)=\expect{x}_i^{(i-1)}(v_i)$, we define the \emph{nondecreasing rearrangement} of $y_i$ to be a nondecreasing function $y_i^\uparrow:[\bottomV_i,\topV_i]\to[0,m]$ that is equimeasurable to $y_i$ when $v_i\sim F_i$; this nondecreasing function is guaranteed to exist~\cite{LiebLoss2001}. In other words, $y_i^\uparrow$ is a nondecreasing function that has the same distribution as $y_i$.
Our aim is to change the allocation and payment rules of buyer $i$ from $\A^{(i-1)}$, so that in $\A^{(i)}$ the interim allocation of buyer $i$ is $y_i^\uparrow$ instead of $y_i$, which is indeed nondecreasing. The way we do this is by ``relabeling'' the bids of buyer $i$: essentially, we change the bid $b_i$ of buyer $i$ in $\A^{(i)}$ to a different bid $b_i'$, and then feed $b_i'$ into the allocation and payment rules of $\A^{(i-1)}$.
Since we need to preserve the distribution of outcomes for all buyers $j\neq i$, we cannot just relabel the bids of buyer $i$ in an arbitrary way, but rather we need to find a relabeling that preserves the distribution of outcomes for all buyers. Hence, we need to make sure that the relabeled bids are distributed like truthful bids.
Indeed, using measure-theoretic techniques, we prove that there exists a \emph{randomized relabeling} of $v_i$ that essentially implements the nondecreasing rearrangement of $y_i$ on the interim allocation function of buyer $i$, while preserving the bid distribution of buyer $i$.

\begin{corollary*}[Informal, see Corollary~\ref{cor:random-relabeling} for the formal statement]
Let $\A$ be a truthful auction, and fix buyer $i$.
Let \(y_i^\uparrow\) denote the nondecreasing rearrangement of $\expect{x}_i(v_i)$, which is
buyer $i$'s interim allocation in $\A$.
Then there exists a map
$\sigma_i:[\bottomV_i,\topV_i]\times[0,1]\to[\bottomV_i,\topV_i]$
such that
\begin{enumerate}
    \item
    $\sigma_i(v_i,r)\sim F_i$
     (when $v_i\sim F_i$ and $r\sim\Uni[0,1]$), and
    \item
    $\expect{x}_i(\sigma_i(v_i))=y_i^\uparrow(v_i).$
\end{enumerate}

\end{corollary*}

Here, we think of the second input $r \in [0,1]$ as pure randomness.
We remark that in general there does not exist a deterministic relabeling with the properties that we need;
however, for the sake of simplicity, in the rest of the proof sketch we use informal notation and assume that the relabeling is deterministic, so $\sigma_i$ is defined from $[\bottomV_i,\topV_i]$ onto itself.

Now that we have the map $\sigma_i$ that implements the nondecreasing rearrangement of buyer $i$'s interim allocation, we can define the new auction $\A^{(i)}$ by rearranging the allocation and payment rules of $\A^{(i-1)}$ according to $\sigma_i$:
\begin{align*}
&x_j^{(i)}(\vec{b},r)=x_j^{(i-1)}((\sigma_i(b_i),\vec{b}_{-i}),r)\qquad\text{for every buyer $j\in[n]$, and}\\
&p_j^{(i)}(\vec{b},r)=p_j^{(i-1)}((\sigma_i(b_i),\vec{b}_{-i}),r)\qquad\text{for every buyer $j\neq i$}.
\end{align*}
It remains to define the payments of buyer $i$ in $\A^{(i)}$, which we will do soon enough.
First, one can observe that
if all buyers bid truthfully in $\A^{(i)}$, in which case $\sigma_i(v_i)$ is distributed like $v_i$,
then
for every $j\neq i$,
the interim outcome of buyer $j$ in $\A^{(i)}$ is distributed exactly like their outcome in $\A^{(i-1)}$. Thus if we make sure buyer $i$ is truthful in $\A^{(i)}$, then the payoff of the seller from all buyers $j\neq i$ is equal in $\A^{(i)}$ and in $\A^{(i-1)}$.

Notice that we have not yet defined buyer $i$'s payment, $p_i^{(i)}(\bv)$, and this is for a good reason:
we must make sure that it will be a best response for buyer $i$ to bid truthfully in $\A^{(i)}$, so we cannot simply use the ``relabeled payments'' from $\A^{(i-1)}$ as we did with the allocation.
For this purpose, we change the payments of buyer $i$ so that their interim payments are equal to their interim Myerson payments.
This again must be done carefully, because we must ensure that the outcome for buyer $i$ does not violate their constraint
when bidding truthfully in $\A^{(i)}$.

If we just use a deterministic payment rule that implements the Myerson payments, the buyer constraint may be violated, because although we are guaranteed that the interim allocation of buyer $i$ in $\A^{(i)}$ is nondecreasing, we have no pointwise guarantee on outcomes. That is, for some realizations of the auction, buyer $i$ may have negative utility if the payment is deterministically equal to the interim Myerson payment. To overcome this, in the full proof we design a pointwise payment rule that implements the interim Myerson payments, but guarantees that the utility of buyer $i$ is nonnegative for every realization of the auction when they bid truthfully.
Since buyer constraints cannot explicitly disallow truthful bidding (see Section~\ref{sec:constrained-buyers}),
this guarantees that the constraint is satisfied.

Once we have interim Myerson payments, it is easy to show, by analysis similar to that of Myerson, that it is a best response for buyer $i$ to bid truthfully in $\A^{(i)}$.
However, it is not clear that this change does not reduce the seller's payoff.
We prove that in fact it does not:
we consider the revenue term and the welfare term separately,
and prove that for each term,
the contribution from buyer $i$ does not decrease in $\A^{(i)}$ compared to $\A^{(i-1)}$.

Let us start with the revenue term.
Instead of directly proving that modifying buyer $i$'s allocation and payment does not reduce
the expected payments of buyer $i$,
we
compare the revenue of all the \emph{single-buyer} auctions (including non-truthful ones, in the full proof)
in which buyer $i$'s allocation is equimeasurable with that of
buyer $i$'s allocation in $\A^{(i-1)}$  (essentially treating all buyers $j\neq i$ as part of the mechanism instead of as buyers).
This comparison includes auction $\A^{(i)}$, because our rearrangement is measure-preserving (specifically, of the allocation of buyer $i$).
We show that among all such auctions,
the auction where the allocation is nondecreasing and payments are Myerson payments (namely, $\A^{(i)}$) yields optimal revenue,
and conclude that in particular, $\A^{(i)}$'s
revenue
is no lower than $\A^{(i-1)}$.
The challenge in proving this is that the auctions we must compare can have non-monotone allocations and
can use all sorts of odd payment rules to incentivize the buyer to bid truthfully.
However, we are able to characterize these ``truthful payments'' by showing the following upper bound, which notably \emph{cannot be concluded} from the upper bound of Lemma~\ref{lem:rev-aligned-payment-bound}:

\begin{claim}[Informal, combination of Claim~\ref{clm:tau-bound} and~Claim~\ref{clm:nondecreasing-payment-threshold-int}]\label{clm:informal-tau-bound}
    Let $\A$ be a truthful auction with interim allocation and payment rules $\expect{x}_i$ and $\expect{p}_i$.
     Define $\mathcal{V}_t=\set{v_i\in[\bottomV_i,\topV_i]:\ \expect{x}_i(v_i)\geq t}$
    and $\tau(t)=\inf\mathcal{V}_t$.
    Then for every $v_i$,
    $\expect{p}_i(v_i)\leq \int_{0}^{\expect{x}_i(v_i)}\tau(t)dt$.
    In particular, if $\expect{x}_i$ is nondecreasing, then the upper bound coincides with the Myerson payment.
\end{claim}

Both auctions we wish to compare,
$\A^{(i-1)}$ and $ \A^{(i)}$, are subject to the upper bound of Claim~\ref{clm:informal-tau-bound}.
Moreover, in $\A^{(i)}$ the payments of buyer $i$ attain the upper bound, because they are Myerson payments and $\A^{(i)}$ has a monotone allocation rule for buyer $i$.
To show that $\A^{(i)}$ has no lower payoff than $\A^{(i-1)}$,
we compare the expectation over $v_i\sim F_i$ of the RHS of the upper bound for the two auctions,
and show that it is at least as large for $\A^{(i)}$ as for $\A^{(i-1)}$.
This allows us to conclude that $\E_{v_i}[\expect{p}_i^{(i-1)}(v_i)] \leq \E_{v_i}[\expect{p}^{(i)}_i(v_i)]$.

It remains to compare the welfare term
that buyer $i$ yields for the seller
in the two auctions.
We do this using the Hardy-Littlewood inequality~\cite{LiebLoss2001}, which essentially states that the integral of the product of two functions is maximized when the two functions are similarly ordered (i.e., both nondecreasing or both nonincreasing), assuming they are equimeasurable.
Indeed in our case, the expected welfare term is an integral of the product of $\payoffValue(\cdot)$ and $\expect{x}_i(\cdot)$, and both $\payoffValue(\cdot)$ and $\expect{x}_i^{(i)}(\cdot)$ are nondecreasing, and $\expect{x}_i^{(i)}(\cdot)$ and $\expect{x}_i^{(i-1)}(\cdot)$ are equimeasurable. Thus we can apply the Hardy-Littlewood inequality to conclude that the expected welfare term is no smaller in $\A^{(i)}$ than in $\A^{(i-1)}$.
Together with the revenue
comparison above, this shows that buyer $i$'s contribution to the seller's payoff does not decrease.

To summarize, we have shown that $\A^{(i)}$ is truthful, has payoff at least as high  as $\A^{(i-1)}$ for \emph{all}  buyers, is monotone for buyer $i$, has Myerson payments for buyer $i$, and the outcomes of all buyers $j\neq i$ are unchanged. By iterating over all buyers, we obtain an auction $\A^{(n)}$ that is truthful, monotone, has Myerson payments, and has at least as high payoff as the original auction $\A^{(0)}$. This concludes the proof sketch.
\end{proof}

Theorem~\ref{thm:informal-truthification} essentially shows that although buyer constraints change the space of implementable allocation rules,
for revenue-aligned objectives,
when searching for an \emph{optimal} auction
we need only consider the same design space considered by Myerson:
truthful auctions with monotone allocation rules with Myerson payments.
From here, we prove the optimality of \revMaximizer
in a manner that is similar to Myerson's original proof.

\paragraph{Discussion.}
Before moving on to consumer-aligned objectives,
let us briefly discuss why the results of the current section rely on the fact that the buyer constraints are \emph{monotone} and \emph{uncapped} (and thus, where new ideas may be required to extend our result to other families of constraints).
\emph{Monotonicity} is relatively straightforward:
we used it to derive a one-sided version of the envelope theorem,
relying on the fact that underbidding is always allowed.
The fact that the constraints are \emph{uncapped} is perhaps more fundamental:
although in general introducing buyer constraints can change the space of implementable auctions so that it becomes incomparable to the space of auctions that are implementable with unconstrained buyers,
with \emph{uncapped} constraints the design space \emph{only expands},
as any auction that is implementable with unconstrained buyers remains implementable with uncapped constraints.
In particular, Myerson-like auctions remain implementable,
as buyers can always afford to pay Myerson payments (an uncapped constraint does not prevent them from doing so).
We remark that our results extend to constraints that are not uncapped, provided that under the constraints, \revMaximizer remains truthful.

\subsection{Consumer-Aligned Objectives}\label{sec:overview-cons-aligned}

In this section we  analyze auctions where the seller's payoff function
is \emph{consumer-aligned},
i.e.,
it has the form $\payoffValue(v)x-\conPayoffPay \cdot p$,
where $\payoffValue:\R_{\ge 0}\to\R_{\ge 0}$ is nondecreasing and
$\conPayoffPay>0$.

\paragraph{A few easy cases.}
There are a few natural scenarios where buyer-constrained auctions with consumer-aligned objectives are easy to solve.
For example, under the no-overbidding constraint, the seller can simply allocate all units to the highest bidder for free. Clearly, buyers bid truthfully, and the seller's payoff is maximized. Similarly, under the ex-post IR constraint, the seller can allocate all $m$ units to the highest bidder and, with a small probability, charge the winner the full reported value $m b_{(1)}$ of those units. The probability is chosen so that the expected payment is $\varepsilon m b_{(2)}$, where $b_{(2)}$ is the second-highest bid (as in Section~\ref{sec:intro}). Shading while remaining the winner does not change the expected payment, while overbidding to win risks being charged more than one's value and is therefore prohibited. As $\varepsilon\to0$, the seller's payoff approaches the optimal benchmark.

In the remainder of the section, we turn our attention to the stagewise-IR constraint, where the questions become a lot more interesting. As before, this technical overview omits many details; the full proofs for the stagewise-IR constraint appear in Section~\ref{sec:consumer}.

\subsubsection{Stagewise Individually Rational Buyers}

Stagewise individual rationality (stagewise-IR) concerns auctions that unfold over time.
We define such auctions formally in Section~\ref{sec:stagewise-outcomes},
but for purposes of this overview,
one may think of them as proceeding in the following (simplified) manner: first, each buyer places a bid;
then the auction mechanism determines an \emph{outcome},
which consists of a sequence of allocations $x_i^1,x_i^2,\ldots$ and payments $p_i^1,p_i^2,\ldots$ for each buyer $i$.
Such auctions arise, for example, in the context of \emph{rental games}~\cite{berzack2025dynamic}.
The \emph{cumulative utility} of buyer $i$ with valuation $v_i$ at point $t$ in time in the outcome is given by $\sum_{j = 1}^t (x_i^j \cdot v_i - p_i^j)$.

Buyers subject to the \emph{stagewise-IR} constraint cannot incur debt at any point in time, i.e.\ their cumulative utility at all times must be nonnegative with probability 1.
It was already shown in~\cite{berzack2025dynamic} that a designer can leverage this constraint to extract higher consumer-aligned payoff,
by using the first timestep of an auction to \emph{screen} for buyers with high valuations. We now extend the results of~\cite{berzack2025dynamic} to the \emph{multi-buyer} setting, where feasibility constraints become more of a challenge:
for example, in the rental game of~\cite{berzack2025dynamic}
there is only one item, and we must choose which buyer will receive the item on which day.
Allocating to a given buyer on a given day means that no other buyer can have the item that day.
This is not an issue that arises in the single-buyer setting of~\cite{berzack2025dynamic}.

We restrict attention to deterministic, DSIC, ex-post-IR auctions.%
\footnote{Otherwise, the problem becomes trivial, and it is easy to find an auction that almost achieves the optimal benchmark.}
Thus, in this section by ``truthful'' we mean ``DSIC'', and by ``monotone'' we mean ``ex-post monotone''.
In terms of notation, because we consider only deterministic DSIC auctions,
the allocation $x_i(\bv)$ and the payment $p_i(\bv)$ of each buyer are now \emph{constants} (where before they were random variables).

\paragraph{Threshold auctions.}
Adopting the definition from \cite{berzack2025dynamic}, a \emph{monotone} auction with allocation rule
$\vec{x}$ is called a \emph{threshold auction}
if it consists of two stages,
a \emph{screening phase} and a \emph{main phase},
with the following prescribed behavior for each.
During the \emph{screening phase},
buyers place bids;
if the bid profile is $\vec{b}$,
the seller
offers each buyer $i$ exactly one unit at a price
of
\begin{equation*}
    p_i(\vec{b}) =
    \begin{cases}
         \inf\set{z\in[\bottomV_i,\topV_i] :
        x_i(z,\vec{b}_{-i})\ge x_i(\vec{b})},
        & \text{if  $\inf\set{z\in[\bottomV_i,\topV_i] :
        x_i(z,\vec{b}_{-i})\ge x_i(\vec{b})} > \bottomV_i$}
        \\
        0 & \text{otherwise.}
    \end{cases}
\end{equation*}
These payments are called \emph{threshold payments} (w.r.t.\ $\vec{x}$).
During the \emph{main phase}, each buyer that purchased a unit during the screening phase is given the remainder of their total allocation (i.e., $x_i(\vec{b}) - 1$ units), and no payments are charged.
Buyers that did not purchase a unit during the screening phase are not allocated units or charged payments during the main phase.

The threshold auction is designed in a way that the payment in the screening phase serves to ``screen out'' buyers with low valuations: since buyers are subject to stagewise-IR, they are unable to overbid, because this would result in negative quasilinear utility in the screening phase and violate their constraint; thus, the seller can choose to allocate a large number of units only to buyers with ``high enough'' valuations, without requiring high payments.
Threshold payments deter overbidding only if a buyer who bids exactly the threshold already receives the corresponding allocation; otherwise, a buyer whose valuation equals the threshold may profit from overbidding, as the following example shows.

\begin{example}\label{ex:threshold-tie}
Consider an auction with two units and two stagewise-IR buyers, whose allocation rule gives both units to the higher bidder, breaking ties in favor of buyer~$1$.
Fix $b_1=\tfrac12$.
Buyer~$2$ receives both units if and only if $b_2>\tfrac12$, so their threshold payment when winning is $\inf\set{z : x_2(z,\tfrac12)\geq 2}=\tfrac12$.
However, this infimum is not attained: bidding exactly $\tfrac12$ loses the tie.

Now suppose buyer~$2$'s valuation is $v_2=\tfrac12$.
Bidding truthfully, they lose and obtain utility $0$.
Bidding $b_2=\tfrac34$ instead, they receive one unit in the screening phase at price $\tfrac12$ and one more unit in the main phase for free.
The filter of this outcome is $\tfrac12\leq v_2$, so this deviation is acceptable for a stagewise-IR buyer, and it yields utility $2\cdot\tfrac12-\tfrac12=\tfrac12>0$.
Hence this auction is not DSIC.
\end{example}

\paragraph{Shifted thresholds.}
To avoid this problem, we \emph{shift} the thresholds.
Given a monotone allocation rule, let $\tau_{i,k}(\vec{b}_{-i})$ be the threshold for buyer $i$ to receive $k$ units, i.e., the lowest bid with which buyer $i$ receives at least $k$ units, given the bids $\vec{b}_{-i}$ of the other buyers.
Since the threshold depends on the other buyers' bids, it may be determined by them directly.
For example, in Example~\ref{ex:threshold-tie}, where the higher bidder receives both units, buyer $i$'s threshold for both units is the other buyer's bid, $\tau_{i,2}(\vec{b}_{-i})=b_{-i}$; in particular, when the two bids are tied, each buyer bids exactly their threshold.
In the \emph{$\delta$-shifted threshold auction}, for a small $\delta>0$, buyer $i$ receives the largest number of units $k$ such that $b_i\geq\tau_{i,k}(\vec{b}_{-i})+\delta$, and pays $\tau_{i,k}(\vec{b}_{-i})+\delta$ for the first unit in the screening phase.
The shift ensures that a buyer who bids exactly $\tau_{i,k}(\vec{b}_{-i})+\delta$ receives $k$ units, so no buyer can profit from overbidding as in Example~\ref{ex:threshold-tie}.

We show that the shifted threshold auction is DSIC and ex-post IR.
The price of the shift is a small loss in payoff, which vanishes as $\delta\to0$. In fact, we show that
 such a loss is unavoidable.

\paragraph{Restricting attention to monotone auctions.}
As we saw in Section~\ref{sec:overview-rev-aligned}, a key step in finding an optimal auction is showing that we may restrict attention to \emph{monotone} auctions.
For revenue-aligned objectives, this followed from showing that every \emph{optimal} auction is monotone.
For consumer-aligned objectives this is not enough, since the optimal payoff may not be attained by any DSIC auction (Proposition~\ref{prop:consumer-surplus-nonattainment}).
Instead, we show directly that every auction can be replaced by a monotone one, without decreasing its payoff:

\begin{theorem}\label{thm:cons-surplus-mono-wlog}
Suppose the payoff is consumer surplus and buyers are stagewise-IR.
Then every deterministic, DSIC, ex-post-IR auction can be made monotone without decreasing its payoff: for every such auction $\A$ there is a deterministic, DSIC, ex-post-IR, monotone auction $\A'$ with $\E[\A']\geq\E[\A]$.
\end{theorem}

We also show that this is true for any consumer-aligned objective, as long as only one buyer in total can be allocated a positive amount. This is a natural setting, for example, in the rental game where multiple customers arrive each day, but only one customer can rent the asset at a time.

\begin{theorem}\label{thm:cons-aligned-mono-wlog}
Suppose the payoff is consumer-aligned, buyers are stagewise-IR, and the environment is single-winner.
Then every deterministic, DSIC, ex-post-IR auction can be made monotone without decreasing its payoff: for every such auction $\A$ there is a deterministic, DSIC, ex-post-IR, monotone auction $\A'$ with $\E[\A']\geq\E[\A]$.
\end{theorem}

In both proofs, we construct a monotone allocation rule whose payoff with threshold payments is at least that of the original auction, and then charge shifted threshold payments, as above, to obtain a DSIC auction.
For consumer surplus (Theorem~\ref{thm:cons-surplus-mono-wlog}), we give each valuation the smallest allocation of any higher valuation. This only decreases allocations, so it preserves allocation-feasibility, and we show that it does not decrease any buyer's utility; this builds on an idea of~\cite{berzack2025dynamic} for a single buyer.
For general consumer-aligned objectives (Theorem~\ref{thm:cons-aligned-mono-wlog}), we instead rearrange each buyer's allocation into a nondecreasing one, similarly to Section~\ref{sec:overview-rev-aligned}.
With stagewise-IR buyers, once a buyer starts winning they keep winning at every higher valuation, so we can rearrange the allocation within the buyer's winning set; in a single-winner environment, this does not affect the other buyers.
The main difference from Section~\ref{sec:overview-rev-aligned} is that here the monotone auction we construct must be DSIC and not only BIC, which requires a more careful payment design.

\paragraph{Finding an almost optimal auction.}
Having established that in these settings we can restrict attention to monotone auctions, we would ideally like to proceed in a way that closely mirrors the structure of Myerson's proof,
although this turns out to be possible only in some cases, not in general. For these results we assume that $\bottomV_i=0$ for every buyer $i$. The monotonicity results in Theorems~\ref{thm:cons-surplus-mono-wlog} and \ref{thm:cons-aligned-mono-wlog} do not require this assumption.

We adopt a generalization from~\cite{berzack2025dynamic} of virtual values designed specifically for threshold auctions, and show that the expected payoff of any monotone DSIC auction is upper bounded by the expected \emph{ironed virtual welfare} of its allocation.
This naturally suggests using an allocation that maximizes the ironed virtual welfare pointwise, which we call the \emph{IVW allocation} $\vec{x}^*$ (Definition~\ref{def:cons-aligned-ironed-maximizer}).
As explained above, combining $\vec{x}^*$ with threshold payments need not yield a DSIC auction, so we instead use the $\delta$-shifted threshold auction of $\vec{x}^*$, which we call the \consMaximizer (Definition~\ref{def:cons-shifted-ivw}).
It is DSIC and ex-post IR for every $\delta>0$, and the question is whether its expected payoff approaches the ironed virtual welfare bound as $\delta\to0$.
Unfortunately, this is not the case in general: the bound is tight only when each buyer's allocation jumps once, directly from $0$ to $m$ units, so the gap need not vanish when $\vec{x}^*$ has intermediate allocation levels.

Fortunately, there \emph{are} cases where this approach succeeds. This happens when the cost function is \emph{concave}: then $\vec{x}^*$ either allocates no units or allocates all $m$ units to a single buyer.
We show that in this case, the expected payoff of \consMaximizer converges to the ironed virtual welfare bound as $\delta\to0$.

\begin{theorem}\label{thm:cons-surplus-optimal-auc}
    Suppose the payoff is consumer surplus, buyers are stagewise-IR, and the cost function is concave.
    Then for every $\varepsilon>0$ there is $\delta>0$ such that \consMaximizer is $\varepsilon$-optimal among deterministic, DSIC, ex-post-IR auctions.
\end{theorem}

\begin{theorem}\label{thm:cons-aligned-optimal-auc}
    Suppose the payoff is consumer-aligned, buyers are stagewise-IR, the environment is single-winner, and the cost function is concave.
    Then for every $\varepsilon>0$ there is $\delta>0$ such that \consMaximizer is $\varepsilon$-optimal among deterministic, DSIC, ex-post-IR auctions.
\end{theorem}

This is the best one can hope for: Proposition~\ref{prop:consumer-surplus-nonattainment} shows that, even for consumer surplus with zero cost, the supremum may not be attained by any deterministic, DSIC, ex-post-IR auction.
A related nonattainment phenomenon arises in the menus of randomized contracts of~\cite{castiglioni2022designing}, where the principal's optimal expected utility may be a supremum that no menu attains.
In both settings, incentive-compatible mechanisms approach the supremum, but their limit fails to preserve incentives.
In our setting, this failure stems from the discontinuity of acceptability at screening thresholds (as in Example~\ref{ex:threshold-tie}); in theirs, it stems from unbounded payments on vanishing-probability events.

Interestingly, \consMaximizer is DSIC (in the above settings) but \emph{not U-DSIC}. This is not a coincidence:

\begin{remark}\label{rem:cons-not-one-shot}
For consumer-aligned payoffs with stagewise-IR buyers, one-shot auctions are not always optimal; the examples of~\cite[Corollary 7.3]{berzack2025dynamic} carry over directly to our setting.
In other words, the seller can sometimes benefit from a DSIC auction that is not U-DSIC.
\end{remark}

\section{Revenue-Aligned Objectives}\label{sec:revenue}

We use the revenue-aligned generalized virtual value $\varphi_{i,\Pi}$ and its ironing $\irn_{i,\Pi}$ from Definition~\ref{def:rev-aligned-virtual-values}.
This section turns the overview from Section~\ref{sec:overview-rev-aligned} into the full proof.
The argument has three parts: first we define the candidate ironed-virtual-welfare maximizer; then we prove payment upper bounds that replace the usual envelope theorem; finally we use measure-preserving rearrangements to reduce an arbitrary auction to monotone truthful auctions with improved payoff, and based on this, prove the main theorem (Theorem~\ref{thm:rev-myerson}).

\subsection{The Ironed-Virtual-Welfare Maximizer}

We begin with the auction that will ultimately be optimal.
It uses the familiar Myerson payment identity, but with the revenue-aligned ironed virtual values in the allocation objective.

\begin{definition}\label{def:rev-aligned-ironed-maximizer}
    Let $\Pi$ be a revenue-aligned payoff with its ironed virtual value function $\irn_{i,\Pi}$.
Fix a total order $\prec$ over $\Omega$.
Then \emph{\revMaximizer} is the deterministic one-shot auction
$\A^*$
defined by:
\begin{align}\label{eq:rev-aligned-argmax-def}
&\vec{x}^{*}(\bv) := \max_{\prec} \argmax_{\vec{y}\in\Omega}
\brackets{ \sum_{i=1}^n \irn_{i,\Pi}(v_i)y_i - c\left(\sum_{i=1}^n y_i\right) }, \\
&p_i^{*}(\bv) := v_i x_i^{*}(\bv) - \int_{\bottomV_i}^{v_i} x_i^{*}(z,\bv_{-i})dz,
\label{eq:IVW-max-payment}
\qquad
\text{for each buyer $i$}
.
\end{align}
\end{definition}

\begin{lemma}\label{lem:rev-aligned-optimal-among-tight}
The \revMaximizer from Definition~\ref{def:rev-aligned-ironed-maximizer} is U-DSIC, IR, allocation-feasible and monotone.
In addition, it satisfies
\begin{align*}
    \E_{\bv}\sqbr{\sum_{i=1}^n \payoffValue(v_i)x_i(\bv)+\payoffPay p_i(\bv)}=\E_{\bv}\sqbr{\sum_{i=1}^n \irn_{i,\Pi}(v_i)x_i(\bv)}.
\end{align*}
\end{lemma}

The next theorem states the revenue-aligned conclusion.
The rest of the section is devoted to proving it: Lemma~\ref{lem:rev-aligned-payment-bound} and Corollary~\ref{cor:pay-bound-by-varphi} give the payment side of the comparison, while Theorem~\ref{thm:rev-aligned-bid-constraint-mono-wlog} gives the monotonicity and truthfulness reduction.

\begin{theorem}\label{thm:rev-myerson}
    Assume that each buyer is subject to a UM-constraint $\pred_i$, and the objective $\Pi=(\payoffValue,\payoffPay)$ is revenue-aligned.
    Then:

        \begin{enumerate}
        \item There exists an optimal auction that is one-shot and U-DSIC.
        \item If $\payoffValue(\cdot)$ is strictly increasing or $\payoffPay>0$,
        then every optimal auction is monotone-\aee
        \item If $\payoffPay>0$ then in every optimal BIC auction,%
        \footnote{In fact, this claim holds for \emph{every optimal auction}, even non-truthful ones,
        provided the auction has a BNE $\vec{\bid}$ that achieves the revenue of the auction (i.e., achieves the supremum). In this case $\bid_i(v_i)$ replaces $v_i$ inside the allocation and payment function, and similarly $\bid_i(z)$ replaces the $z$ inside the integral.
        }
        the payment rule satisfies, for $F_i$-\aee $v_i$,
        \begin{align}\label{eq:unique-pay}
            \expect{p}_i(v_i)=\expect{x}_i(v_i)v_i-\int_{\bottomV_i}^{v_i} \expect{x}_i(z)\,dz.
        \end{align}
               (If $\payoffPay = 0$ then there \emph{exists}
        an optimal BIC auction with these payments,
        but there may also exist optimal BIC auctions with other payments.)

    \end{enumerate}
\end{theorem}

\subsection{Payment Bounds}

The following is a one-sided envelope result that bounds the payment of the buyer in a BNE of an auction with buyers subject to monotone constraints. This result is known in the special case of a single buyer with the no-overbidding constraint (see~\cite{krahmer2025unidirectional}); the proof below provides an alternative and more general argument that applies to any number of buyers and any monotone constraint.

\begin{lemma}\label{lem:rev-aligned-payment-bound}
Assume each buyer is subject to a monotone constraint $\pred_i$. Let $\A$ be an auction that admits a Bayesian-Nash equilibrium $\vbid$. Then for every $i$ and type $v_i$,
\begin{align}\label{eq:rev-aligned-payment-bound}
        \expect{p}_i(\bid_i(v_i);\vbid)\leq\expect{x}_i(\bid_i(v_i);\vbid)v_i-\int_{\bottomV_i}^{v_i}\expect{x}_i(\bid_i(z);\vbid)\,dz.
    \end{align}
\end{lemma}
\begin{proof}
    Fix $i$. Let $\vbid$ be a BNE.

    To simplify notation, let $U_i(v_i)\coloneqq \expect{\util}_i(\bid_i(v_i); v_i)$ denote the equilibrium expected utility of buyer $i$ with valuation $v_i$.

    Consider a valuation $v_i$ and a slightly lower deviation $v_i-\delta\geq\bottomV_i$. Since $\vbid$ is a BNE and $\A$ is IR, $\bid_i(v_i-\delta)$ is acceptable for type $v_i-\delta$.
    Therefore,
    \begin{align*}
        U_i(v_i)-U_i(v_i-\delta)
        &=
        U_i(v_i)-\expect{\unconstrained}_i(\bid_i(v_i-\delta);v_i-\delta)\\
        &=U_i(v_i)-\expect{x}_i(\bid_i(v_i-\delta))(v_i-\delta)+\expect{p}_i(\bid_i(v_i-\delta))\\
        &=U_i(v_i)-\expect{\unconstrained}_i(\bid_i(v_i-\delta);v_i)+\delta\expect{x}_i(\bid_i(v_i-\delta)).
    \end{align*}
    Crucially, since $\pred_i$ is monotone, $\bid_i(v_i-\delta)$ is also acceptable for type $v_i$.  Hence,
    \begin{align}\label{eq:util-diff}
        U_i(v_i)-U_i(v_i-\delta)
        &=U_i(v_i)-\expect{\util}_i(\bid_i(v_i-\delta);v_i)+\delta\expect{x}_i(\bid_i(v_i-\delta))\\
        &\geq \delta\expect{x}_i(\bid_i(v_i-\delta)),
    \end{align}
    where the inequality follows from the fact that $\vbid$ is a BNE.

    We now partition $[\bottomV_i,v_i]$ into $N$ equally-sized intervals, each of size $\Delta_N\coloneqq\frac{v_i-\bottomV_i}{N}$.
    Using a telescopic sum, and due to Eq.~\eqref{eq:util-diff},
    \begin{align*}
      U_i(v_i)\geq U_i(v_i)-U_i(\bottomV_i)
       \geq\sum_{j=1}^{N} \Delta_N\cdot \expect{x}_i(\bid_i(\bottomV_i+(j-1)\Delta_N)).
    \end{align*}
    The first inequality is due to $\vbid$ being a BNE and $\A$ being IR.

    Since $x_i(\bid_i(\cdot))$ is Riemann-integrable, so is $\expect{x}_i:[\bottomV_i,\topV_i]\to [0,m]$. Thus we take $N$ to $\infty$, and get
    \begin{align}\label{eq:rev-aligned-util-lower-bound}
        \expect{\util}_i(\bid_i(v_i);v_i)\geq\int_{\bottomV_i}^{v_i}\expect{x}_i(\bid_i(z))\,dz.
    \end{align}

    Since $\A$ is IR and $\vbid$ an equilibrium,
    $\expect{\util}_i(\bid_i(v_i);v_i)=\expect{\unconstrained}_i(\bid_i(v_i);v_i)=\expect{x}_i(\bid_i(v_i))v_i-\expect{p}_i(\bid_i(v_i))$. Combining this with~\eqref{eq:rev-aligned-util-lower-bound} completes the proof.
\end{proof}

For truthful auctions, Lemma~\ref{lem:rev-aligned-payment-bound} immediately implies the following corollary.
\begin{corollary}\label{cor:pay-bound-by-varphi}
    Let $\A$ be a BIC auction with a revenue-aligned payoff function $\Pi=(\payoffValue,\payoffPay)$, and buyers subject to monotone constraints.
    Then, for every $i$,
    \begin{align*}
        \E_{\bv,r}\sqbr{\payoffValue(v_i)x_i(\bv,r)+\payoffPay p_i(\bv,r)}\leq \E_{\bv,r}\sqbr{\varphi_{i,\Pi}(v_i)x_i(\bv,r)}.
    \end{align*}
\end{corollary}
\begin{proof}
    Since $\A$ is BIC, the strategy $\bid_j(v_j)=v_j$ for every $j$ is the BNE played. Thus, by Lemma~\ref{lem:rev-aligned-payment-bound}, for every $i$:
    \begin{align*}
        \E_{\bv_{-i},r}\sqbr{p_i(\bv,r)}\leq \E_{\bv_{-i},r}\sqbr{x_i(\bv,r)v_i-\int_{\bottomV_i}^{v_i}x_i((z,\bv_{-i}),r)dz}.
    \end{align*}
    Integrating over $v_i$ and repeating Myerson's analysis gives
    \begin{align*}
         \E_{\bv,r}\sqbr{p_i(\bv,r)}\leq\E_{\bv,r}\sqbr{\varphi_i(v_i)x_i(\bv,r)}\quad\forall i\in[n].
    \end{align*}
    Therefore,
    \begin{align*}
        \E_{\bv,r}\sqbr{\payoffValue(v_i)x_i(\bv,r)+\payoffPay p_i(\bv,r)}
        \leq\E_{\bv,r}\sqbr{\payoffValue(v_i)x_i(\bv,r)+\payoffPay \varphi_i(v_i)x_i(\bv,r)}
        =\E_{\bv,r}\sqbr{\varphi_{i,\Pi}(v_i)x_i(\bv,r)}.
    \end{align*}
\end{proof}

\subsection{Truthification and Monotonicity}

We now move on to the structural reduction.
The goal is to show that any auction can be replaced, without lowering (or even improving) revenue-aligned payoff, by one in which buyers bid truthfully, interim allocations are monotone, and payments are tight.

\paragraph{Auctions with tight payments.}
When an auction has payments that always match the upper bound of Eq.~\eqref{eq:rev-aligned-payment-bound} with equality for $\vbid$, we say that it has \emph{tight payments}  with respect to $\vbid$, and if equality holds $F_i$-almost surely, we say that it has \emph{tight payments almost everywhere}.
Note that in the case of a truthful auction, where the truthful BNE $\vbid(\bv)=\bv$ is played, tight payments with respect to $\vbid$ are exactly Myerson payments.

As described in Section~\ref{sec:overview-rev-aligned}, tight payments are not necessarily optimal or even truthful, depending on the allocation rule and the payoff function. However, we show that for revenue-aligned payoffs, there always exists an optimal auction that is monotone and has tight payments.

\begin{theorem}\label{thm:rev-aligned-bid-constraint-mono-wlog}
Let $\A$ be an auction with revenue-aligned payoff and UM-constrained buyers, and suppose
$\vbid$ is the BNE played in $\A$.
Then there exists a BIC-\aee and monotone-\aee auction $\A'$
 such that $\E[\A']\ge \E[\A| \vbid]$ and the interim payments
 satisfy the Myerson payment identity~\eqref{eq:unique-pay} $F_i$-\aee for every buyer $i$.
Moreover, if $\A$ is not monotone-\aee for some buyer $i$
and either $\payoffPay>0$ or $\payoffValue(\cdot)$ is
strictly increasing, then the inequality is strict.
\end{theorem}

The proof of Theorem~\ref{thm:rev-aligned-bid-constraint-mono-wlog} uses two auxiliary payment identities.
The first provides another upper bound on the payment, which cannot be deduced from the bound in Lemma~\ref{lem:rev-aligned-payment-bound}; the second shows that, for monotone allocations, this threshold form is exactly the Myerson payment identity.

\begin{claim}\label{clm:tau-bound}
    Let $\A$ be an auction with UM-constrained
    buyers and a BNE-$\aee^{(i)}$ $\vbid$.

    Define $\mathcal{V}_t=\set{v_i\in[\bottomV_i,\topV_i]:\ \expect{x}_i(\beta_i(v_i);\vbid)\geq t}$
    and $\tau(t)=\inf\mathcal{V}_t$,
    where we set
     $\inf\emptyset=\bottomV_i$.

     Then for every $v_i$,
     \begin{align}
        \expect{p}_i(\bid_i(v_i);\vbid)\leq \int_{0}^{\expect{x}_i(\bid_i(v_i);\vbid)}\tau(t)dt.
        \label{eq:payment_upper}
    \end{align}
\end{claim}
\begin{proof}
In this proof we abuse notation and
denote
\begin{align*}
    x(v):=\expect{x}_i(\bid_i(v);\vbid)
    \quad\text{and}\quad
    p(v):=\expect{p}_i(\bid_i(v);\vbid).
\end{align*}
Additionally, we denote the interim utility for buyer $i$
with value $v$, when all buyers bid by $\vbid$ as
\begin{align*}
    U(v):=\expect{\util}_i(\bid_i(v);v,\vbid).
\end{align*}
Using these notations, we must prove that for every $v\in [\bottomV_i,\topV_i]$,
$p(v)\leq\int_0^{x(v)}\tau(t)dt$.

Since $\A$ is IR and $\vbid$ is a BNE-\aee,
for every $v\in [\bottomV_i,\topV_i]$ we have
\begin{align*}
    U(v)=\unconstrained(v;v)=x(v)v-p(v).
\end{align*}

Together with the fact that the buyers' constraints are
monotone and that buyer $i$ with type $v$ bids $\bid_i(v)$,
we get that for every $z\leq v$,
\begin{align*}
U(v)\geq\util(z;v)=\unconstrained(z;v)=x(z)v-p(z)=\unconstrained(z;z)+(v-z)x(z)\quad\forall z\leq v.
\end{align*}
Thus, for every $z\leq v$,
\begin{align}\label{eq:clm-2-proof-1}
    U(v)\geq U(z)+(v-z)x(z).
\end{align}

Now fix $v\in [\bottomV_i,\topV_i]$.
If $x(v)=0$ then by IR and since $\vbid$ is a BNE-\aee,
$p(v)=0\leq\int_0^0\tau(t)dt$ and we're done, so we assume $x(v)>0$.

Select a number $N\in\mathbb{N}_{> 0}$. Define $K=\lfloor Nx(v) \rfloor$ and levels
\begin{align*}
    t_k=\frac{k}{N}\quad(k=0,1,\dots,K).
\end{align*}

For each $k\geq 1$, by definition of $\tau(t_k)$, we can pick a type $w_k$ such that
\begin{align}\label{eq:clm-2-proof-2}
    w_k\leq\tau(t_k)+\frac{1}{N}\quad\text{and}\quad x(w_k)\geq t_k.
\end{align}
We select the sequence such that $w_k\geq w_{k-1}$ for every $k\geq1$. The specific selection of $w_k$ is as follows:
\begin{enumerate}
    \item First select some $w_K\in \mathcal{V}_{t_K}$ such that $w_K\leq\min\set{v,\tau(t_K)+\frac{1}{N}}$.

    This is possible since $v\geq \tau(t_K)$.
    \item Next, for $k=K-1,\dots,1$ (in decreasing order):
    \begin{itemize}
        \item If $w_{k+1}\leq\tau(t_k)+\frac{1}{N}$, set $w_k\coloneqq w_{k+1}$ (which is acceptable since $\mathcal{V}_{t_{k+1}}\subseteq \mathcal{V}_{t_k}$).
        \item Else, pick any $w_{k}\in \mathcal{V}_{t_{k}}$ such that $w_{k}\leq\tau(t_k)+\frac{1}{N}$.
    \end{itemize}
\end{enumerate}
We also set $w_0=\bottomV_i$. By IR, $U(w_0)\geq0$.

Importantly, this selection process ensures that $w_k\geq w_{k-1}$ for every $k=1,\dots,K$.

Since $w_k$ is nondecreasing, we can apply~\eqref{eq:clm-2-proof-1} for each $k$ to get
\begin{align*}
    U(w_k)\geq U(w_{k-1})+(w_k-w_{k-1})x(w_{k-1})\geq U(w_{k-1})+(w_k-w_{k-1})t_{k-1}.
\end{align*}
The last transition is due to~\eqref{eq:clm-2-proof-2}.

Summing over all $k$'s we get
\begin{align*}
    U(w_K)\geq\sum_{k=1}^K (w_k-w_{k-1})t_{k-1}.
\end{align*}

We now add another step to the inequality by applying~\eqref{eq:clm-2-proof-1} on $v\geq w_K$ to get

\begin{align}\label{eq:clm-2-proof-3}
    U(v)\geq U(w_K)+(v-w_K)x(w_K)
    \geq \sum_{k=1}^K (w_k-w_{k-1})t_{k-1}+(v-w_K)t_K.
\end{align}

Because $p(v)=x(v)v-U(v)$,~\eqref{eq:clm-2-proof-3} yields
\begin{align}\label{eq:clm-2-proof-4}
    p(v)\leq v(x(v)-t_K)+w_Kt_K-\sum_{k=1}^K (w_k-w_{k-1})t_{k-1}.
\end{align}

When $N\longrightarrow\infty$,
the term $t_K=\frac{\lfloor Nx(v)\rfloor}{N}$ tends to $x(v)$,
 so the first term in the above inequality is
  $v(x(v)-t_K)=v(x(v)-x(v))+O\left(\frac{1}{N}\right)$ and tends to $0$.
  For the second term, plugging in $t_k=\frac{k}{N}$, we have
\begin{align*}
    w_Kt_K-\sum_{k=1}^K (w_k-w_{k-1})t_{k-1}
    &=\frac{1}{N}\brackets{Kw_K-\sum_{k=1}^K (w_k-w_{k-1})(k-1)}\\
    &=\frac{1}{N}\brackets{ Kw_K-\sum_{k=1}^K w_k (k-1)+\sum_{k=0}^{K-1} w_kk}\\
    &=\frac{1}{N}\brackets{ w_K-\sum_{k=1}^{K-1}w_{k}(k-1) +\sum_{k=1}^{K-1} w_kk}\\
    &=\frac{1}{N}\sum_{k=1}^{K}w_{k}.
\end{align*}
Plugging these into~\eqref{eq:clm-2-proof-4},
\begin{align*}
    p(v)\leq O\brackets{\frac{1}{N}}+\frac{1}{N}\sum_{k=1}^{K}w_{k}.
\end{align*}

Using $w_k\leq\tau(t_k)+\frac{1}{N}$ from~\eqref{eq:clm-2-proof-2},
\begin{align*}
      p(v)\leq O\brackets{\frac{1}{N}}+\frac{1}{N}\sum_{k=1}^{K}\tau(t_k).
\end{align*}

Since $\tau(\cdot)$ is nondecreasing, we have a Riemann sum for which we can take $N$ to $\infty$, getting
\begin{align*}
    p(v)\leq\int_{0}^{x(v)}\tau(t)dt.
\end{align*}
\end{proof}

We then go on to show that if an allocation rule is \emph{monotone}, the two payment bounds coincide:

\begin{claim}\label{clm:nondecreasing-payment-threshold-int}
    Let $y:[\bottomV_i,\topV_i]\to [0,m]$ be a measurable
    nondecreasing function.
    Then
    setting $\mathcal{V}_t=\set{v_i\in[\bottomV_i,\topV_i]:\,y(v_i)\geq t}$
    and $\tau(t)=\inf\mathcal{V}_t$, where we set $\inf\emptyset=\bottomV_i$,
    it holds that
    \begin{align}\label{eq:sketch-layer-cake}
        y(v_i)v_i
        -\int_{\bottomV_i}^{v_i}y(z)dz
        =\int_{0}^{y(v_i)}\tau(t)dt.
    \end{align}
\end{claim}
\begin{proof}
By the layer cake representation, since we assume
$y(v)$ is measurable and nondecreasing,

\begin{align*}
    \int_{\bottomV_i}^{v_i} y(z)dz
    &=\int_0^{y(v_i)}
        \text{length}(\set{z\in[\bottomV_i,{v_i}]:y(z)\geq t})dt
    \\&=\int_0^{y(v_i)} \text{length}([\tau(t),v_i])dt
    \\&=\int_0^{y(v_i)} {v_i}-\tau(t)dt
    =y(v_i){v_i}-\int_0^{y(v_i)}\tau(t)dt.
\end{align*}

The claim follows immediately.
\end{proof}

We are now ready to prove Theorem~\ref{thm:rev-aligned-bid-constraint-mono-wlog}.

\begin{proof}[Proof of Theorem~\ref{thm:rev-aligned-bid-constraint-mono-wlog}]
Let $\A$ be an auction with BNE-\aee $\vbid^{(0)}$.
 We show that there exists a sequence of auctions
$\A^{(0)}:=\A, \A^{(1)},\ldots,\A^{(n)}$ such that for each $i$:
\begin{enumerate}[label=(\alph*)]
    \item\label{it:bic} The strategy profile $\vbid^{(i)}$ defined as
        \begin{align*}
            \bid^{(i)}_j(v_j):=\begin{cases}
            v_j & \text{if }j\in\set{1,\dots,i},\\
            \bid^{(i-1)}_j(v_j) & \text{if }j\in\set{i+1,\dots,n}
            \end{cases}
        \end{align*}
        is a BNE-$\aee^{(i+1,\dots,n)}$ of $\A^{(i)}$.
    \item\label{it:ir-alloc-feas}  $\A^{(i)}$ is IR and allocation-feasible.

    \item\label{it:mono} Buyers $j=1,\ldots,i$ have monotone-\aee interim allocations $\expect{x}_j(\bid_j(\cdot);\vbid)$.
    \item Buyers $j=1,\ldots,i$ have Myerson payments almost everywhere.
    \item\label{it:payoff} $\E[\A^{(i)}|\vbid^{(i)}]\ge \E[\A^{(i-1)}|\vbid^{(i-1)}]$,
    with the inequality being strict if $\A^{(i-1)}$ is not monotone-\aee for buyer $i$,
    and either $\payoffPay>0$ or $\payoffValue(\cdot)$ is strictly increasing.
\end{enumerate}

Setting $\A':=\A^{(n)}$ will conclude the proof.

Throughout the proof, we
use randomness  $r\coloneqq(r_{A},r_{\ell})$ to denote two independent random
 variables distributed by $\Uni[0,1]$:
  one is used for the internal randomness of the auction;
   and the other for a randomized relabeling (defined below).
Expectations $\E_r[\cdot]$ are taken over
the joint draw of $(r_{A},r_{\ell})$.

Given $\A^{(i-1)}$: if $\A^{(i-1)}$ is given and fulfills
 conditions~\ref{it:bic}-\ref{it:payoff} for $i$, set $\A^{(i)}:=\A^{(i-1)}$;
 else, define $\A^{(i)}$ to be the truthification of
 $\A^{(i-1)}$ for buyer $i$ as defined in Definition~\ref{def:truthification}.

 We will prove by induction on $i$ that
 $\A^{(i)}$ satisfies conditions~\ref{it:bic}-\ref{it:payoff} for $i$.
 For $i=0$ these conditions are trivially satisfied, so we prove for
 $i\geq1$.

Suppose $\A^{(i-1)}$ satisfies conditions~\ref{it:bic}-\ref{it:payoff} for $i-1$ but not for $i$.
We show $\A^{(i)}$ has all the desired properties:
\begin{enumerate}
    \item \textbf{$\vbid^{(i)}$ is a BNE-$\mathbf{\aee^{(i+1,\dots,n)}}$ of $\A^{(i)}$:}
    By Lemma~\ref{lem:truthification-buyer-i-bic}, $\A^{(i)}$ is truthful-\aee for buyer $i$ w.r.t.\ $\vbid$,
    which means that $\bid^{(i)}_i(v_i)=v_i$ is indeed a best response for buyer $i$ in $\A^{(i)}$
     when the other buyers play according to $\vbid^{(i)}$.
     In Lemma~\ref{lem:rearranging-preserves-bic} we showed that for every $j\neq i$, \aee $v_j$, and every $b_j$,
     $        \E_{\bv_{-j}}\sqbr{\expect{\util}^{(i)}_j(\bid^{(i)}_{-j}(\bv_{-j}),b_j;v_j)}
        =\E_{\bv_{-j}}\sqbr{\expect{\util}^{(i-1)}_j(\bid^{(i-1)}_{-j}(\bv_{-j}),b_j;v_j)}$.
        Since $\vbid^{(i-1)}$ is a BNE-\aee of $\A^{(i-1)}$, it follows that for every $j\neq i$, \aee $v_j$ and every $b_j$,
        \begin{align*}
            \E_{\bv_{-j}}\sqbr{\expect{\util}^{(i)}_j(\bid^{(i)}_{-j}(\bv_{-j}),\bid^{(i)}_j(v_j);v_j)}
            &=\E_{\bv_{-j}}\sqbr{\expect{\util}^{(i-1)}_j(\bid^{(i-1)}_{-j}(\bv_{-j}),\bid^{(i)}_j(v_j);v_j)}
            \\
            &=\E_{\bv_{-j}}\sqbr{\expect{\util}^{(i-1)}_j(\bid^{(i-1)}_{-j}(\bv_{-j}),\bid^{(i-1)}_j(v_j);v_j)}
            \\&\geq
            \E_{\bv_{-j}}\sqbr{\expect{\util}^{(i-1)}_j(\bid^{(i-1)}_{-j}(\bv_{-j}),b_j;v_j)}
            \\&=\E_{\bv_{-j}}\sqbr{\expect{\util}^{(i)}_j(\bid^{(i)}_{-j}(\bv_{-j}),b_j;v_j)},
        \end{align*}
        meaning that $\bid^{(i)}_j$ is indeed a Bayesian best response for buyer $j$ in $\A^{(i)}$
        when the other buyers play according to $\vbid^{(i)}$.
        For $j=i+1,\dots,n$, by the induction hypothesis
        buyer $j$ bids $\bid^{(i-1)}_j(v_j)=\bid^{(i)}_j(v_j)$ pointwise for every
        valuation $v_j$, so we can follow the exact
        chain from above to get that also in $\A^{(i)}$,
        strategy $\bid_j$ is a Bayesian best response, pointwise for every valuation of buyer $j$.
        Thus, $\vbid^{(i)}$ is a BNE-$\aee^{(i+1,\dots,n)}$ of $\A^{(i)}$.
    \item \textbf{Payments are nonnegative and it is allocation-feasible:}
    This is shown in Claim~\ref{clm:truthification-nonnegative-pay-allocation-feasible}.
    \item \textbf{It is IR:}
    In Lemma~\ref{lem:truthification-buyer-i-bic} we showed that for buyer $i$,
    bidding truthfully in $\A^{(i)}$ results in nonnegative utility.
    For $j\neq i$, due to Lemma~\ref{lem:rearranging-preserves-bic} together
    with the fact that $\vbid^{(i-1)}$ is a BNE-\aee of $\A^{(i-1)}$ and the fact that $\A^{(i-1)}$ is IR,
    we have that for \aee $v_j$ and every $b_j$,
    \begin{align*}
        \E_{\bv_{-j}}\sqbr{\expect{\util}^{(i)}_j(\bid^{(i)}_{-j}(\bv_{-j}),v_j;v_j)}
        &=\E_{\bv_{-j}}\sqbr{\expect{\util}^{(i)}_j(\bid^{(i)}_{-j}(\bv_{-j}),\bid^{(i)}_j(v_j);v_j)}
       \geq0.
    \end{align*}
     \item \textbf{Interim allocations are nondecreasing for $j=1,\dots,i$:}
     Let $j<i$.
     Denote the interim allocation of buyer $j$ in $\A^{(i)}$
     given BNE-\aee $\vbid^{(i)}$ by
     \begin{align*}
        \hat{x}_j^{(i)}(v_j):=\E_{\bv_{-j},r}\sqbr{x^{(i)}_j(\vbid^{(i)}(\bv),r)},
     \end{align*}
     and similarly for $\A^{(i-1)}$.
     By definition of $\bid_j^{(i)}$ and
     $\bid_j^{(i-1)}$,
     for $j<i$ we have
     \begin{align*}
        \hat{x}_j^{(i)}(v_j)=\E_{\bv_{-j},r}\sqbr{x^{(i)}_j((\vbid_{-j}^{(i)}(\bv_{-j}),v_j),r)}
        \text{ and}\quad
        \hat{x}_j^{(i-1)}(v_j)=\E_{\bv_{-j},r}\sqbr{x^{(i-1)}_j((\vbid_{-j}^{(i-1)}(\bv_{-j}),v_j),r)}.
     \end{align*}

    It follows from Lemma~\ref{lem:rearranging-preserves-bic}
    that $\hat{x}_j^{(i)}(v_j)=\hat{x}_j^{(i-1)}(v_j)$.
    We assumed that $\A^{(i-1)}$
    satisfies condition~\ref{it:mono} for $i-1$,
    meaning that $\hat{x}_j^{(i-1)}(\cdot)$ is nondecreasing,
    and thus so is $\hat{x}_j^{(i)}(\cdot)$.

     For $i$, we showed in Lemma~\ref{lem:truthification-buyer-i-bic}
      that $\hat{x}_i^{(i)}(\cdot)$ is monotone-\aee
\item \textbf{Payments are Myerson payments almost everywhere for buyers $j=1,\dots,i$:}
  For $j=i$ this follows immediately from the definition of the truthification,
together with the fact we proved that $\A^{(i)}$ is BIC-\aee for $j=1,\dots,i$. For $j<i$ this follows from the induction hypothesis.

 \item \textbf{Expected payoff increased (weakly, and sometimes strictly):}
 For $j\neq i$ we
 first note that by definition of $\vbid^{(i)}$, we have
 $\bid_j^{(i)}(v_j)=\bid_j^{(i-1)}(v_j)$.
 Then we
 apply Lemma~\ref{lem:rearranging-preserves-bic} that shows that
 the realized allocations and payments for buyer $j$ when bidding
 $\bid_j^{(i)}(v_j)=\bid_j^{(i-1)}(v_j)$
 in $\A^{(i)}$ and $\A^{(i-1)}$
 are the same (when the other buyers play according to $\vbid^{(i)}$
 and $\vbid^{(i-1)}$, respectively).
  This implies that for every $j\neq i$,
\begin{align*}
    &\E_{\bv,r}\sqbr{\payoffValue(v_j)x_j^{(i)}(\vbid^{(i)}(\bv),r)
    +\payoffPay p_j^{(i)}(\vbid^{(i)}(\bv),r)}
    \\
    &\qquad=
    \E_{\bv,r}\sqbr{\payoffValue(v_j)x_j^{(i-1)}(\vbid^{(i-1)}(\bv),r)
    +\payoffPay p_j^{(i-1)}(\vbid^{(i-1)}(\bv),r)}.
\end{align*}

From Lemma~\ref{lem:truthification-buyer-i-bic}, it follows that $\expect{x}_i^{(i)}(\cdot;\vbid^{(i)})$
is equimeasurable to a nondecreasing rearrangement, denoted
$x_i^\uparrow(\cdot)$, of
$\expect{x}_i^{(i-1)}(\bid_i^{(i-1)}(\cdot);\vbid^{(i-1)})$.
Since $\bid_i^{(i)}(v_i)=v_i$, also  $\expect{x}_i^{(i)}(\cdot;\vbid^{(i)})$
is equimeasurable to $x_i^\uparrow(\cdot)$, and therefore
\begin{align*}
    \E_{v_i}\sqbr{\expect{x}_i^{(i)}(v_i;\vbid^{(i)})}=\E_{\bv,r}\sqbr{x_i^{(i)}(\vbid^{(i)}(\bv),r)}=\E_{v_i}\sqbr{x_i^\uparrow(v_i)}.
\end{align*}
Moreover,
by definition of $p_i^{(i)}$ and notations from Definition~\ref{def:truthification},
for every $v_i\in T_i$ (meaning also $\bid_i^{(i)}(v_i)=v_i\in T_i$), we have
\begin{align*}
    &\E_{\bv,r}\sqbr{p_i^{(i)}(\vbid^{(i)}(\bv),r)}
    =\E_{\bv,r}\left[x_i^{(i)}(\vbid^{(i)}(\bv),r)\gamma_i(\bid_i^{(i)}(v_i))\ \middle|\ v_i\in T_i\right]
    \\
    &\qquad=\E_{\bv,r}\left[x_i^{(i)}(\vbid^{(i)}(\bv),r)\gamma_i(v_i)\ \middle|\ v_i\in T_i\right]
    \\
    &\qquad=\E_{\bv,r}\left[x_i^{(i)}(\vbid^{(i)}(\bv),r)\brackets{v_i-\frac{\int_{\bottomV_i}^{v_i}\expect{x}_i^{(i)}(z;\vbid^{(i)})dz}{\expect{x}_i^{(i)}(v_i;\vbid^{(i)})}}\ \middle|\ v_i\in T_i\land\expect{x}_i^{(i)}(v_i;\vbid^{(i)})>0\right]
    \\
    &\qquad\qquad\cdot\Pr_{v_i}\sqbr{\expect{x}_i^{(i)}(v_i;\vbid^{(i)})>0}
    \\
    &\qquad=\E_{v_i}\left[\expect{x}_i^{(i)}(v_i;\vbid^{(i)})v_i-\int_{\bottomV_i}^{v_i}\expect{x}_i^{(i)}(z;\vbid^{(i)})dz\ \middle|\ v_i\in T_i\land\expect{x}_i^{(i)}(v_i;\vbid^{(i)})>0\right]
    \\
    &\qquad\qquad\cdot\Pr_{v_i}\sqbr{\expect{x}_i^{(i)}(v_i;\vbid^{(i)})>0}
    \\
    &\qquad=\E_{v_i}\left[x_i^\uparrow(v_i)v_i-\int_{\bottomV_i}^{v_i}x_i^\uparrow(z)dz\right]
    =\E_{v_i}\sqbr{p_i^\uparrow(v_i)}
\end{align*}
Thus, we can apply Claim~\ref{clm:single-buyer-rearrangement-payoff},
which yields
\begin{align*}
    &\E_{\bv,r}\sqbr{\payoffValue(v_i)x_i^{(i)}(\vbid^{(i)}(\bv),r)
    +\payoffPay p_i^{(i)}(\vbid^{(i)}(\bv),r)}
    \\
    &\qquad=
    \E_{\bv,r}\sqbr{\payoffValue(v_i)x_i^\uparrow(v_i)
    +\payoffPay p_i^\uparrow(v_i)}
    \\
    &\qquad\geq \E_{\bv,r}[\payoffValue(v_i)x^{(i-1)}_i(\vbid^{(i-1)}(\bv),r)
        +\payoffPay p^{(i-1)}_i(\vbid^{(i-1)}(\bv),r)].
\end{align*}
    Moreover, if $\expect{x}_i^{(i-1)}(\bid_i^{(i-1)}(\cdot);\vbid^{(i-1)})$ is not monotone-\aee
    and either $\payoffPay>0$ or $\payoffValue(\cdot)$ is
    strictly increasing, then the inequality is strict.

As for the allocation cost, using the definition of $T_i$ and $\A^{(i)}$ from
Definition~\ref{def:truthification},
\begin{align*}
    &\E_{\bv,r}\sqbr{c\brackets{\sum_{j=1}^n x_j^{(i)}\brackets{\vbid^{(i)}(\bv),r}}}
    \\
    &\qquad\underset{(1)}{=}\E_{\bv,r}\left[c\brackets{\sum_{j=1}^n x_j^{(i)}\brackets{\brackets{v_i,\vbid^{(i-1)}_{-i}(\bv_{-i})},r}}\,\middle|\ v_i\in T_i\right]
    \\
    &\qquad\underset{(2)}{=}\E_{\bv,r}\left[c\brackets{\sum_{j=1}^n x_j^{(i-1)}\brackets{\brackets{\bid_i^{(i-1)}(\sigma_i(v_i,r_\ell)),\vbid^{(i-1)}_{-i}(\bv_{-i})},r_a}}\,\middle|\ v_i\in T_i\right]
    \\
    &\qquad\underset{(3)}{=}\E_{\bv,r}\left[c\brackets{\sum_{j=1}^n x_j^{(i-1)}\brackets{\brackets{\bid_i^{(i-1)}(v_i),\vbid^{(i-1)}_{-i}(\bv_{-i})},r_a}}\,\middle|\ v_i\in T_i\right]
    \\
    &\qquad\underset{(4)}{=}\E_{\bv,r}\left[c\brackets{\sum_{j=1}^n x_j^{(i-1)}\brackets{\vbid^{(i-1)}(\bv),r_a}}\right]
\end{align*}
where (1) uses the definition of $\vbid^{(i)}$ and the fact that $\Pr_{v_i\sim F_i}\sqbr{v_i\in T_i}=1$;
(2) uses the definition of $\A^{(i)}$;
(3) relies on the fact that for $r\sim\Uni[0,1]$ and $v_i\sim F_i$, we have that $\sigma_i(v_i,r)\sim F_i$;
and (4) relies on the fact that $\Pr_{v_i\sim F_i}\sqbr{v_i\in T_i}=1$.

    Putting it all together,
    \begin{align*}
        \E\sqbr{\A^{(i)}\mid\,\vbid^{(i)}}
        &=\sum_{j=1}^n\E_{\bv,r}\sqbr{\payoffValue(v_j)x_j^{(i)}(\vbid^{(i)}(\bv),r)
    +\payoffPay p_j^{(i)}(\vbid^{(i)}(\bv),r)}\\
        &\qquad-\E_{\bv,r}\sqbr{c\brackets{\sum_{j=1}^n x^{(i)}_j(\vbid^{(i)}(\bv),r)}}\\
        &\geq \sum_{j=1}^n\E_{\bv,r}\sqbr{\payoffValue(v_j)x_j^{(i-1)}(\vbid^{(i-1)}(\bv),r)
    +\payoffPay p_j^{(i-1)}(\vbid^{(i-1)}(\bv),r)}\\
        &\qquad-\E_{\bv,r}\sqbr{c\brackets{\sum_{j=1}^n x^{(i-1)}_j(\vbid^{(i-1)}(\bv),r)}}\\
        &=\E_{\bv,r}\sqbr{\A^{(i-1)}\mid\,\vbid^{(i-1)}},
    \end{align*}
    with a strict inequality if $\A^{(i-1)}$ is not monotone-\aee for buyer $i$
    and either $\payoffPay>0$ or $\payoffValue(\cdot)$ is strictly increasing.
\end{enumerate}
Set $\A':=\A^{(n)}$, and we are done.
\end{proof}

Now we can prove Theorem~\ref{thm:rev-myerson}.

\begin{proof}[Proof of Theorem~\ref{thm:rev-myerson}]
    Let $\A^{(0)}$ be an auction
    with BNE $\vbid^{(0)}$.
    By Theorem~\ref{thm:rev-aligned-bid-constraint-mono-wlog}
     there exists a BIC-\aee and monotone-\aee auction $\A^{(1)}$,
     such that $\E\sqbr{\A^{(1)}\mid\vbid^{tru}}\geq\E\sqbr{\A^{(0)}\mid\vbid^{(0)}}$,
     where $\vbid^{tru}$ denotes the truthful strategy (since $\A^{(1)}$
     is BIC-\aee, $\vbid^{tru}$ is a BNE-\aee).
    By Theorem~\ref{thm:rev-aligned-bid-constraint-mono-wlog}, the inequality is strict if
    $\A^{(0)}$ is not monotone-\aee and either $\payoffValue(\cdot)$
    is strictly increasing or $\payoffPay>0$.
    Thus if $\A^{(0)}$ is optimal and either $\payoffValue(\cdot)$
    is strictly increasing or $\payoffPay>0$,
    it must be monotone-\aee and
    item (2) follows immediately.

    Now we consider auction $\A^{(2)}$ as defined in Definition~\ref{def:rev-aligned-ironed-maximizer}, which is one-shot and satisfies Eq.~\eqref{eq:unique-pay}. Additionally, by Lemma~\ref{lem:rev-aligned-optimal-among-tight} it is
    U-DSIC and monotone.

    Let $\mathcal{V}:=\mathcal{V}_1\times\dots\mathcal{V}_n\subseteq[\bottomV_1,\topV_1]\times\dots[\bottomV_n,\topV_n]$
    be the set of valuation profiles such that:
    \begin{enumerate}
        \item $\expect{x}^{(1)}_i(\cdot;\vbid^{tru})$ is nondecreasing over $\mathcal{V}_i$
        for every $i\in[n]$,
        \item for every $i\in[n]$ and $v_i\in\mathcal{V}_i$, buyer $i$ bids
        truthfully in $\A^{(1)}$ given valuation $v_i$, and
        \item for every $i\in[n]$ and $v_i\in\mathcal{V}_i$, buyer $i$'s payment
        in $\A^{(1)}$ is the Myerson payment.
    \end{enumerate}
    Since $\A^{(1)}$ is BIC-\aee, monotone-\aee,
    and has Myerson payments a.e.,
    $\Pr_{\bv}[\bv\in\mathcal{V}]=1$.

    The expected payoff of $\A^{(2)}$ is
    \begin{align*}
        \E\sqbr{\A^{(2)}}&\underset{(1)}{=}
        \E_{\bv,r}\sqbr{\sum_{i=1}^n \irn_{i,\Pi}(v_i)x^{(2)}_i(\bv,r)-c\brackets{\sum_{i=1}^n x^{(2)}_i(\bv,r)}}\\
        &\underset{(2)}{\geq}\E_{\bv,r}\sqbr{\sum_{i=1}^n \irn_{i,\Pi}(v_i)x^{(1)}_i(\bv,r)-c\brackets{\sum_{i=1}^n x^{(1)}_i(\bv,r)}}\\
        &\underset{(3)}{\geq}\E_{\bv,r}\left[\sum_{i=1}^n \varphi_{i,\Pi}(v_i)x^{(1)}_i(\bv,r)-c\brackets{\sum_{i=1}^n x^{(1)}_i(\bv,r)}\right]\\
        &\underset{(4)}{=}\E_{\bv,r}\left[\sum_{i=1}^n \varphi_{i,\Pi}(v_i)x^{(1)}_i(\bv,r)-c\brackets{\sum_{i=1}^n x^{(1)}_i(\bv,r)}\middle|\ \bv\in\mathcal{V}\right]\\
        &\underset{(5)}{=}\E_{\bv,r}\left[\sum_{i=1}^n \brackets{\payoffValue(v_i)x^{(1)}_i(\bv,r)+\payoffPay p^{(1)}_i(\bv,r)}-c\brackets{\sum_{i=1}^n x^{(1)}_i(\bv,r)}\middle|\ \bv\in\mathcal{V}\right]\\
        &\underset{(6)}{=}\E_{\bv,r}\left[\sum_{i=1}^n \brackets{\payoffValue(v_i)x^{(1)}_i(\bv,r)+\payoffPay p^{(1)}_i(\bv,r)}-c\brackets{\sum_{i=1}^n x^{(1)}_i(\bv,r)}\right]\\
        &=\E\sqbr{\A^{(1)}}\\
        &\geq \E\sqbr{\A^{(0)}\mid\vbid^{(0)}},
    \end{align*}
    where (1) is due to Lemma~\ref{lem:rev-aligned-optimal-among-tight},
    (2) follows from $\vec{x}^{(2)}$ being a maximizer of the ironed
    virtual welfare,
      (3) is due to Lemma~\ref{lemma:ironing-properties},
    (4) and (6) are due to $\Pr_{\bv}[\bv\in\mathcal{V}]=1$,
    and finally (5) is due to Myerson's analysis~\cite{myerson1981optimal},
    combined with the payments on $\mathcal{V}$ being Myerson payments.
    This completes the proof of item (1).

    For item (3), let $\A$ be an optimal BIC auction and assume $\payoffPay>0$.
    If the payments are not Myerson payments a.e.\ for
    some $i\in[n]$ then there exists a set
     $\mathcal{U}_i\subseteq[\bottomV_i,\topV_i]$
     with positive probability
    such that for every $v_i\in\mathcal{U}_i$,
    we have $\E_{\bv_{-i},r}[p_i(\bv,r)]<\E_{\bv_{-i},r}\sqbr{p_i^{My}(v_i)}$,
    and thus by Myerson's analysis,
    \begin{align*}
        \E_{\bv,r}\sqbr{\payoffValue(v_i)x_i(\bv,r)+\payoffPay p_i(\bv,r)}
        &=\E_{\bv,r}\sqbr{\payoffValue(v_i)x_i(\bv,r)+\payoffPay p_i(\bv,r)\mid v_i\in\mathcal{U}_i}
        \Pr[ v_i\in\mathcal{U}_i]
        \\&\qquad+\E_{\bv,r}\sqbr{\payoffValue(v_i)x_i(\bv,r)+\payoffPay p_i(\bv,r)\mid v_i\notin\mathcal{U}_i}
        \Pr[ v_i\notin\mathcal{U}_i]\\
        &<\E_{\bv,r}\sqbr{\payoffValue(v_i)x_i(\bv,r)+\payoffPay p^{My}_i(v_i)\mid v_i\in\mathcal{U}_i}
        \Pr[ v_i\in\mathcal{U}_i]
        \\&\qquad+\E_{\bv,r}\sqbr{\payoffValue(v_i)x_i(\bv,r)
        +\payoffPay p^{My}_i(v_i)\mid v_i\notin\mathcal{U}_i}
        \Pr[ v_i\notin\mathcal{U}_i]\\
        &=\E_{\bv,r}\sqbr{\payoffValue(v_i)x_i(\bv,r)
        +\payoffPay p^{My}_i(v_i)}\\
        &=\E_{\bv,r}\sqbr{\varphi_{i,\Pi}(v_i)x_i(\bv,r)},
    \end{align*}
    where the inequality is strict due to $\Pr[ v_i\in\mathcal{U}_i]>0$
    and $\payoffPay>0$.

By item (2), $\A$ is monotone-a.e. Then by
    Corollary~\ref{cor:pay-bound-by-varphi} and
     Lemma~\ref{lemma:ironing-properties}, for every $j\in[n]$,
    \begin{align*}
        \E_{\bv,r}\sqbr{\payoffValue(v_j)x_j(\bv,r)+\payoffPay p_j(\bv,r)}
        &\leq\E_{\bv,r}\sqbr{\varphi_{j,\Pi}(v_j)x_j(\bv,r)}
        \le\E_{\bv,r}\sqbr{\irn_{j,\Pi}(v_j)x_j(\bv,r)}.
    \end{align*}

    Comparing to $\A^{(2)}$ just like the analysis above, together
    with the strict inequality for $i$ yields
    \begin{align*}
        \E[\A]<\E\sqbr{\A^{(2)}},
    \end{align*}
    in contradiction to the optimality of $\A$. This completes the
    proof of item (3).

\end{proof}

\section{Consumer-Aligned Objectives}\label{sec:consumer}

This section gives the full proofs for the consumer-aligned results from Section~\ref{sec:overview-cons-aligned}.
We first set up the relevant notations and definitions, then prove the threshold-payment and virtual-value bounds used by the optimality argument, and finally prove the two monotonicity results.
Throughout this section, we write $\conPayoffPay=-\payoffPay>0$ for the payment coefficient in a consumer-aligned payoff.
Unless otherwise stated, the comparison class in this section consists of deterministic, DSIC, ex-post-IR auctions with feasible allocations and nonnegative payments. Here ex-post IR is a property of the mechanism, as defined in Section~\ref{sec:constrained-buyers}, and is distinct from the stagewise-IR constraint imposed on buyers.

\subsection{Stagewise Outcomes and Screening Auctions}\label{sec:stagewise-outcomes}

We begin by spelling out the stagewise-IR constraint in the notation of our model.
For stagewise-IR,
    \begin{align*}
        \pred_i(S_i,b_i;v_i)=\begin{cases}
            \true, & \Pr\left[\sum_{t=1}^{T} \brackets{x_i^t(S_i)v_i-p_i^t(S_i)}<0\right]= 0\quad \forall T=1,\dots,|S_i|\\
            \false, &\text{otherwise}
        \end{cases}.
    \end{align*}

\paragraph{Stagewise auctions with seller cost (SWAC).}
In~\cite{berzack2025dynamic} the notion of a \emph{SWAC} is introduced:
it is equivalent to a single-buyer auction in our setting, and the buyer is subject to stagewise-IR. They consider a class of objectives, ``negative tradeoff'', which is just slightly more restrictive than our consumer-aligned class.
We adopt some definitions from~\cite{berzack2025dynamic}, and extend several of their results to our more general setting; building on these results, we use new techniques, such as the rearrangement technique described in Section~\ref{sec:overview-rev-aligned}, to provide \emph{new} results for the multi-buyer setting.

\paragraph{Auction outcomes.}
The \emph{outcome} of the auction
{is represented by a tuple $\vec{s}$ of $n$ sequences, such that $s_i \in \left( [m] \times \mathbb{R}_{\geq 0} \right)^{\ast}$ for all $i\in\{1, \dots, n\}$. The component
$s_i^t = (x_i^t, p_i^t)$
denotes the allocation and the payments charged
to each buyer
at time $t = 1,\ldots,|s_i|$.%
\footnote{We note that in the general case an auction outcome can include additional properties, but the results in our paper are restricted to outcomes of the above type.}
It is convenient to denote by $x_i, p_i$
the \emph{total allocation} $x_i(s_i) = \sum_{t = 1}^{|s_i|} x_i^t$ and the \emph{total payment} $p_i(\vec{s}) = \sum_{t = 1}^{|s_i|} p_i^t$ of buyer $i$.
}

We refer to a standard auction, where in every outcome the allocation and payment happen together in one go, as a \emph{one-shot} auction.
We omit the outcome $\vec{s}$ from the notation when it is clear from the context.

Beyond the \emph{unrestricted feasibility set}, which admits all allocations constrained only by the total supply $m$, in this section we consider also the \emph{single-winner} feasibility set (or environment), which admits every allocation of at most $m$ units in which at most one buyer receives units between $1$ and $m$.

\paragraph{Screening auctions.}
For an arbitrary deterministic stagewise auction, the \emph{filter} of an outcome $s_i$ is the largest average price per unit that buyer $i$ must be able to pay along the way.
Formally, for each prefix $T=0,\dots,|s_i|$, let $X_i^T(s_i)=\sum_{t=1}^{T}x_i^t(s_i)$ and $P_i^T(s_i)=\sum_{t=1}^{T}p_i^t(s_i)$.
Then
\[
    f_i(s_i)\coloneqq
    \max_{T=0,\dots,|s_i|}
    \frac{P_i^T(s_i)}{X_i^T(s_i)},
\]
where a term with $X_i^T(s_i)=0$ is interpreted as $0$ if $P_i^T(s_i)=0$ and as $\infty$ otherwise.
For a deterministic stagewise auction, we write $f_i(\vec{b})$ for $f_i(s_i(\vec{b}))$, the filter of the outcome induced by bid profile $\vec{b}$, and similarly $f_i(v_i,\bv_{-i})$ when the auction is evaluated at a valuation profile.
Thus a stagewise-IR buyer with valuation $v_i$ can accept an outcome $s_i$ if and only if $v_i\geq f_i(s_i)$.

A screening auction is a deterministic stagewise auction with the following two-phase form:
for any submitted bid profile $\vec{b}$,
the outcome contains two phases, a \emph{screening phase} ($t=1$) and a \emph{main phase} ($t=2$).
During the screening phase, the seller allocates at most one unit to every buyer $i$ at a \emph{personalized screening price} that must be paid in that phase;
if buyer $i$ was allocated a unit in the screening phase,
then in the
main phase the seller can allocate more units and collect more payment from that buyer (otherwise, buyer $i$ gets nothing in the main phase).
We further require that payments in this mechanism are \emph{front-loaded}:
if for some bid profile $\vec{b}$ a buyer receives a positive allocation in the screening phase (i.e., $x_i^1(\vec{b}) > 0$), the screening payment must be at least $p_i(\vec{b}) / x_i(\vec{b})$.
Under this condition, the personalized screening price coincides with the general filter $f_i(\vec{b})$ of the resulting stagewise outcome.
This guarantees that if a
stagewise-IR buyer is able to bid $b_i$ given bids $\vec{b}_{-i}$, their valuation is at least as high as the filter: $v_i\geq f_i(\vec{b})$.

For stagewise-IR buyers in a deterministic stagewise auction, the acceptability predicate can equivalently be written as
\begin{align*}
    \pred_i(S_i,b_i;v_i)=\begin{cases}
        \true, & \Pr[f_i(S_i)>v_i]=0\\
        \false, &\text{otherwise}
    \end{cases}.
\end{align*}

The following claim lower-bounds the payment of a stagewise-IR buyer in any deterministic DSIC auction by the threshold at which the buyer's allocation increases; note that it does not require the auction to be monotone.

\begin{claim}\label{clm:payment-at-least-thresh}
    Suppose buyers are stagewise-IR, and
    let $\A=(\vec{x},\vec{p})$ be a deterministic DSIC auction.

    Then for every $i$ and $\bv_{-i}$,
    the payment is lower bounded by

    \begin{align}
        p_i(\bv)\geq\sup\set{z:x_i(z,\bv_{-i})<x_i(\bv)},
    \end{align}
    where $\sup\emptyset$ is defined to be $0$.
\end{claim}
\begin{proof}
    Let $\A=(\vec{x},\vec{p})$ be deterministic and DSIC.

    Suppose towards contradiction that $p_i(\bv)<\sup\set{z:x_i(z,\bv_{-i})<x_i(\bv)}$ for some $i$ and $\bv$. Thus there exists some $z>p_i(\bv)$ with $x_i(z,\bv_{-i})<x_i(\bv)$.
    Observe that
    \begin{align*}
        u_i((z,\bv_{-i});z)=x_i(z,\bv_{-i})z-p_i(z,\bv_{-i})\leq x_i(\bv)z-z<x_i(\bv)z-p_i(\bv)=\unconstrained_i(\bv;z)=u_i(\bv;z),
    \end{align*}
    where the first inequality is due to $x_i(z,\bv_{-i})\leq x_i(\bv)-1$ (because allocations are deterministic and thus integers) and $p_i(z,\bv_{-i})\geq0$, and the last transition is due to $z>p_i(\bv)\geq f_i(\bv)$, so it is safe for $z$ to bid $v_i$ (given $\bv_{-i}$). We've reached a contradiction to DSIC, completing the proof.
\end{proof}

Complementing Claim~\ref{clm:payment-at-least-thresh}, a buyer who receives fewer units than some lower valuation must pay more than that valuation.

\begin{claim}\label{clm:no-upward-imitation}
    Suppose buyers are stagewise-IR, and let $\A=(\vec{x},\vec{p})$ be a deterministic, DSIC, and ex-post-IR auction.
    For every $i$ and $\bv_{-i}$, if $y<w$ and $x_i(y,\bv_{-i})>x_i(w,\bv_{-i})$, then $p_i(w,\bv_{-i})>y$.
\end{claim}
\begin{proof}
    Write $x(\cdot)=x_i(\cdot,\bv_{-i})$ and $p(\cdot)=p_i(\cdot,\bv_{-i})$.
    Since $\A$ is ex-post IR, the truthful outcome at $(y,\bv_{-i})$ is acceptable, so its filter is at most $y<w$; hence a buyer with valuation $w$ can bid $y$, and by DSIC, $wx(w)-p(w)\geq wx(y)-p(y)$.
    If a buyer with valuation $y$ could bid $w$, then by DSIC also $yx(y)-p(y)\geq yx(w)-p(w)$, and adding the two inequalities would give $(x(y)-x(w))(w-y)\leq0$, a contradiction.

    Hence, a buyer with valuation $y$ cannot bid $w$, meaning that the filter of the outcome at bid $w$ is larger than $y$.
    This filter is at most $p(w)$, so we are done.
\end{proof}

In particular, once a buyer starts winning, higher valuations must also win.

\begin{claim}\label{clm:start-win-always-win}
    Suppose buyers are stagewise-IR, and let $\A=(\vec{x},\vec{p})$ be a deterministic, DSIC, and ex-post-IR auction.
    For every $i$ and $\bv$, if $x_i(\bv)>0$, then $x_i(v_i',\bv_{-i})>0$ for every $v_i'>v_i$.
\end{claim}
\begin{proof}
    Otherwise, $x_i(\bv)>0=x_i(v_i',\bv_{-i})$ for some $v_i'>v_i$, and Claim~\ref{clm:no-upward-imitation} gives $p_i(v_i',\bv_{-i})>v_i\geq0$.
    This contradicts ex-post IR, since buyer $i$ receives nothing at $(v_i',\bv_{-i})$.
\end{proof}

\subsection{The Optimal Payoff Might Not Be Attainable}

\begin{proposition}[Non-attainment of optimal consumer surplus]
\label{prop:consumer-surplus-nonattainment}
Consider a setting with three units, two stagewise-IR buyers with independent valuations uniformly distributed on $[0,1]$, and zero seller cost. Over all deterministic, DSIC, ex-post-IR auctions with nonnegative payments, the supremum of the expected consumer surplus is $5/3$, but no auction attains this payoff.
\end{proposition}

\begin{proof}
Throughout, $\Pi$ is the consumer-surplus payoff ($\payoffValue(v)=v$, $\payoffPay=-1$) with $c\equiv 0$, so for a DSIC auction $\A=(\vec{x},\vec{p})$,
\[
\E[\A]=\E_{\bv}\sqbr{\sum_{i=1}^2 \brackets{v_i x_i(\bv) - p_i(\bv)}}.
\]

We first establish an upper bound on the expected consumer surplus that does not rely on the allocation rule being monotone.
Fix a buyer $i$ and the other buyer's valuation $\bv_{-i}$.
By Claim~\ref{clm:payment-at-least-thresh}, for every $v_i$,
\[
p_i(\bv) \ge \sup\set{z\in[0,1] : x_i(z,\bv_{-i}) < x_i(\bv)},
\]
where $\sup\emptyset = 0$.
Thus $x_i(z,\bv_{-i}) \ge x_i(\bv)$ for every $z\in(p_i(\bv),v_i]$, and integrating the allocation yields
\[
\int_0^{v_i} x_i(z,\bv_{-i})\,dz \ge \brackets{v_i - p_i(\bv)}_+ x_i(\bv).
\]
Rearranging, and using $x_i(\bv)\le 3$,
\[
v_i x_i(\bv) - \int_0^{v_i} x_i(z,\bv_{-i})\,dz \le \min\set{v_i, p_i(\bv)} x_i(\bv) \le 3p_i(\bv).
\]
We now integrate over $v_i \in [0,1]$.
For the integral term, swapping the order of integration gives
\[
\int_0^1 \int_0^{v_i} x_i(z,\bv_{-i})\,dz\,dv_i
= \int_0^1 x_i(z,\bv_{-i}) \int_z^1 dv_i\,dz
= \int_0^1 (1-z)\, x_i(z,\bv_{-i})\,dz.
\]
Renaming the integration variable $z$ to $v_i$ and combining with the first term, integrating the previous inequality over $v_i\in[0,1]$ yields
\[
3 \int_0^1 p_i(v_i,\bv_{-i})\,dv_i
\ge \int_0^1 v_i\, x_i(v_i,\bv_{-i})\,dv_i - \int_0^1 (1-v_i)\, x_i(v_i,\bv_{-i})\,dv_i
= \int_0^1 (2v_i - 1)\, x_i(v_i,\bv_{-i})\,dv_i.
\]
Thus the expected consumer surplus of buyer $i$ conditional on $\bv_{-i}$ is
\begin{align*}
\E_{v_i}\sqbr{v_i x_i(\bv) - p_i(\bv) \mid \bv_{-i}}
&= \int_0^1 v_i\, x_i(v_i,\bv_{-i})\,dv_i - \int_0^1 p_i(v_i,\bv_{-i})\,dv_i\\
&\le \int_0^1 \brackets{v_i - \frac{2v_i - 1}{3}} x_i(v_i,\bv_{-i})\,dv_i
= \E_{v_i}\sqbr{\frac{v_i + 1}{3}\, x_i(\bv) \;\middle|\; \bv_{-i}}.
\end{align*}
Taking the expectation over $\bv_{-i}$ and summing over both buyers gives
\begin{equation}\label{eq:cs-nonattain-virtual}
\E[\A] \le \E_{\bv}\sqbr{\sum_{i=1}^2 \frac{v_i + 1}{3}\, x_i(\bv)}.
\end{equation}
Next, for every $\bv$ we have $\frac{v_i+1}{3}\le\frac{\max\set{v_1,v_2}+1}{3}$ for both $i$, and at most $3$ units are allocated in total, so pointwise
\begin{equation}\label{eq:cs-nonattain-pointwise}
\sum_{i=1}^2 \frac{v_i + 1}{3}\, x_i(\bv)
\le \frac{\max\set{v_1,v_2}+1}{3}\sum_{i=1}^2 x_i(\bv)
\le \max\set{v_1,v_2}+1.
\end{equation}
Since $\E[\max\set{v_1,v_2}]=2/3$ for two independent uniform valuations, combining~\eqref{eq:cs-nonattain-virtual} and~\eqref{eq:cs-nonattain-pointwise} yields
\[
\E[\A] \le \E_{\bv}\sqbr{1 + \max\set{v_1, v_2}} = 1+\frac{2}{3} = \frac{5}{3}.
\]

We have now proven that the expected consumer surplus of any deterministic DSIC auction is at most $5/3$, and we next show that this bound is not attainable by a DSIC auction.

Assume, for contradiction, that there exists a DSIC auction $\A$ with $\E[\A] = 5/3$. Then~\eqref{eq:cs-nonattain-virtual} and both inequalities in~\eqref{eq:cs-nonattain-pointwise} are tight almost everywhere, so all three units go to the buyer with the higher valuation.
Hence, for almost every $\bv_{-i}$ and almost every $v_i$, buyer $i$ receives $3$ units if $v_i > v_{-i}$ and $0$ units if $v_i < v_{-i}$.
For such $\bv_{-i}$ and a winning $v_i$, the set $\set{z : x_i(z,\bv_{-i})<3}$ has supremum at least $v_{-i}$, so by Claim~\ref{clm:payment-at-least-thresh}, $p_i(\bv)\ge v_{-i}$.
The expected total payment is therefore at least $\E[\min\set{v_1, v_2}] = 1/3$. Since the expected welfare is $3\, \E[\max\set{v_1, v_2}] = 2$, attaining $\E[\A]=5/3$ requires that the winner pays exactly the other buyer's valuation and the loser pays nothing, almost everywhere.

Since these conclusions hold almost everywhere, for almost every $w \in (0,1)$ they hold on the slice $\bv_{-i}=w$ for both buyers: for each $i\in\set{1,2}$ and almost every $b_i > w$, at the profile $(b_i,w)$ buyer $i$ receives $3$ units and pays exactly $w$.
Fix such a $w$ and consider the tie profile $\bv=(w,w)$.
A buyer with valuation $w$ could bid such a $b_i > w$, win $3$ units, and pay $w$, so their filter is at most $w$ and the bid is acceptable; this yields utility $3w - w = 2w$.
Since $\A$ is DSIC, both buyers must obtain utility at least $2w$ at $(w,w)$.
However, the maximum welfare at $(w,w)$ is $3w$, and with nonnegative payments the total utility cannot exceed it, so $4w \le 3w$, which contradicts $w > 0$. Therefore no deterministic DSIC auction attains $5/3$.

Finally, to see that $5/3$ is the supremum, consider for small $\delta > 0$ the screening auction $\A_\delta$ defined as follows.
Buyer $i$ is a winner if $b_i \ge b_{-i} + \delta$; since $\delta>0$, at most one buyer wins.
In the screening phase, the winner is allocated one unit at screening price $b_{-i} + \delta$, and in the main phase they receive the remaining two units at no charge; the loser is allocated nothing and pays nothing.
Thus the winner's filter is $f_i(\vec{b}) = p_i(\vec{b}) = b_{-i} + \delta$, which is exactly the lowest bid with which buyer $i$ wins.
Hence a losing buyer cannot profitably deviate to winning (their valuation is below the filter of a winning bid), and a truthful winner has nonnegative utility for every fixed opposing bid, so $\A_\delta$ is DSIC and ex-post IR.
Its expected consumer surplus is
\[
\E[\A_\delta] = 2\int_0^{1-\delta}\int_{v_{-i}+\delta}^1 \brackets{3v_i-v_{-i}-\delta}\,dv_i\,dv_{-i} = \frac{5}{3} - 3\delta + \delta^2 + \frac{\delta^3}{3},
\]
which converges to $5/3$ as $\delta \downarrow 0$, establishing the value of the supremum.
\end{proof}

\subsection{Threshold Payments and Virtual-Welfare Bounds}

We use the consumer-aligned virtual values $\varphi_{i,\Pi}^{m}$ and $\irn_{i,\Pi}^{m}$ from Definition~\ref{def:cons-aligned-virtual-values}.
The next definition is the allocation underlying our concave-cost results: it maximizes ironed virtual welfare pointwise.

\begin{definition}\label{def:cons-aligned-ironed-maximizer}
    Let $\Pi$ be a consumer-aligned payoff.
    Call an allocation $\vec{z}\in\Omega$ \emph{concentrated} if all its allocated units go to a single bidder, i.e., $z_j=0$ for all $j$ except one.
    Fix a total order $\prec$ over $\Omega$ satisfying the following two properties:
\begin{enumerate}
    \item Among allocations with the same total number of units, concentrated allocations are preferred to non-concentrated allocations.
    \item If $\vec{x},\vec{y}\in\Omega$ differ only in the allocation of bidder $i$, with $x_i=k<m$ and $y_i=m$, then $\vec{x}\prec\vec{y}$.
\end{enumerate}

Then the \emph{IVW allocation} $\vec{x}^*$ is defined by
\begin{align}\label{eq:cons-aligned-argmax-def}
&\vec{x}^*(\bv) := \max_{\prec} \argmax_{\vec{y}\in\Omega}
\brackets{ \sum_{i=1}^n \irn_{i,\Pi}^m(v_i)y_i - c\left(\sum_{i=1}^n y_i\right) }.
\end{align}
\end{definition}

The IVW allocation is not by itself an auction; in Section~\ref{sec:cons-near-optimality} we use it to define the auction \consMaximizer (Definition~\ref{def:cons-shifted-ivw}), as a special case of the following construction.

\paragraph{Shifted threshold auctions.}
Charging threshold payments for a monotone allocation rule need not yield a DSIC auction with multiple buyers: if a buyer who bids exactly the threshold for $k$ units receives fewer than $k$ units (e.g., because of a tie), then a buyer with that valuation can profit from overbidding (see Example~\ref{ex:threshold-tie}).
We avoid this by shifting all thresholds up by $\delta>0$.

\begin{definition}[$\delta$-shifted threshold auction]\label{def:shifted-threshold}
    Let $\vec{x}$ be a deterministic, allocation-feasible, ex-post monotone allocation rule, and let $\delta>0$.
    For every $i$, $\vec{b}_{-i}$ and $k\in\set{1,\dots,m}$, let
    \[
        \tau_{i,k}(\vec{b}_{-i})\coloneqq
        \begin{cases}
            0, & \text{if } x_i(\bottomV_i,\vec{b}_{-i})\geq k,\\
            \inf\set{z\in[\bottomV_i,\topV_i] : x_i(z,\vec{b}_{-i})\geq k}, & \text{otherwise},
        \end{cases}
    \]
    be the threshold for receiving $k$ units under $\vec{x}$, where $\inf\emptyset=\infty$.
    The \emph{$\delta$-shifted threshold auction} of $\vec{x}$ is the screening auction in which, for a bid profile $\vec{b}$, buyer $i$ receives
    \[
        x^\delta_i(\vec{b})\coloneqq\max\set{k\in\set{1,\dots,m} : b_i\geq\tau_{i,k}(\vec{b}_{-i})+\delta}
    \]
    units, where $\max\emptyset=0$.
    If $k=x^\delta_i(\vec{b})\geq1$, buyer $i$ is allocated one unit in the screening phase at screening price $\tau_{i,k}(\vec{b}_{-i})+\delta$, and the remaining $k-1$ units in the main phase at no charge; otherwise, buyer $i$ is allocated nothing and pays nothing.
\end{definition}

\begin{lemma}\label{lem:shifted-threshold-dsic}
    Suppose buyers are stagewise-IR, and let $\vec{x}$ be a deterministic, allocation-feasible, ex-post monotone allocation rule.
    Then for every $\delta>0$, the $\delta$-shifted threshold auction of $\vec{x}$ is deterministic, allocation-feasible, ex-post monotone, DSIC, and ex-post IR.
    Moreover, $x^\delta_i(\vec{b})\leq x_i(\vec{b})$ for every $i$ and $\vec{b}$, and $x^\delta_i(\vec{b})\to x_i(\vec{b})$ as $\delta\to0$ whenever $b_i\notin\set{\tau_{i,1}(\vec{b}_{-i}),\dots,\tau_{i,m}(\vec{b}_{-i})}$.
\end{lemma}
\begin{proof}
    Fix $i$ and $\vec{b}_{-i}$, and write $\tau_k=\tau_{i,k}(\vec{b}_{-i})$; note that $\tau_1\leq\dots\leq\tau_m$ and $\tau_k\geq0$.

    It is easy to see that $\vec{x}^\delta(\vec{b})\leq\vec{x}(\vec{b})$ coordinatewise, and since the feasibility set is downward-closed, $\vec{x}^\delta(\vec{b})$ is allocation-feasible. Also the ex-post monotonicity of $\vec{x}^\delta$ follows from the monotonicity of the thresholds in $k$.

    Next we prove that the auction is DSIC and ex-post IR.
    Fix a valuation $v_i$ and let $K=x^\delta_i(v_i,\vec{b}_{-i})$ be the number of units buyer $i$ receives when bidding truthfully.
    An outcome in which buyer $i$ receives $k\geq1$ units has filter $\tau_k+\delta$, hence bidding truthfully is acceptable.
    Any bid that yields $k>K$ units is unacceptable, since by the definition of $K$, the filter $\tau_k+\delta$ exceeds $v_i$.
    If $K\geq1$, the truthful utility is $u_i((v_i,\vec{b}_{-i});v_i)=Kv_i-(\tau_K+\delta)\geq(K-1)v_i\geq0$, since $v_i\geq\tau_K+\delta$; if $K=0$, it is $0$.
    Any bid $b_i'$ that yields $0$ units yields $u_i((b_i',\vec{b}_{-i});v_i)=0$.
    Any bid $b_i'$ that yields $1\leq k<K$ units is acceptable, and
    \begin{align*}
        u_i((v_i,\vec{b}_{-i});v_i)-u_i((b_i',\vec{b}_{-i});v_i)
        &=\brackets{Kv_i-(\tau_K+\delta)}-\brackets{kv_i-(\tau_k+\delta)}\\
        &=(K-k)v_i-(\tau_K-\tau_k)\geq v_i-\tau_K>0,
    \end{align*}
    since $K-k\geq1$, $\tau_k\geq0$, and $v_i\geq\tau_K+\delta$.
    Hence truthful bidding is a best response and yields nonnegative utility.

    It remains to prove the convergence property.
    Recall that $\tau_k$ is defined with respect to the original allocation rule $\vec{x}$ (not the shifted one): it is the infimum of the bids $z$ with which buyer $i$ receives at least $k$ units under $\vec{x}$, i.e., with $x_i(z,\vec{b}_{-i})\geq k$, or $0$ if already the lowest bid is such a bid.
    If $b_i<\tau_k$, then $b_i$ is not such a bid, so $x_i(\vec{b})<k$.
    If $b_i>\tau_k$, then either the lowest bid $\bottomV_i$ is such a bid, or by the definition of the infimum there is such a bid $z<b_i$; in both cases, by monotonicity of $\vec{x}$, $x_i(\vec{b})\geq k$.
    Hence if $b_i\neq\tau_k$ for all $k$, then $x_i(\vec{b})=\max\set{k : b_i>\tau_k}$, which equals $x^\delta_i(\vec{b})$ whenever $\delta<b_i-\tau_k$ for every $k$ with $\tau_k<b_i$.
\end{proof}

As $\delta\to0$, the payoff of the $\delta$-shifted threshold auction of $\vec{x}$ approaches the payoff that $\vec{x}$ would obtain with (unshifted) threshold payments, even though the latter need not be DSIC. We denote this limit by $\textsf{threshPayoff}(\vec{x})$.

\begin{lemma}\label{lem:shifted-threshold-payoff}
    Suppose buyers are stagewise-IR, and let $\vec{x}$ be a deterministic, allocation-feasible, ex-post monotone allocation rule with thresholds $\tau_{i,k}$ as in Definition~\ref{def:shifted-threshold}, and let $\tau_{i,0}\coloneqq0$.
    Then, as $\delta\to0$, the expected payoff of the $\delta$-shifted threshold auction of $\vec{x}$ converges to the \emph{threshold payoff} of $\vec{x}$,
    \[
        \textsf{threshPayoff}(\vec{x})\coloneqq\E_{\bv}\sqbr{\sum_{i=1}^n\brackets{\payoffValue(v_i)x_i(\bv)-\conPayoffPay\,\tau_{i,x_i(\bv)}(\bv_{-i})}-c\brackets{\sum_{i=1}^n x_i(\bv)}}.
    \]
\end{lemma}
\begin{proof}
    Denote by $\vec{x}^\delta$ and $\vec{p}^\delta$ the allocation and payment rules of the $\delta$-shifted threshold auction of $\vec{x}$, which is DSIC by Lemma~\ref{lem:shifted-threshold-dsic}, so buyers bid truthfully.
    Fix $i$ and $\bv_{-i}$.
    For every $v_i\notin\set{\tau_{i,1}(\bv_{-i}),\dots,\tau_{i,m}(\bv_{-i})}$, since allocations are integers, Lemma~\ref{lem:shifted-threshold-dsic} gives $x^\delta_i(\bv)=x_i(\bv)$ for every sufficiently small $\delta$.
    Hence, for small enough $\delta$ the shifted payment is equal to $p^\delta_i(\bv)=\tau_{i,x_i(\bv)}(\bv_{-i})+\delta\cdot\ind[x_i(\bv)\geq1]$, which tends to exactly $\tau_{i,x_i(\bv)}(\bv_{-i})$ as $\delta\to0$.
    This excludes at most $m$ values of $v_i$ (the thresholds themselves), which have probability $0$ since $F_i$ is atomless; hence the integrand of the expected payoff converges for almost every $\bv$.

    Since the integrand is indeed bounded, the claim follows immediately.
\end{proof}

The next lemma shows that an auction that is non-monotone only on a probability-zero set of opponents' profiles can be made monotone without changing its payoff.

\begin{lemma}\label{lem:mono-reduction}
    Suppose buyers are stagewise-IR, and let $\A=(\vec{x},\vec{p})$ be a deterministic, DSIC, and ex-post-IR auction such that for every buyer $i$ and almost every $\bv_{-i}$, the function $x_i(\cdot,\bv_{-i})$ is nondecreasing.
    Then $\A$ can be made monotone without changing its payoff: there is a deterministic, DSIC, ex-post-IR, and ex-post monotone auction $\A'$ with $\E[\A']=\E[\A]$.
\end{lemma}
\begin{proof}
    Obtain $\A'$ from $\A$ as follows: for every buyer $i$ and every $\vec{b}_{-i}$ for which $x_i(\cdot,\vec{b}_{-i})$ is not nondecreasing, buyer $i$ receives nothing and pays nothing for every bid $b_i$.
    Then $\A'$ is deterministic and ex-post monotone, and it is allocation-feasible, since allocations only decrease and the feasibility set is downward-closed.
    It is DSIC and ex-post IR, since for every buyer $i$ and $\vec{b}_{-i}$, the outcomes of buyer $i$ are either unchanged or all empty, and changing them does not affect the outcomes of the other buyers.
    Finally, $\A'$ differs from $\A$ only on valuation profiles $\bv$ for which, for some $i$, $\bv_{-i}$ lies in a set of probability $0$; hence $\E[\A']=\E[\A]$.
\end{proof}

The next lemma bounds the seller's expected payoff from each buyer by the expected consumer-aligned virtual value of that buyer's allocation, and shows that this bound is tight when the allocation jumps once, directly from $0$ to $m$ units.

\begin{lemma}
\label{lem:cons-aligned-reward-bound}
    Let $\A$ be a deterministic DSIC auction for consumer-aligned payoff, with stagewise-IR buyers, whose allocation rule is ex-post monotone.
Then for every $i$ and $\bv_{-i}$,
\begin{align}\label{eq:cons-aligned-reward-bound}
    \E_{v_i}\sqbr{\payoffValue(v_i)x_i(\bv)-\conPayoffPay p_i(\bv) \mid \bv_{-i}}\leq \E_{v_i}\sqbr{\varphi_{i,\Pi}^{m}(v_i)x_i(\bv) \mid \bv_{-i}}
    .
\end{align}

    In addition, if the payments are threshold payments, and there is only a single jump
    in allocation from $0$ to $m$, then inequality~\eqref{eq:cons-aligned-reward-bound} tightens to equality.
\end{lemma}

\begin{proof}[Proof of Lemma~\ref{lem:cons-aligned-reward-bound}]
Let  $\A=(\bold{x},\bold{p})$ be an auction satisfying the assumptions of
the claim.

Fix $i$ and $\bv$. The allocation rule $x_i(\cdot,\bv_{-i})$ is an integer nondecreasing step function. Denote its jump locations in $[\bottomV_i,v_i]$ by $\bottomV_i=z_1,\dots,z_\ell$ and the jump sizes by $\Delta_1,\dots,\Delta_\ell$ (respectively).
If $\bottomV_i>0$, we set $z_1=0$.
Then by Claim~\ref{clm:payment-at-least-thresh}
\begin{align*}
     p_i(\bv)\geq\sup\set{z:x_i(z,\bv_{-i})<x_i(\bv)}=z_\ell.
\end{align*}
On the other hand, Myerson's payment has
\begin{align*}
    p^{My}_i(\bv)=\sum_{j=1}^\ell z_j\Delta_j\leq z_\ell\sum_{j=1}^\ell \Delta_j\leq z_\ell m=m\cdot p_i(\bv).
\end{align*}
Hence, by Myerson~\cite{myerson1981optimal},

\begin{align}
    \E_{v_i}\sqbr{\payoffValue(v_i)x_i(\bv)-\conPayoffPay p_i(\bv) \mid \bv_{-i}}
    &\leq
    \E_{v_i}\sqbr{\payoffValue(v_i)x_i(\bv)-\tfrac{\conPayoffPay}{m} p^{My}_i(\bv) \mid \bv_{-i}}\label{eq:ineq-myerson}\\
    &=\E_{v_i}\sqbr{\payoffValue(v_i)x_i(\bv)-\tfrac{\conPayoffPay}{m} \brackets{v_i-\tfrac{1-F_i(v_i)}{f_i(v_i)}}x_i(\bv) \mid \bv_{-i}}
    \\
    &=
    \E_{v_i}\sqbr{\varphi_{i,\Pi}^{m}(v_i)x_i(\bv) \mid \bv_{-i}}.
\end{align}

  Now suppose the allocation $x_i(\cdot,\bv_{-i})$ jumps once from $0$ to $m$ at $z_\ell$, and buyer $i$ pays $z_\ell$ when allocated and $0$ otherwise. Then:
  \begin{align*}
      p_i(\bv)=\begin{cases}
          0, & x_i(\bv)=0\\
          z_\ell, & x_i(\bv)=m
      \end{cases}
      =\frac{p^{My}_i(\bv)}{m}\quad\text{a.e.}
  \end{align*}

Hence, inequality~\eqref{eq:ineq-myerson} tightens to equality and we are done.
\end{proof}

\subsection{Near-Optimality Under Concave Costs}\label{sec:cons-near-optimality}
We now prove Theorem~\ref{thm:cons-aligned-optimal-auc} and Theorem~\ref{thm:cons-surplus-optimal-auc}.
The proof shares similarities with the proof of Theorem 6.7 from~\cite{berzack2025dynamic}.
The main differences are that here there are multiple buyers and the cost function is an arbitrary concave function, rather than the specific over-time cost function for i.i.d.\ buyers in the rental game.
Throughout this subsection, the cost function $c$ is concave with $c(0)=0$, and $\bottomV_i=0$ for every buyer $i$.
We first show that the IVW allocation either allocates nothing or allocates all $m$ units to a single buyer, then use it to define the auction \consMaximizer, and finally show that the payoff of \consMaximizer converges, as $\delta\to0$, to an upper bound on the payoff of every DSIC auction.

\begin{claim}\label{clm:ivw-all-or-nothing}
    Let $\vec{x}^*$ be the IVW allocation (Definition~\ref{def:cons-aligned-ironed-maximizer}).
    For every $\bv$, either $\vec{x}^*(\bv)=\vec{0}$, or there is a buyer $i$ with $x^*_i(\bv)=m$ and $x^*_j(\bv)=0$ for all $j\neq i$.
\end{claim}
\begin{proof}
    Fix $\bv$, write $\phi_i=\irn_{i,\Pi}^{m}(v_i)$, let $\Phi=\max_i \phi_i$, and let $M=m\Phi-c(m)$.
    For any $\vec{y}\in\Omega$ with $q=\sum_i y_i$ units,
    \[
        \sum_i \phi_i y_i-c(q)\le q\Phi-c(q)\le \frac{q}{m}\,M\le \max\set{0,M},
    \]
    where the second inequality holds since $c$ is concave with $c(0)=0$, so $c(q)\geq \frac{q}{m}c(m)$.
    The allocation $\vec{0}$ attains $0$, and allocating all $m$ units to a buyer attaining $\Phi$ attains $M$; both are feasible, since every buyer can receive all $m$ units.
    Hence the maximum in~\eqref{eq:cons-aligned-argmax-def} equals $\max\set{0,M}$.

    Let $\vec{y}\neq\vec{0}$ be a maximizer with $q\geq1$ units. Then $M\geq0$ (otherwise $\vec{y}$ would have negative value), every buyer allocated by $\vec{y}$ has $\phi_i=\Phi$ (by tightness of the first inequality), and allocating all $m$ units to such a buyer $i$ is also a maximizer.
    If $\vec{y}$ is not of this form, then by the two properties of $\prec$, $\vec{y}$ is $\prec$-smaller than the concentrated allocation of $q$ units to buyer $i$, which is $\preceq$-smaller than the allocation of all $m$ units to buyer $i$.
    Hence $\vec{y}$ is not the $\prec$-maximal maximizer, and $\vec{x}^*(\bv)$ is either $\vec{0}$ or allocates all $m$ units to a single buyer.
\end{proof}

By Claim~\ref{clm:pointwise-maximizer-nondecreasing} (applied with the ironed virtual values, which are nondecreasing by Observation~\ref{obs:irn-nondecreasing}), $x^*_i(\cdot,\bv_{-i})$ is nondecreasing for every $i$ and $\bv_{-i}$.
Together with Claim~\ref{clm:ivw-all-or-nothing}, this means that each buyer has a winning threshold under $\vec{x}^*$: for every $i$ and $\vec{b}_{-i}$, let
\[
    t_i(\vec{b}_{-i})\coloneqq\inf\set{z\in[0,\topV_i]:\ x^*_i(z,\vec{b}_{-i})=m},
\]
where $\inf\emptyset=\infty$. Then $x^*_i(z,\vec{b}_{-i})=m$ for every $z>t_i(\vec{b}_{-i})$ and $x^*_i(z,\vec{b}_{-i})=0$ for every $z<t_i(\vec{b}_{-i})$.

\begin{definition}[$\delta$-shifted IVW-threshold auction]\label{def:cons-shifted-ivw}
    For $\delta>0$, the \emph{$\delta$-shifted IVW-threshold auction} \consMaximizer is the $\delta$-shifted threshold auction (Definition~\ref{def:shifted-threshold}) of the IVW allocation $\vec{x}^*$.
\end{definition}

Since $\vec{x}^*$ allocates either $0$ or $m$ units to each buyer, all its thresholds for buyer $i$ coincide: $\tau_{i,k}(\vec{b}_{-i})=t_i(\vec{b}_{-i})$ for every $k\in\set{1,\dots,m}$.
Hence in \consMaximizer, buyer $i$ receives all $m$ units if $b_i\geq t_i(\vec{b}_{-i})+\delta$, paying $t_i(\vec{b}_{-i})+\delta$ for the first unit and receiving the remaining $m-1$ units for free, and receives nothing otherwise.

\begin{remark}\label{rem:why-shift}
    The shift $\delta$ is necessary.
    Consider instead the auction with allocation $\vec{x}^*$, in which the winner pays $t_i(\vec{b}_{-i})$ for the first unit and receives the remaining units for free.
    Suppose buyer $i$ has valuation $v_i=t_i(\vec{b}_{-i})>0$ but loses, e.g., because of a tie that $\prec$ breaks in favor of another buyer.
    By bidding any $b_i>v_i$, buyer $i$ wins all $m\geq2$ units at total price $t_i(\vec{b}_{-i})=v_i$; the filter of this outcome is $v_i$, so the deviation is acceptable, and it yields utility $(m-1)v_i>0$.
    Hence this auction is not DSIC.
    In \consMaximizer, the winning region of buyer $i$ is $\set{b_i\geq t_i(\vec{b}_{-i})+\delta}$, which contains its threshold, and the positive margin guarantees that a winner strictly clears their threshold under $\vec{x}^*$, which preserves feasibility (Lemma~\ref{lem:shifted-threshold-dsic}).
\end{remark}

\begin{lemma}\label{lem:cons-shifted-dsic}
    Suppose buyers are stagewise-IR.
    For every $\delta>0$, \consMaximizer is deterministic, allocation-feasible, ex-post monotone, DSIC, and ex-post IR, and for every $\vec{b}$ at most one buyer is allocated.
    Moreover, for every $i$ and $\vec{b}_{-i}$, buyer $i$ receives $m$ units if $b_i\geq t_i(\vec{b}_{-i})+\delta$ and $0$ units otherwise, and pays $t_i(\vec{b}_{-i})+\delta$ when allocated and $0$ otherwise.
\end{lemma}
\begin{proof}
    The IVW allocation $\vec{x}^*$ is deterministic, allocation-feasible, and ex-post monotone, so the first part follows from Lemma~\ref{lem:shifted-threshold-dsic}.
    The allocation and payment of buyer $i$ follow from $\tau_{i,k}(\vec{b}_{-i})=t_i(\vec{b}_{-i})$ for every $k\geq1$.
    Finally, by Lemma~\ref{lem:shifted-threshold-dsic}, every buyer allocated in \consMaximizer is also allocated under $\vec{x}^*(\vec{b})$, which allocates at most one buyer by Claim~\ref{clm:ivw-all-or-nothing}.
\end{proof}

\begin{proof}[Proof of Theorem~\ref{thm:cons-aligned-optimal-auc} and Theorem~\ref{thm:cons-surplus-optimal-auc}]
    For brevity, write $\varphi_i=\varphi_{i,\Pi}^m$ and $\bar{\varphi}_i=\irn_{i,\Pi}^m$, and let
    \[
        B\coloneqq\E_{\bv}\sqbr{\sum_{i=1}^n \bar{\varphi}_i(v_i)x^*_i(\bv)-c\brackets{\sum_{i=1}^n x^*_i(\bv)}}.
    \]
    We show that every deterministic, DSIC, ex-post-IR auction has expected payoff at most $B$, and that the expected payoff of \consMaximizer converges to $B$ as $\delta\to0$.

    \textbf{Upper bound.}
    Let $\A'=(\vec{x}',\vec{p}')$ be a deterministic, DSIC, ex-post-IR auction.
    By Theorem~\ref{thm:cons-aligned-mono-wlog} (in the single-winner case) or Theorem~\ref{thm:cons-surplus-mono-wlog} (for consumer surplus), there is a deterministic, DSIC, ex-post-IR, and ex-post monotone auction whose expected payoff is at least $\E[\A']$; hence we may assume that the allocation rule $\vec{x}'$ of $\A'$ is ex-post monotone.
    Summing Lemma~\ref{lem:cons-aligned-reward-bound} over all buyers and taking the expectation over $\bv_{-i}$,
    \[
        \E[\A']\leq\E_{\bv}\sqbr{\sum_{i=1}^n \varphi_i(v_i)x_i'(\bv)-c\brackets{\sum_{i=1}^n x_i'(\bv)}}.
    \]
    Since $x_i'(\cdot,\bv_{-i})$ is nondecreasing, Lemma~\ref{lemma:ironing-properties} gives $\E_{v_i}[\varphi_i(v_i)x_i'(\bv)\mid\bv_{-i}]\leq\E_{v_i}[\bar{\varphi}_i(v_i)x_i'(\bv)\mid\bv_{-i}]$ for every $i$ and $\bv_{-i}$.
    Moreover, $\vec{x}^*(\bv)$ maximizes $\sum_i\bar{\varphi}_i(v_i)y_i-c(\sum_i y_i)$ pointwise over $\vec{y}\in\Omega$. Hence
    \[
        \E[\A']\leq\E_{\bv}\sqbr{\sum_{i=1}^n \bar{\varphi}_i(v_i)x_i'(\bv)-c\brackets{\sum_{i=1}^n x_i'(\bv)}}\leq B.
    \]

    \textbf{The payoff of \consMaximizer.}
    Denote by $\vec{x}^\delta$ the allocation rule of \consMaximizer.
    By Lemma~\ref{lem:cons-shifted-dsic}, for every $i$ and $\bv_{-i}$ the allocation $x^\delta_i(\cdot,\bv_{-i})$ jumps once from $0$ to $m$ at $t_i(\bv_{-i})+\delta$, and buyer $i$ pays exactly $t_i(\bv_{-i})+\delta$ when allocated (if $t_i(\bv_{-i})+\delta>\topV_i$, buyer $i$ is never allocated, and both sides of~\eqref{eq:cons-aligned-reward-bound} are $0$).
    Hence Lemma~\ref{lem:cons-aligned-reward-bound} holds with equality, and
    \[
        \E[\consMaximizer]=\E_{\bv}\sqbr{\sum_{i=1}^n \varphi_i(v_i)x^\delta_i(\bv)-c\brackets{\sum_{i=1}^n x^\delta_i(\bv)}}.
    \]
    By Lemma~\ref{lem:shifted-threshold-dsic}, $x^\delta_i(\bv)\to x^*_i(\bv)$ as $\delta\to0$ whenever $v_i\neq t_i(\bv_{-i})$, since all thresholds of buyer $i$ under $\vec{x}^*$ equal $t_i(\bv_{-i})$.
    For each $i$ and $\bv_{-i}$ this excludes a single value of $v_i$, so $\vec{x}^\delta\to\vec{x}^*$ almost everywhere.
    The integrand is bounded by $m\sum_i|\varphi_i(v_i)|+c(m)$, which is integrable, so by dominated convergence,
    \[
        \lim_{\delta\to0}\E[\consMaximizer]=\E_{\bv}\sqbr{\sum_{i=1}^n \varphi_i(v_i)x^*_i(\bv)-c\brackets{\sum_{i=1}^n x^*_i(\bv)}}.
    \]
    Finally, $\vec{x}^*$ depends on $v_i$ only through $\bar{\varphi}_i(v_i)$, so $x^*_i(\cdot,\bv_{-i})$ is constant on every interval on which $\bar{\varphi}_i$ is constant, and the equality condition of Lemma~\ref{lemma:ironing-properties} gives $\E_{v_i}[\varphi_i(v_i)x^*_i(\bv)\mid\bv_{-i}]=\E_{v_i}[\bar{\varphi}_i(v_i)x^*_i(\bv)\mid\bv_{-i}]$.
    Hence the limit equals $B$.

    Combining both parts, the supremum of the expected payoff over deterministic, DSIC, ex-post-IR auctions is $B$, and for every $\varepsilon>0$ there is $\delta>0$ such that \consMaximizer, which is DSIC and ex-post IR by Lemma~\ref{lem:cons-shifted-dsic}, has expected payoff at least $B-\varepsilon$, i.e., is $\varepsilon$-optimal among deterministic, DSIC, ex-post-IR auctions.
\end{proof}

\subsection{Proof that Monotonicity is W.L.O.G.\ for Consumer-Aligned Single-Winner Auctions}\label{sec:proof-cons-monotonicity-conditions}

We next justify the monotonicity reduction used in the optimality proof above.
As in Section~\ref{sec:overview-rev-aligned}, the idea is to rearrange the allocation rule of a non-monotone auction into a monotone one, but two obstacles arise.
First, even in a single-winner auction we cannot rearrange the allocation of one buyer in isolation from the rest.
For example, if in a single-item auction buyer $1$ gets the item when $v_1<2$ and does not get it when $v_1\geq 2$,
then buyer $2$ might get the item exactly when $v_1\geq 2$, so moving buyer $1$'s allocation to higher valuations could create two winners.
Second, the distribution-preserving joint rearrangement of outcomes from Section~\ref{sec:overview-rev-aligned} only guarantees a BIC auction, and here we want to construct a DSIC auction.

Claim~\ref{clm:start-win-always-win} rules out the example above: with stagewise-IR buyers, once a buyer starts winning, they keep winning at every higher valuation.
This alone does not make the auction monotone, since we could still have $x_i(v_i,\bv_{-i})>x_i(v_i',\bv_{-i})>0$ for $v_i<v_i'$.
However, since the auction is single-winner, all other buyers receive nothing whenever buyer $i$ wins,
so we can rearrange buyer $i$'s allocations within their winning set of valuations without affecting the other buyers.
The construction therefore rearranges each buyer's allocation function into a nondecreasing one, while preserving the event that the buyer wins; this preservation is what keeps the single-winner feasibility constraint intact.
We do this in a measure-preserving way, which allows us to charge threshold payments that keep buyer $i$ truthful and do not increase the expected payment.

\begin{lemma}\label{lem:sw-rearrangement}
    Suppose the payoff is consumer-aligned, buyers are stagewise-IR, and the environment is single-winner, and let $\A=(\vec{x},\vec{p})$ be a deterministic, DSIC, and ex-post-IR auction.
    Then there is a deterministic, allocation-feasible, ex-post monotone allocation rule $\vec{x}'$ with $\textsf{threshPayoff}(\vec{x}')\geq\E[\A]$, where the inequality is strict if, for some buyer $i$, the set of $\bv_{-i}$ for which $x_i(\cdot,\bv_{-i})$ is not nondecreasing has positive probability.
\end{lemma}
\begin{proof}
    We first construct $\vec{x}'$, then compare its threshold payoff with $\E[\A]$, and finally show that the inequality is strict in the stated case.
    Fix $i$ and $\bv_{-i}$, and write $x(z)=x_i(z,\bv_{-i})$ and $p(z)=p_i(z,\bv_{-i})$.

    \paragraph{The allocation rule $\vec{x}'$.}
    Let $W=\set{z : x(z)>0}$ be the set of valuations with which buyer $i$ is allocated.
    By Claim~\ref{clm:start-win-always-win}, $W$ is either empty or an interval of the form $[\alpha,\topV_i]$ or $(\alpha,\topV_i]$.
    Let $x'(z)\coloneqq0$ for $z\notin W$, and on $W$, let $x'$ be the nondecreasing rearrangement of $x$ on $W$.
    Since $x\geq1$ on $W$, $x'$ takes values in $\set{1,\dots,m}$ on $W$; and since it is a nondecreasing step function, we may (by changing it at finitely many points) take it to be right-continuous on $W$.
    Finally, let $x'_i(\bv)\coloneqq x'(v_i)$.

    By construction, $x'$ is nondecreasing on $[\bottomV_i,\topV_i]$, so $\vec{x}'$ is ex-post monotone.
    Moreover, $x'$ is equimeasurable with $x$.
    In particular, $\Pr_{v_i}[x'(v_i)=t]=\Pr_{v_i}[x(v_i)=t]$ for every $t$.
    Since $x'_i(\bv)>0$ if and only if $x_i(\bv)>0$, for every $\bv$ the buyers allocated under $\vec{x}'(\bv)$ are exactly those allocated under $\vec{x}(\bv)$; as $\A$ is single-winner, so is $\vec{x}'$, and hence $\vec{x}'$ is allocation-feasible.

    \paragraph{Comparing $\textsf{threshPayoff}(\vec{x}')$ with $\E[\A]$.}
    We compare the contribution of buyer $i$ to each payoff term, conditioned on $\bv_{-i}$.

    \emph{Welfare.} Since $\payoffValue$ is nondecreasing and $x'$ is a nondecreasing rearrangement of $x$, the Hardy--Littlewood inequality~\cite{LiebLoss2001} gives $\E_{v_i}[\payoffValue(v_i)x'(v_i)]\geq\E_{v_i}[\payoffValue(v_i)x(v_i)]$.

    \emph{Payments.}
    Let $\tau'_t$ be the thresholds of $x'$ from Definition~\ref{def:shifted-threshold}, and let $\sigma_t\coloneqq\sup\set{z : x(z)<t}$, where $\sup\emptyset=0$.
    We claim that $\tau'_t\leq\sigma_t$ for every $t\geq1$ with $\Pr_{v_i}[x(v_i)=t]>0$.
    If $x'(\bottomV_i)\geq t$, then $\tau'_t=0\leq\sigma_t$.
    Otherwise, the set $\set{z : x(z)<t}$ is nonempty (using right-continuity of $x'$).
    Since $x'$ is nondecreasing, $[\bottomV_i,\tau'_t)\subseteq\set{z : x'(z)<t}$, and since $\set{z : x(z)<t}\subseteq[\bottomV_i,\sigma_t]$,
    \[
        F_i(\tau'_t)\leq\Pr_{v_i}\sqbr{x'(v_i)<t}=\Pr_{v_i}\sqbr{x(v_i)<t}\leq F_i(\sigma_t),
    \]
    so $\tau'_t\leq\sigma_t$, as $F_i$ is strictly increasing on $[\bottomV_i,\topV_i]$.
    By Claim~\ref{clm:payment-at-least-thresh}, $p(z)\geq\sigma_t$ whenever $x(z)=t$, and $p\geq0$. Hence
    \begin{equation}\label{eq:mono-wlog-payments}
        \E_{v_i}\sqbr{\tau'_{x'(v_i)}}=\sum_{t=1}^m\tau'_t\Pr_{v_i}\sqbr{x'(v_i)=t}\leq\sum_{t=1}^m\sigma_t\Pr_{v_i}\sqbr{x(v_i)=t}\leq\E_{v_i}\sqbr{p(v_i)},
    \end{equation}
    where terms with $\Pr_{v_i}[x(v_i)=t]=0$ are omitted from both sums.

    \emph{Cost.}
    Since the environment is single-winner and $c(0)=0$, for every $\bv$ we have $c(\sum_k x_k(\bv))=\sum_k c(x_k(\bv))$, as at most one term on the right is nonzero; similarly for $\vec{x}'$.
    Moreover, for every $k$ and $\bv_{-k}$, the functions $x_k(\cdot,\bv_{-k})$ and $x'_k(\cdot,\bv_{-k})$ are equimeasurable. Hence
    \begin{align*}
        \E_{\bv}\sqbr{c\brackets{\sum_k x'_k(\bv)}}
        &=\sum_k\E_{\bv_{-k}}\sqbr{\E_{v_k}\sqbr{c(x'_k(v_k,\bv_{-k}))}}\\
        &=\sum_k\E_{\bv_{-k}}\sqbr{\E_{v_k}\sqbr{c(x_k(v_k,\bv_{-k}))}}
        =\E_{\bv}\sqbr{c\brackets{\sum_k x_k(\bv)}}.
    \end{align*}

    Since $\conPayoffPay>0$, combining the three terms, taking the expectation over $\bv_{-i}$, and summing over all buyers yields $\textsf{threshPayoff}(\vec{x}')\geq\E[\A]$.

    \paragraph{Strict improvement.}
    Suppose $x$ is not nondecreasing, but equality holds
    in~\eqref{eq:mono-wlog-payments}.
    Then $\sigma_t=\tau'_t$ for every level $t$ of positive probability.
    For every such $t$, both $\set{w:x(w)<t}$ and
    $\set{w:x'(w)<t}$ are contained in $[\bottomV_i,\tau'_t]$
    and have probability $F_i(\tau'_t)$, so they agree almost everywhere.
    Since both functions take levels of positive probability almost surely,
    $x=x'$ almost everywhere.

    Choose $y<z$ with $x(y)>x(z)$, and
    $\varepsilon\in(0,z-y)$ such that $x'$ is constant,
    say equal to $j$, on $(y,y+\varepsilon)$.
    Since $y\in W$, we have $j\geq1$ and $\tau'_j\leq y$.
    Moreover, $j$ occurs with positive probability, so
    $\sigma_j=\tau'_j\leq y<z$, implying $x(z)\geq j$.
    Thus, for almost every $w\in(y,y+\varepsilon)$,
    \[
        x(w)=j\leq x(z)<x(y).
    \]
    Claim~\ref{clm:no-upward-imitation} gives
    $p(w)>y\geq\sigma_j$ on this positive-probability set,
    contradicting equality in~\eqref{eq:mono-wlog-payments}.

    Hence the expected payment strictly decreases on every non-monotone
    slice. If such slices have positive probability for some buyer $i$,
    averaging over $\bv_{-i}$ gives a strict decrease in expected payments.
    Together with the comparisons above and
    $\conPayoffPay>0$, this yields
    $\textsf{threshPayoff}(\vec{x}')>\E[\A]$.
\end{proof}

We are now ready to prove Theorem~\ref{thm:cons-aligned-mono-wlog}.

\begin{proof}[Proof of Theorem~\ref{thm:cons-aligned-mono-wlog}]
    Let $\A=(\vec{x},\vec{p})$ be a deterministic, DSIC, ex-post-IR auction in a single-winner environment.
    If for every buyer $i$ and almost every $\bv_{-i}$ the function $x_i(\cdot,\bv_{-i})$ is nondecreasing, then Lemma~\ref{lem:mono-reduction} gives a monotone auction $\A'$ with $\E[\A']=\E[\A]$.
    Otherwise, Lemma~\ref{lem:sw-rearrangement} gives a deterministic, allocation-feasible, ex-post monotone allocation rule $\vec{x}'$ with $\textsf{threshPayoff}(\vec{x}')>\E[\A]$, so by Lemma~\ref{lem:shifted-threshold-payoff} there is $\delta>0$ for which the $\delta$-shifted threshold auction $\A'$ of $\vec{x}'$ satisfies $\E[\A']>\E[\A]$.
    By Lemma~\ref{lem:shifted-threshold-dsic}, $\A'$ is deterministic, allocation-feasible, ex-post monotone, DSIC, and ex-post IR.
\end{proof}

\subsection{Proof that Monotonicity is W.L.O.G.\ for Consumer-Surplus Auctions}
Our proof builds on an idea from the proof of Lemma 5.7 in~\cite{berzack2025dynamic}, which shows that every optimal single-buyer stagewise auction (SWAC) maximizing consumer surplus is monotone a.e.
The key property of consumer surplus is that if a buyer switches to an outcome with fewer units because it increases their utility, then the seller's payoff from this buyer increases as well, since the utility increases and the cost does not.
In~\cite{berzack2025dynamic}, this is used to fix a single violation of monotonicity, by lowering the filter of one outcome, which suffices to contradict optimality.
Here we fix all violations at once: every valuation $v_i$ receives the smallest allocation of any valuation $z\geq v_i$, with threshold payments.
Since allocations only decrease, allocation-feasibility is preserved with multiple buyers, and the resulting allocation rule is monotone.

\begin{proof}[Proof of Theorem~\ref{thm:cons-surplus-mono-wlog}]
    Let $\A=(\vec{x},\vec{p})$ be a deterministic, DSIC, and ex-post-IR auction, where the payoff is consumer surplus, i.e., $\payoffValue(v)=v$ and $\conPayoffPay=1$.
    For every $i$ and $\bv$, let $U_i(\bv)\coloneqq v_ix_i(\bv)-p_i(\bv)$ be the utility of buyer $i$ when all buyers bid truthfully, so that $\E[\A]=\E_{\bv}[\sum_iU_i(\bv)-c(\sum_ix_i(\bv))]$.
    As in the proof of Theorem~\ref{thm:cons-aligned-mono-wlog}, with the following construction in place of Lemma~\ref{lem:sw-rearrangement}, it suffices to construct a deterministic, allocation-feasible, ex-post monotone allocation rule $\vec{x}'$ with $\textsf{threshPayoff}(\vec{x}')\geq\E[\A]$, where the inequality is strict if, for some buyer $i$, the set of $\bv_{-i}$ for which $x_i(\cdot,\bv_{-i})$ is not nondecreasing has positive probability.

    \paragraph{The allocation rule $\vec{x}'$.}
    For every $i$ and $\bv$, let
    \[
        x'_i(\bv)\coloneqq\min\set{x_i(z,\bv_{-i}) : z\in[v_i,\topV_i]},
    \]
    which exists since $x_i$ takes finitely many values.
    Then $x'_i(\cdot,\bv_{-i})$ is nondecreasing, and $\vec{x}'\leq\vec{x}$ coordinatewise; since the feasibility set is downward-closed, $\vec{x}'$ is allocation-feasible.
    It is deterministic and ex-post monotone by construction.

    Fix $i$ and $\bv_{-i}$, and write $x(z)=x_i(z,\bv_{-i})$, $x'(z)=x'_i(z,\bv_{-i})$, and $p(z)=p_i(z,\bv_{-i})$.
    Let $\tau'_k$ be the thresholds of $x'$ from Definition~\ref{def:shifted-threshold}, let $\tau'_0\coloneqq0$, and let $\sigma_k\coloneqq\sup\set{z : x(z)<k}$, where $\sup\emptyset=0$.
    We claim that $\tau'_k=\sigma_k$ for every $k\geq1$ attained by $x'$.
    Indeed, $x'(y)<k$ if and only if $x(z)<k$ for some $z\geq y$.
    Hence if $x'(\bottomV_i)\geq k$, then $\set{z : x(z)<k}=\emptyset$, and $\tau'_k=0=\sigma_k$.
    Otherwise, $\set{y : x'(y)<k}$ is nonempty and consists of the points that lie below some point of $\set{z : x(z)<k}$; since $x'$ is nondecreasing, $\tau'_k=\inf\set{y : x'(y)\geq k}=\sup\set{y : x'(y)<k}=\sup\set{z : x(z)<k}=\sigma_k$.

    \paragraph{A pointwise bound on buyer utilities.}
    Fix $v_i$, let $K=x'(v_i)$, and let $z\geq v_i$ with $x(z)=K$.
    Since $\A$ is ex-post IR, the truthful outcome of buyer $i$ at $(v_i,\bv_{-i})$ is acceptable, so its filter is at most $v_i\leq z$; hence a buyer with valuation $z$ can bid $v_i$, and since $\A$ is DSIC,
    \[
        U_i(z,\bv_{-i})\geq u_i((v_i,\bv_{-i});z)=x(v_i)z-p(v_i)=U_i(\bv)+x(v_i)(z-v_i).
    \]
    On the other hand, $U_i(z,\bv_{-i})=Kz-p(z)\leq Kz-\tau'_K$: for $K\geq1$, Claim~\ref{clm:payment-at-least-thresh} gives $p(z)\geq\sigma_K=\tau'_K$, and for $K=0$ we use $p(z)\geq0=\tau'_0$.
    Combining the two bounds,
    \begin{equation}\label{eq:cs-mono-pointwise}
        U_i(\bv)\leq Kv_i-\tau'_K-\brackets{x(v_i)-K}(z-v_i)\leq x'_i(\bv)v_i-\tau'_{x'_i(\bv)}.
    \end{equation}
    Moreover, if $x(v_i)>x'(v_i)$, then $z\neq v_i$, so $z>v_i$ and the second inequality in~\eqref{eq:cs-mono-pointwise} is strict.

    \paragraph{Comparing $\textsf{threshPayoff}(\vec{x}')$ with $\E[\A]$.}
    Summing~\eqref{eq:cs-mono-pointwise} over all buyers, and using that $c$ is nondecreasing and $\vec{x}'\leq\vec{x}$,
    \begin{align*}
        \E[\A]=\E_{\bv}\sqbr{\sum_iU_i(\bv)-c\brackets{\sum_ix_i(\bv)}}
        &\leq\E_{\bv}\sqbr{\sum_i\brackets{x'_i(\bv)v_i-\tau'_{i,x'_i(\bv)}(\bv_{-i})}-c\brackets{\sum_ix'_i(\bv)}}\\
        &=\textsf{threshPayoff}(\vec{x}'),
    \end{align*}
    where $\tau'_{i,k}(\bv_{-i})$ denotes the threshold $\tau'_k$ defined above for buyer $i$ and $\bv_{-i}$.

    \paragraph{Strict improvement.}
    Fix $i$ and $\bv_{-i}$ for which $x$ is not nondecreasing; we show that~\eqref{eq:cs-mono-pointwise} is strict for a set of $v_i$ of positive probability.
    If this happens for a set of $\bv_{-i}$ of positive probability, we get $\textsf{threshPayoff}(\vec{x}')>\E[\A]$.

    If $\Pr_{v_i}[x(v_i)>x'(v_i)]>0$, this was shown above.
    Otherwise, $x=x'$ almost everywhere.
    Take $y<z$ with $x(y)>x(z)$, and $\varepsilon\in(0,z-y]$ such that $x'$ is constant, say equal to $K$, on $(y,y+\varepsilon)$.
    Then $K\leq x'(z)\leq x(z)<x(y)$, and $x\geq x'\geq K$ on $(y,\topV_i]$, so $\tau'_K=\sigma_K\leq y$.
    For almost every $w\in(y,y+\varepsilon)$, we have $x(w)=K<x(y)$, so $p(w)>y$ by Claim~\ref{clm:no-upward-imitation}, and hence
    \[
        U_i(w,\bv_{-i})=Kw-p(w)<Kw-\tau'_K=x'(w)w-\tau'_{x'(w)},
    \]
    that is,~\eqref{eq:cs-mono-pointwise} is strict.
\end{proof}

\section{Conclusion and Future Research Directions}\label{sec:conclusion}

    In this paper we initiate the systematic study of auction design with constrained buyers, and give Myerson-like results showing that for revenue-aligned auctions, when the constraints are uncapped and monotone (UM), the designer does not benefit from exploiting the constraints of the buyers.
    For consumer-aligned payoff, our main bottom-line contribution is an almost optimal mechanism for multi-buyer auctions with a single winner and concave cost.
    Perhaps the most interesting open question (beyond the relaxation of the  conditions above) is to ask whether there is a comprehensive theory of auctions for buyers with constraints that are \emph{not uncapped}, for the full range of objectives we consider.
    For example, consider buyers with budget or return-on-investment (RoI) constraints, and the consumer surplus objective - what is the optimal auction? We believe that our framework and techniques may possibly serve as the foundation for studying this direction.

\section*{Acknowledgements}
\ifanonymous
\else
\addcontentsline{toc}{section}{Acknowledgements}

We thank Yannai Gonczarowski and Amos Fiat for helpful discussions and comments.
This work
received funding from the European Research Council (ERC) under the European Union’s Horizon
Europe research and innovation program (grant No.: 101077862), from the Israel Science Foundation
(grant No.: 3331/24, 2801/20 and 3725/24), from the NSF-BSF (grant No.: 2021680), and from a Google Research Scholar
Award.

\fi

\bibliographystyle{siam}
\bibliography{references}

\appendix
\addtocontents{toc}{\protect\setcounter{tocdepth}{1}}

\section*{Appendix Organization}
\begin{itemize}
    \item Appendix~\ref{sec:additional-related}: Additional related work.

    \item Appendix~\ref{app:auxiliary}: Auxiliary definitions and examples.
    \item Appendix~\ref{app:ironing}: Virtual values and ironing.

    \item Appendix~\ref{app:randomized-relabeling-truthification}: Randomized relabeling and partial truthification.

    \item Appendix~\ref{app:rev-aligned}: Remaining details for Section~\ref{sec:revenue} on revenue-aligned objectives.
\end{itemize}

\section{Additional Related Work}\label{sec:additional-related}

\subsection{Related Work on Buyer Constraints}

\paragraph{The no-overbidding \emph{assumption} and ex-post individually rational buyers.} The assumption that bidders do not or cannot bid above their valuations finds significant precedent in the auction literature.
Buyers that do not overbid are sometimes referred to as \emph{conservative bidders}.
In the Price of Anarchy line of work, assuming that buyers do not overbid is a very common and essential assumption, beginning with~\cite{christodoulou2008bayesian}, and continuing in many papers such as~\cite{bhawalkar2011welfare,leme2010pure,lucier2010price,roughgarden2017price}.
There are multiple justifications for this assumption. One of them, relied on by~\cite{lucier2010price,christodoulou2008bayesian}, assumes that agents are ``ex-post individually rational'', meaning they will never make decisions that might (in some scenario) result in them getting negative utility. With these types of agents, the mechanism can utilize this assumption to prevent agents from overbidding. Another justification, given for example by~\cite{lucier2010price,leme2009sponsored}, involves designing a mechanism that inserts a random bid with vanishing probability, allowing the application of trembling-hand techniques~\cite{bielefeld1988reexamination} to show that agents will be conservative.
Beyond the Price of Anarchy literature, \cite{ahunbay2020two} study a two-buyer sequential auction and make a behavioral no-overbidding assumption, motivated by the fact that in practice buyers often do not overbid, as they may not have perfect trust in the seller.

\paragraph{The no-overbidding \emph{constraint} and verifiable bids.}
A related line of work studies mechanism design with verifiable bids, where some misreports are ruled out exogenously. Early work by \cite{green1986partially} shows that partial verification can expand the designer's power and mainly characterizes when the revelation principle continues to hold (this does not cover all of our constraints, see Example~\ref{ex:failure-of-revelation-principle}). This perspective naturally connects to auctions with a no-overbidding constraint, where reports above the true valuation are infeasible. Recent work by \cite{reuter2023revenue} studies revenue maximization with partially verifiable reports in single-item auctions, where buyers report an auxiliary characteristic correlated with value that can be partially checked, and shows that verification can strictly improve revenue over Myerson in some settings. Relatedly, \cite{celik2006mechanism} shows that in single-agent screening problems, if the full-information optimal action is monotone, then optimizing subject only to downward (unidirectional) incentive constraints yields the same outcome as full incentive compatibility; however, their argument relies on the absence of feasibility coupling and does not extend to multi-buyer auctions.

In a similar manner, auctions with budget-constrained buyers, who cannot over-report their budgets, have been considered (e.g., \cite{devanur2017optimal,che2000optimal}). They call these \emph{conditional} mechanisms and show that in some cases the optimal conditional mechanism can yield more revenue than the optimal \emph{unconditional} mechanism. They also explain how this can be enforced, for example, by requiring the buyer to pay their full budget with some small probability. This demonstrates that certain reports are effectively \emph{impossible} for some bidders to make.

\paragraph{The stagewise individually rational constraint.} \cite{krahmer2025dynamic} consider a dynamic screening procurement auction where the agent is short-term liquidity constrained, which is equivalent to stagewise-IR buyers. In addition,~\cite{berzack2025dynamic} define stagewise-IR buyers exactly as we do, although the definition in the multi-buyer case is slightly more delicate. They show that, for some objectives, the designer can utilize the buyers' constraints to improve their outcome. We strictly improve upon the results of~\cite{berzack2025dynamic} even in the single-buyer case. For a detailed comparison, see Appendix~\ref{app:comparison-to-rental}.

\subsection{Related Work on Incorporating Buyer Constraints Into the Utility Function}\label{app:prior-papers-that failed}\hfill

Several past works have explicitly incorporated specific buyer constraints into the utility function, while implicitly assuming that Myerson's classical theory still applies. As we have shown, this assumption is incorrect,
but does not in itself invalidate their results.
One of our goals is to provide theoretical underpinnings for correct designs that incorrectly rely on Myerson's theory.

For example,~\cite{dobzinski2014efficiency,hirai2025polyhedral} study budget-constrained buyers and explicitly say that the buyers' utility is $-\infty$ in case of a payment that violates their budget. However, the rest of their analysis continues with respect to the standard unconstrained utility, assuming, for example, that all truthful auctions are monotone, or that if buyers do not deviate, it is because their quasilinear utility would not increase.
However, under the constrained utility model, if buyers are constrained, their decision not to deviate may be due to their constraint rather than their quasilinear utility.
Similarly,~\cite{lv2023auction,tang2024towards} do the same for RoI-constrained buyers.

The work of \cite{golrezaei2021auction} considers buyers who are not prepared to participate if they do not get a high enough RoI, but does not explicitly incorporate this constraint into the buyer's utility; instead, the buyers are treated as unconstrained, and at equilibrium the mechanism must ensure that RoI constraints are met.

\subsection{Related Work on General Seller Objectives}

\paragraph{Consumer surplus.}
Hartline and Roughgarden~\cite{HartlineR08} initiated the algorithmic study of consumer surplus maximization, focusing on single-parameter settings. Recent work studies this objective in multi-parameter, combinatorial, online, and learning-augmented settings~\cite{goldner2025multidimensional,goldner2026knowing,ezra2025multi,fotakis2016efficient,ganesh2023combinatorial}, most recently obtaining constant-factor Bayesian approximations for unit-demand buyers~\cite{goldner2026constant} and tight prior-free approximations in general multi-dimensional environments~\cite{ezra2026optimal}.
Berzack et al.~\cite{berzack2025dynamic} were the first to combine consumer surplus maximization with constrained buyers. They considered stagewise-IR agents and designed a mechanism that utilizes this constraint to extract higher consumer surplus (among other objectives). See Appendix~\ref{app:comparison-to-rental} for a detailed comparison to our results. Very recently, Eden et al.~\cite{alon2026consumer} study a more general objective that includes consumer surplus maximization as a special case in the context of budget-feasible procurement auctions.

\paragraph{Other payoff functions.}
Many papers have considered maximizing both welfare and revenue simultaneously with unconstrained buyers, which is included in our revenue-aligned tradeoff, such as~\cite{likhodedov2003mechanism,diakonikolas2012efficiency}. Even more general objectives have been considered by~\cite{berzack2025dynamic} and~\cite{feldman2025budget} (what they term BEST objectives), the latter studying budget-feasible contracts.

\subsection{Comparison to Previous Work on Rental Games~\cite{berzack2025dynamic}}\label{app:comparison-to-rental}

Our results greatly expand upon the results of~\cite{berzack2025dynamic}.
In addition to providing results for a much more general model, our proofs unite many objectives and our techniques are more robust:
they are not merely a generalization of those of~\cite{berzack2025dynamic}, since the single-buyer monotonicity-repair techniques are inapplicable in the multi-buyer setting or when the feasibility set is restricted.
This motivates a re-examination of the problem from first principles. Instead of relying on objective-specific arguments as in~\cite{berzack2025dynamic}, we develop a unified proof framework that applies across several objectives. As a result, we obtain stronger optimality guarantees than~\cite{berzack2025dynamic} even in the single-buyer setting, establishing global optimality rather than optimality limited to finite-menu auctions.

The core extension is from single-buyer auctions to multi-buyer settings. This introduces the primary technical challenge, as feasibility constraints become coupled across buyers, substantially restricting the modifications that can be made to an auction. This difficulty is further amplified by our allowance of arbitrary feasibility sets, which imposes additional constraints on how allocations can be altered.

\subsubsection{Revenue-Aligned Results Comparison}
Our revenue-aligned payoff objective combines the revenue-like, welfare-like, and positive tradeoff objectives from~\cite{berzack2025dynamic}.

The main differences in our results are highlighted in Table~\ref{tab:rev-aligned}.
The rows of the table should be read as follows.
The \emph{buyer constraints} row specifies which hard restrictions on acceptable bids or outcomes are covered:~\cite{berzack2025dynamic} studies stagewise-IR buyers, while our revenue-aligned results apply to any uncapped-payment monotone (UM) constraint.
The \emph{buyers} row records whether the result is single-agent or multi-agent.
The \emph{feasibility set} row describes which allocations the seller may choose from: $\Omega_{\mathrm{all}} = \set{ (x_1,\ldots,x_n) : \sum_{i = 1}^n x_i = m}$ denotes the unrestricted allocation set for $m$ units, whereas our results allow an arbitrary feasibility set $\Omega$.
The \emph{optimality} row distinguishes guarantees that hold only within a finite-menu subclass from guarantees that are global over the relevant auction class.
Finally, the \emph{auction space} row says whether the theorem is restricted to deterministic DSIC auctions or also covers auctions with lotteries and Bayesian incentive compatibility (BIC).

\begin{table}[h]
\centering
\begin{tabular}{lcc}
\hline
\textbf{Dimension} & \textbf{Berzack et al.\ (2025)} & \textbf{This paper} \\
\hline
Buyer constraints & Stagewise IR & Any uncapped-monotone \\
Buyers & Single-buyer & Multi-buyer \\
Feasibility set & Only $\Omega_{\mathrm{all}}$ & Any $\Omega$ \\
Optimality & Some restricted to finite-menu & All global \\
Auction space & Only deterministic, DSIC & Including lotteries, BIC \\
\hline
\end{tabular}
\caption{Comparison for the revenue-aligned case}\label{tab:rev-aligned}
\end{table}

\subsubsection{Consumer-Aligned Results Comparison}

Our consumer-aligned payoff is slightly more general than the negative tradeoff objective from~\cite{berzack2025dynamic}.

The main differences in our results are highlighted in Table~\ref{tab:cons-aligned}.
Here the comparison focuses on the stagewise-IR setting studied by~\cite{berzack2025dynamic}.
The \emph{buyers} row indicates that we move from the original single-buyer rental game to settings with multiple strategic buyers.
The \emph{optimality} row has the same meaning as in Table~\ref{tab:rev-aligned}: our guarantee is global rather than restricted to finite-menu auctions.
The \emph{robustness} row describes the cost functions for which the mechanism is proved optimal; instead of relying on the particular rental-game cost functions from~\cite{berzack2025dynamic}, our proof works for any concave cost.

\begin{table}[h]
\centering
\begin{tabular}{lcc}
\hline
\textbf{Dimension} & \textbf{Berzack et al.\ (2025)} & \textbf{This paper} \\
\hline
Buyers & Single-buyer & Multi-buyer \\
\hline
Optimality & Finite-menu & All global \\
\hline
Robustness & Specific cost functions & Any concave cost \\
\hline
\end{tabular}
\caption{Comparison for the consumer-aligned case}\label{tab:cons-aligned}
\end{table}

While~\cite{berzack2025dynamic} find an optimal auction for specific cost functions derived from the rental game with i.i.d.\ buyers, we prove that this mechanism form is optimal for \emph{any} concave cost function in the multi-buyer settings covered by our consumer-aligned theorem. This decouples the mechanism's validity from the specific ``rental'' interpretation and establishes it as a general solution for consumer-aligned objectives in these settings.

\section{Auxiliary Definitions and Examples}\label{app:auxiliary}

\subsection{Common Preliminaries}\label{app:common}

\begin{definition}[Nondecreasing rearrangement, simplified from~\cite{LiebLoss2001}]
\label{def:nondecreasing-rearrangement}
Let $g:X\to Y$ be a measurable function, with $X,Y\subseteq\reals_{\geq0}$.
A measurable function \(g^\uparrow:X\to Y\) is called a
\emph{nondecreasing rearrangement} of $g$ with respect to distribution $\mu$ if:

\begin{enumerate}
    \item \(g^\uparrow\) is nondecreasing on \(X\), and
    \item \(g^\uparrow\) is equimeasurable with \(g\) with respect to $\mu$, i.e., for every \(t\in Y\),
    \[
    \Pr_{u\sim\mu}\!\big[g^\uparrow(u)\ge t\big]
    =
    \Pr_{u\sim \mu}\!\big[g(u)\ge t\big].
    \]
\end{enumerate}
Such a nondecreasing rearrangement exists. Moreover, it is unique up to \(\mu\)-almost-everywhere
equality: if \(\tilde g^\uparrow\) is another nondecreasing rearrangement of \(g\), then
\[
g^\uparrow(u)=\tilde g^\uparrow(u)
\qquad\text{for }\mu\text{-almost every }u.
\]
Throughout, when we speak of \emph{the} nondecreasing rearrangement,
we take this version so that
\[
g^\uparrow\big(X\big)\subseteq
\overline{\,g\big(X\big)\,},
\]
and hence if \(g^\uparrow(u)\notin g(X)\), then
\(g^\uparrow(u)\) is necessarily a limit point of the image of \(g\).
\end{definition}

\subsection{Examples}\label{app:examples}

\subsubsection{Implementing the Toy Auction with No-Overbidding Buyers}
\label{app:warmup-no-overbidding-payments}

Consider the allocation rule from the toy example in Section~\ref{sec:toy-example}.
Assume both buyers have values in $[1,2]$ and are subject to no-overbidding
constraints, so a buyer with value $v_i$ may only submit bids $b_i\leq v_i$.
The allocation and payment rules are shown in Table~\ref{tab:warmup-no-overbidding-implementation};
entries in the two tables are
$(\textcolor{blue}{x_1},\textcolor{green!50!black}{x_2})$ and
$(\textcolor{blue}{p_1},\textcolor{green!50!black}{p_2})$, respectively.

\begin{table}[h]
\centering
\begin{subtable}{0.48\linewidth}
\centering
\[
\begin{array}{c|c|c}
 & \textcolor{blue}{b_1<1.9} & \textcolor{blue}{b_1\geq 1.9} \\ \hline
\textcolor{green!50!black}{b_2<1.6}
& (\textcolor{blue}{2},\textcolor{green!50!black}{0})
& (\textcolor{blue}{1},\textcolor{green!50!black}{1})
\\[1mm]
\textcolor{green!50!black}{b_2\geq 1.6}
& (\textcolor{blue}{0},\textcolor{green!50!black}{2})
& (\textcolor{blue}{1},\textcolor{green!50!black}{1})
\end{array}
\]
\caption{Allocation rule}
\label{tab:warmup-no-overbidding-allocation}
\end{subtable}
\hfill
\begin{subtable}{0.48\linewidth}
\centering
\[
\begin{array}{c|c|c}
 & \textcolor{blue}{b_1<1.9} & \textcolor{blue}{b_1\geq 1.9} \\ \hline
\textcolor{green!50!black}{b_2<1.6}
& (\textcolor{blue}{2},\textcolor{green!50!black}{0})
& (\textcolor{blue}{0},\textcolor{green!50!black}{0})
\\[1mm]
\textcolor{green!50!black}{b_2\geq 1.6}
& (\textcolor{blue}{0},\textcolor{green!50!black}{0})
& (\textcolor{blue}{0},\textcolor{green!50!black}{0})
\end{array}
\]
\caption{Payment rule}
\label{tab:warmup-no-overbidding-payments}
\end{subtable}
\caption{The toy auction and a payment rule that implements it under no-overbidding constraints.}
\label{tab:warmup-no-overbidding-implementation}
\end{table}

We verify that truthful bidding is DSIC with
respect to the constrained utility. Fix the other buyer's bid. For buyer~$1$,
if $v_1<1.9$, then every feasible bid satisfies $b_1<1.9$, so every feasible bid
induces the same outcome. If $v_1\geq1.9$, then a deviation to the low-bid column gives utility
$2v_1-2$ when $b_2<1.6$, compared with truthful utility $v_1$; this decreases utility since $v_1\leq2$. When $b_2\geq1.6$, the same deviation gives utility
$0$, compared with truthful utility $v_1$.

For buyer~$2$, if $v_2<1.6$, then every feasible bid satisfies $b_2<1.6$, so
again all feasible bids induce the same outcome. If $v_2\geq1.6$, then a
deviation to the low-bid row gives utility $0$ when $b_1<1.9$, compared with
truthful utility $2v_2$; when $b_1\geq1.9$, both rows give buyer~$2$ one unit and
zero payment. Thus no feasible deviation is profitable for either buyer, so the
auction is DSIC under the no-overbidding constraints. Truthful bidding also gives
nonnegative utility in every realized cell, so the auction is IR.

\begin{example}[Revelation principle failure]
\label{ex:failure-of-revelation-principle}
    Consider a single item auction $\A$ with a single buyer with valuation $v \sim \Uni[0,1]$, who is constrained to bid at most twice their valuation ($b \le 2v$). Suppose $\A$ allocates the item whenever the bid is $b\geq0.8$, with zero payment.
    In this auction, only buyers with valuation $v \ge 0.4$ can place a winning bid without violating the buyer constraint. Thus, the item is allocated to the buyers with valuations in the range $[0.4, 1]$.

    If the revelation principle were to hold, there would be a truthful auction allocating the item to exactly the same buyer types at the same price (zero).
    However, this is impossible to implement with a truthful direct mechanism:
    a buyer with valuation $v=0.2$, who receives nothing in the original auction, can now profitably lie
    and bid $\hat{v} = 0.4$.
    This is both profitable and feasible for this buyer,
    even under the constraint (as $0.4 \leq 2 \cdot 0.2$). Thus, there is no truthful mechanism
    equivalent to the original.
\end{example}

\begin{example}[Why binding payments are not always implementable]\label{ex:binding-payment-fail}
Consider a single buyer with valuation $v\in[0,30]$. Let the desired allocation probabilities at equilibrium be:
\begin{equation}
q(v) = \begin{cases}
0.5 & \text{if } 0 \le v \le 10, \\
1 & \text{if } 10 < v \le 20, \\
0.9 & \text{if } 20 < v \le 30.
\end{cases}
\end{equation}

The binding payment rule from~\eqref{eq:rev-aligned-payment-bound} implies that the equilibrium payment of type $v$ must satisfy $\expect{p}(v)=q(v)v-\int_{0}^v q(z)dz$ and thus the expected utility is $\expect{\util}(v)=\int_{0}^v q(z)dz$. Thus at equilibrium, we would like a buyer with type $30$ to get a utility of
$$0.5\cdot10+1\cdot10+0.9\cdot10=24.$$

However, consider the deviation where type $v=30$ mimics the strategy of type $v=20$. Since the constraints are monotone, this deviation is acceptable. The payment required from type $20$ under the binding rule is:
$$\expect{p}(20)=20 \cdot q(20)-\int_{0}^{20} q(z)dz=20\cdot1-15=5.$$
By mimicking type $20$, the buyer with true valuation $v=30$ obtains allocation $q(20)=1$ and pays $\expect{p}(20)=5$, resulting in a utility of
$$30\cdot1-5=25.$$
Since $25>24$, the buyer strictly prefers to deviate. This confirms that the binding payment rule cannot support a non-monotone allocation in equilibrium.
\end{example}

\section{Virtual Values and Ironing}\label{app:ironing}
\subsection{Generic Ironing}\label{sec:ironing}

\begin{definition}[Ironed virtual valuations~\cite{myerson1981optimal,HartlineR08}]\label{def:ironing}
    Given a distribution function $F(\cdot)$ over some interval $I=[\bottomV,\topV]$ and any function $\theta:I\rightarrow\reals$, the \textnormal{ironed virtual value function}, $\Bar{\theta}$, is constructed as follows:

    \begin{enumerate}
        \item For $q\in[0,1]$, define $h(q)=\theta\brackets{F^{-1}(q)}$.
        \item Define $H(q)=\int_0^q h(r)dr$.
        \item Define $\Psi$ as the convex hull of $H$ --- the largest convex function bounded above by $H$ for all $q\in[0,1]$.
        \item Define $\psi(q)$ as the derivative of $\Psi(q)$, where defined, and extend to all of $[0,1]$ by right-continuity.
        \item Finally, $\Bar{\theta}(v)=\psi(F(v))$.
    \end{enumerate}
\end{definition}

    \begin{observation}\label{obs:irn-nondecreasing}
        Any ironed virtual value function defined by Definition~\ref{def:ironing} is nondecreasing. This is immediate by definition, as the derivative of a convex function.
    \end{observation}

The following lemma is generalized to cover allocation functions
that are monotone-\aee, and not necessarily pointwise.
\begin{lemma}[Slight Generalization of Lemma 2.8 from~\cite{HartlineR08}]\label{lemma:ironing-properties}
    Let $F_i$ be a distribution function,
    let $\theta_i:[\bottomV_i,\topV_i]\rightarrow\reals$,
    and let $\expect{x}_i:[\bottomV_i,\topV_i]\to[0,m]$
    be a nondecreasing-\aee interim allocation rule.

    Define $\Psi,H$ and $\Bar{\theta}_i$ as in Definition~\ref{def:ironing}.
    Then
    \begin{equation}
        \E_{v_i\sim F_i}\sqbr{\theta_i(v_i)\expect{x}_i(v_i)}\leq \E_{v_i\sim F_i}\sqbr{\Bar{\theta}_i(v_i)\expect{x}_i(v_i)},
    \end{equation}
    with equality holding if and only if,
    wherever $\Psi(F_i(v_i))<H(F_i(v_i))$,
     $\expect{x}_i(v_i)$ is constant $F_i$-\aee
\end{lemma}
\begin{proof}
Since \(\expect{x}_i\) is nondecreasing \(F_i\)-almost everywhere, there exists a version
\(\tilde{x}_i\) of \(\expect{x}_i\) that is everywhere nondecreasing and satisfies
\(\tilde{x}_i=\expect{x}_i\) \(F_i\)-almost everywhere. Applying Lemma~2.8 of~\cite{HartlineR08} to \(\tilde{x}_i\) yields
\[
\mathbb{E}_{v_i\sim F_i}\!\left[\theta_i(v_i)\tilde{x}_i(v_i)\right]
\le
\mathbb{E}_{v_i\sim F_i}\!\left[\bar{\theta}_i(v_i)\tilde{x}_i(v_i)\right].
\]
Since \(\tilde{x}_i=x_i\) \(F_i\)-almost everywhere, both expectations agree with
the corresponding ones for \(x_i\), and therefore
\[
\mathbb{E}_{v_i\sim F_i}\!\left[\theta_i(v_i)\expect{x}_i(v_i)\right]
\le
\mathbb{E}_{v_i\sim F_i}\!\left[\bar{\theta}_i(v_i)\expect{x}_i(v_i)\right].
\]
Moreover, equality holds if and only if \(\tilde{x}_i\) is constant wherever
\[
\Psi(F_i(v_i)) < H(F_i(v_i)),
\]
equivalently, if and only if \(x_i\) is constant there \(F_i\)-almost everywhere.
\end{proof}

The following result is used throughout the paper, applied on ironed virtual value functions.
\begin{claim}\label{clm:pointwise-maximizer-nondecreasing}
    Suppose $\phi_1,\dots,\phi_n$ are nondecreasing functions mapping buyer valuations to real numbers, $\prec$ is a total order over the feasibility set $\Omega$, and $c:\set{0,\dots,m}\to\reals_{\geq 0}$ is a cost function.

    Then the allocation function $\vec{x}$ defined as
    the pointwise maximizer
    \begin{align*}\vec{x}(\bv)=\max_{\prec}\argmax_{\vec{x}'\in\Omega}\crl*{\sum_{i=1}^n \phi_i(v_i)x_i'-c\brackets{\sum_{i=1}^n x_i'}}
    \end{align*}
    is ex-post monotone nondecreasing.
\end{claim}
\begin{proof}
    Suppose towards contradiction that there is a $j\in[n]$ and some $y<z$ and $\bv_{-j}$ such that
    $x_j(y,\bv_{-j})>x_j(z,\bv_{-j})$. Then since $\vec{x}$ is ``maximizing'', and since feasibility does not depend on the bid profile,
    \begin{align*}
        \sum_{i=1}^n \phi_i((z,\bv_{-j})_i)x_i(z,\bv_{-j})-c\brackets{\sum_{i=1}^nx_i(z,\bv_{-j})}
        \geq \sum_{i=1}^n \phi_i((z,\bv_{-j})_i)x_i(y,\bv_{-j})-c\brackets{\sum_{i=1}^n x_i(y,\bv_{-j})},
    \end{align*}
    and similarly
    \begin{align*}
        &\sum_{i=1}^n \phi_i((y,\bv_{-j})_i)x_i(y,\bv_{-j})-c\brackets{\sum_{i=1}^n x_i(y,\bv_{-j})}
        \geq \sum_{i=1}^n \phi_i((y,\bv_{-j})_i)x_i(z,\bv_{-j})-c\brackets{\sum_{i=1}^n x_i(z,\bv_{-j})}.
    \end{align*}
    Combining both inequalities,
    \begin{align*}
        \phi_j(z)x_j(z,\bv_{-j})+\phi_j(y)x_j(y,\bv_{-j})\geq \phi_j(z)x_j(y,\bv_{-j})+\phi_j(y)x_j(z,\bv_{-j}).
    \end{align*}
    Rearranging,
    \begin{align*}
        \brackets{\phi_j(z)-\phi_j(y)}(x_j(z,\bv_{-j})-x_j(y,\bv_{-j}))\geq 0.
    \end{align*}

Since $x_j(y,\bv_{-j})>x_j(z,\bv_{-j})$ and $\phi_j(\cdot)$ is nondecreasing, it follows that  $\phi_j(y)=\phi_j(z)$.

    But $\vec{x}$ selects a unique maximizer via a fixed deterministic tie-breaking rule ($\max_{\prec}$), so in this case we'd have $x_j(z,\bv_{-j})=x_j(y,\bv_{-j})$, a contradiction.
\end{proof}

\subsection{Revenue-Aligned Virtual Values}\label{app:rev-virtual-values}

For revenue-aligned payoffs, we use a generalized version of Myerson's virtual value, adapted for the payoff function $\Pi$.

\begin{definition}\label{def:rev-aligned-virtual-values}
    Given a revenue-aligned payoff $\Pi=(\payoffValue,\payoffPay)$ and valuation distribution $F_i$ with pdf $f_i$, the \emph{generalized virtual value} $\varphi_{i,\Pi}$ is defined as:
\begin{align*}
    \varphi_{i,\Pi}(v_i)\coloneqq \payoffValue(v_i)+\payoffPay\varphi_i(v_i),
\end{align*}
where $\varphi_i(v_i)=v_i-\frac{1-F_i(v_i)}{f_i(v_i)}$ is the virtual value for revenue as defined in~\cite{myerson1981optimal}.

The \emph{ironed generalized virtual value} $\irn_{i,\Pi}(v_i)$ is defined to be the ironed virtual value of $\varphi_{i,\Pi}$ from Definition~\ref{def:ironing}.
\end{definition}

\begin{observation}
    Observe that for revenue-aligned payoff $\Pi$, every \emph{regular} distribution (in the standard---revenue---sense) is also $\Pi$-regular.
    \end{observation}
\begin{proof}
For every $y<z$:
\begin{align*}
    \varphi_{i,\Pi}(z)-\varphi_{i,\Pi}(y)=\payoffValue(z)-\payoffValue(y)+\payoffPay(\varphi_i(z)-\varphi_i(y)),
\end{align*}
and $\payoffValue(\cdot)$ is nondecreasing and $\payoffPay\geq 0$.
\end{proof}

\subsection{Consumer-Aligned Virtual Values}\label{app:cons-virtual-values}

\begin{definition}\label{def:cons-aligned-virtual-values}
Given a consumer-aligned payoff $\payoffValue(v)x-\conPayoffPay p$,
we define the \emph{consumer-aligned generalized virtual welfare function}:
\begin{equation*}
    \varphi_{i,\Pi}^{m}(v_i)=\payoffValue(v_i)-\frac{\conPayoffPay}{m}\varphi_i(v_i),
\end{equation*}
where $\varphi_i(v_i)$ is the standard virtual value for revenue.
We also define the \emph{ironed consumer-aligned virtual value} $\irn_{i,\Pi}^{m}(v_i)$
as the ironed virtual value of $\varphi_{i,\Pi}^{m}$ from Definition~\ref{def:ironing}.
\end{definition}

We point out that this
notion of virtual value differs from the virtual value for revenue-aligned payoff that we used in Section~\ref{sec:overview-rev-aligned}, as it depends on the number of units $m$.

\section{Randomized Relabeling and Partial Truthification}\label{app:randomized-relabeling-truthification}

This section isolates the randomized relabeling argument and the partial truthification construction used in the revenue-aligned proofs. The first subsection proves the measure-preserving rearrangement result that lets us replace an interim allocation curve by its nondecreasing rearrangement without changing the distribution of bids seen by the rest of the auction. The second subsection uses this relabeling to truthify one buyer at a time while preserving feasibility, nonnegative payments, and the incentives of the other buyers.

\subsection{Randomized Relabeling}\label{app:randomized-relabeling}

The purpose of this subsection is to construct a randomized change of labels that preserves the underlying value distribution while sorting the induced allocation rule. We first state the standard disintegration and transfer results needed to implement a conditional random draw measurably. We then apply them to show that, for any interim allocation curve, there is a randomized relabeling whose induced allocation equals the nondecreasing rearrangement in expectation, while the relabeled value remains distributed according to the original value distribution.

For a random variable $Z$, we denote its law (i.e., its distribution) by $\mathcal L(Z)$.

\begin{proposition}[Restatement of Theorem~8.5 in~\cite{kallenberg1997foundations}]\label{prop:disintegration}

Let $S$ and $T$ be standard Borel spaces, and let $\xi$ and $\eta$ be random
elements taking values in $S$ and $T$, respectively. Then there exists a map (specifically a probability kernel)
\[
K:S\times\mathcal B(T)\to[0,1]
\]
(with $\mathcal B(T)$ the Borel $\sigma$-algebra on $T$) such that:

\begin{enumerate}
\item For each $x\in S$, the set function $A\mapsto K(x,A)$ is a probability
measure on $(T,\mathcal B(T))$; and for each $A\in\mathcal B(T)$, the map
$x\mapsto K(x,A)$ is Borel measurable.

\item For every nonnegative bounded Borel function $f:S\times T\to\mathbb R_{\geq0}$,
\[
\mathbb E[f(\xi,\eta)\mid \xi] \;=\; \int_T f(\xi,t)\,K(\xi,dt)
\qquad\text{a.s.}
\]
In particular, taking $f(\xi,\eta)=\mathbf{1}_{A}(\eta)$ for any $A\in\mathcal B(T)$ yields
\[
K(\xi,A)\;=\; \Pr[\eta\in A\mid \xi]
\qquad\text{a.s.}
\]
\end{enumerate}

\end{proposition}

\begin{proposition}[Restatement of Lemma~4.22 in~\cite{kallenberg1997foundations}]\label{prop:transfer}
Let $S,T$ be standard Borel spaces and let $K:S\times\mathcal B(T)\to[0,1]$ be a probability kernel.
Then there exists a measurable map
\[
\Phi:S\times[0,1]\to T
\]
such that if $R\sim \Uni[0,1]$, then for every $s\in S$,
\[
\mathcal L(\Phi(s,R))=K(s,\cdot).
\]
That is, for each fixed $s$, the random output $\Phi(s,R)$ has distribution $K(s,\cdot)$ on $T$.
\end{proposition}

\begin{lemma}[Randomized relabeling]
\label{lem:random-relabeling}
Let $g:[0,1]\to[0,m]$ be Borel measurable.

Let $g^*$ denote the decreasing rearrangement of $g$
 (so $g$ and $g^*$ are equimeasurable).
Then there exists a measurable map $L:[0,1]\times[0,1]\to[0,1]$ such that for independent
$U,R\sim\Uni[0,1]$:
\begin{enumerate}[label=(\roman*),ref=(\roman*)]
\item\label{item:lem-random-relabeling-i} For $\lambda$-a.e.\ $u\in[0,1]$: \quad $g(L(u,R))=g^*(u)\quad\text{a.s.\ in }R$.
\item\label{item:lem-random-relabeling-ii} $L(U,R)\sim\Uni[0,1]$.
\end{enumerate}
\end{lemma}
\begin{proof}

Let $U_0\sim\Uni[0,1]$.
Set $Y:=g(U_0)\in[0,m]$.

Since $[0,1]$ and $[0,m]$ are standard Borel spaces, we can apply
 Proposition~\ref{prop:disintegration} with $S=[0,m]$ and $T=[0,1]$ and the random
  elements $\eta=U_0$ and $\xi=Y$, and get that
 there exists a probability kernel
 $K:[0,m]\times\mathcal B([0,1])\to[0,1]$ such that for every $y\in[0,m]$,
  $K(y,\cdot)$ is a probability measure on $[0,1]$,
  and for every Borel set $A\subseteq [0,1]$,
\begin{equation}\label{eq:disintegration-riemann}
\Pr[U_0\in A\mid Y]=K(Y,A)\quad\text{a.s.}
\end{equation}

Now we can apply Proposition~\ref{prop:transfer} with $S=[0,m]$, $T=[0,1]$ and the kernel $K$,
giving us a measurable map $\Phi:[0,m]\times[0,1]\to[0,1]$ such that for each $y\in[0,m]$,
 if $R\sim\Uni[0,1]$ is independent of $U_0,Y$, then $\mathcal{L}(\Phi(y,R))= K(y,\cdot)$.

Define the bounded nonnegative Borel function
$f:[0,m]\times[0,1]\to\mathbb R_{\geq0}$ by $f(y,t)=|g(t)-y|$.
Since $Y=g(U_0)$, we have $f(Y,U_0)=0$ a.s.
Hence, a.s.\ in $Y$,
\begin{align*}
    0=\E[f(Y,U_0)\mid Y]=\int_0^1 f(Y,t)\,K(Y,dt)=\int_0^1|g(t)-Y|K(Y,dt)=\
    \int_0^1|g(\Phi(Y,r))-Y|\,dr,
\end{align*}
where the last transition is due to $\mathcal{L}(\Phi(Y,R))= K(Y,\cdot)$.

Define the map                                                                                                                                 $\psi:[0,m]\to\mathbb R_{\ge 0}$ by
\[
\psi(y)\;:=\;\int_0^1 |g(\Phi(y,r))-y|\,dr.
\]
Since $(y,r)\mapsto |g(\Phi(y,r))-y|$ is bounded, Borel measurable and nonnegative, $\psi$ is Borel measurable.
From the previous equation we have $\psi(Y)=0$ a.s.  Let
\[
B\;:=\;\{y\in[0,m]:\ \psi(y)=0\}.
\]
Then $\Pr[Y\in B]=1$.  Moreover, for any fixed $y\in B$,
\[
0=\psi(y)=\int_0^1 |g(\Phi(y,r))-y|\,dr
\]
and since the integrand is nonnegative this implies $g(\Phi(y,r))=y$ for Lebesgue-a.e.\ $r$.
Equivalently, if $R\sim\Uni[0,1]$, then for every $y\in B$,
\begin{equation}\label{eq:fiber-property}
g(\Phi(y,R))=y \qquad \text{a.s.\ in }R.
\end{equation}

Now define
\[
L:[0,1]\times[0,1]\to[0,1],\qquad L(t,r):=\Phi(g^*(t),r).
\]
Since $g^*$ is decreasing (hence Borel measurable) and $\Phi$ is measurable, $L$ is measurable.

\noindent\emph{Proof of Item~\ref{item:lem-random-relabeling-i}.}
Let $U\sim\Uni[0,1]$ be independent of $R$.
Because $g$ and $g^*$ are equimeasurable,
\begin{align*}
    \mathcal{L}(g^*(U))=\mathcal{L}(g(U_0))=\mathcal{L}(Y).
\end{align*}
Hence
\[
\Pr[g^*(U)\in B]=\Pr[Y\in B]=1,
\]
which implies $\lambda(\{u\in[0,1]:\ g^*(u)\in B\})=1$, where $\lambda$ is Lebesgue measure.
Fix any $u$ with $g^*(u)\in B$. Then by \eqref{eq:fiber-property} with $y=g^*(u)$,
\[
g(L(u,R)) = g(\Phi(g^*(u),R)) = g^*(u)\qquad \text{a.s.\ in }R,
\]
where the second transition is due to $g^*(u)\in B$ and the definition of $B$.
This establishes Item~\ref{item:lem-random-relabeling-i}.

\noindent\emph{Proof of Item~\ref{item:lem-random-relabeling-ii}.}
Let $A\subseteq[0,1]$ be Borel. Using that $R$ is independent of $U$ and that
$\mathcal L(\Phi(y,R))=K(y,\cdot)$ for each fixed $y$, we have
\begin{align*}
\Pr[L(U,R)\in A]
&=\E\big[\Pr[\Phi(g^*(U),R)\in A\mid U]\big]
=\E\big[K(g^*(U),A)\big].
\end{align*}
Since $ \mathcal{L}(g^*(U))=\mathcal{L}(Y)$, this equals $\E[K(Y,A)]$.
By \eqref{eq:disintegration-riemann} and the law of total expectation,
\[
\E[K(Y,A)] \;=\; \E[\Pr[U_0\in A\mid Y]] \;=\; \Pr[U_0\in A] \;=\; \lambda(A).
\]
Therefore $\Pr[L(U,R)\in A]=\lambda(A)$ for all Borel $A$, i.e.\ $L(U,R)\sim\Uni[0,1]$.
This proves Item~\ref{item:lem-random-relabeling-ii} and completes the proof.
\end{proof}

\begin{corollary}[Random relabeling]\label{cor:random-relabeling}
Let $\A$ be an auction and let $\vbid$ be a BNE-\aee of $\A$. Fix buyer $i$ and define
\[
y_i(v_i):=
\expect{x}_i(\bid_i(v_i);\vbid).
\]
Let \(x_i^\uparrow:[\bottomV_i,\topV_i]\to[0,m]\) denote the nondecreasing rearrangement of
\(y_i\).

Then there exists a measurable map
\[
\sigma_i:[\bottomV_i,\topV_i]\times[0,1]\to[\bottomV_i,\topV_i]
\]
such that for independent $v_i\sim F_i$ and $r_1\sim\Uni[0,1]$,
\begin{enumerate}
    \item\label{item:cor-random-relabeling-measure}
    $\sigma_i(v_i,r_1)\sim F_i$
    (i.e.\ $\sigma_i$ is $F_i$-measure-preserving), and
    \item\label{item:nondecreasing-equality}
    there exists a set $ T_i\subseteq[\bottomV_i,\topV_i]$ with
     $\Pr_{v_i\sim F_i}[v_i\in  T_i]=1$,
    such that for every $v_i\in  T_i$,
    \[
    \E_{r_1}\!\left[y_i\bigl(\sigma_i(v_i,r_1)\bigr)\right]=x_i^\uparrow(v_i).
    \]
\end{enumerate}

\end{corollary}
\begin{proof}
By the integrability assumption in Section~\ref{sec:auction-prelims}, \(y_i\) is Riemann integrable and therefore Lebesgue measurable. Hence there exists a
Borel measurable function $\tilde y_i:[\bottomV_i,\topV_i]\to[0,m]$
such that \(\tilde y_i(v_i)=y_i(v_i)\) \(F_i\)-almost everywhere. Define
\[
g:[0,1]\to[0,m],
\qquad
g(u):=\tilde y_i\!\left(F_i^{-1}(1-u)\right),
\]
and let \(g^*\) be the decreasing rearrangement of \(g\). Since \(F_i^{-1}\) is Borel measurable,
the function \(g\) is Borel measurable, so Lemma~\ref{lem:random-relabeling} applies.

Let \(L:[0,1]\times[0,1]\to[0,1]\) be the measurable map given by
Lemma~\ref{lem:random-relabeling}, and define
\[
\sigma_i(v_i,r_1):=
F_i^{-1}\!\left(1-L(1-F_i(v_i),r_1)\right).
\]
Since \(F_i\) is continuous and atomless, if \(v_i\sim F_i\) then
\(1-F_i(v_i)\sim \Uni[0,1]\). By Item~\ref{item:lem-random-relabeling-ii} of Lemma~\ref{lem:random-relabeling},
\(L(1-F_i(v_i),r_1)\sim \Uni[0,1]\). Therefore
\[
\sigma_i(v_i,r_1)=F_i^{-1}\!\left(1-L(1-F_i(v_i),r_1)\right)\sim F_i,
\]
which proves Item~\ref{item:cor-random-relabeling-measure}.

It remains to prove Item~\ref{item:nondecreasing-equality}. By Item~\ref{item:lem-random-relabeling-i} of Lemma~\ref{lem:random-relabeling}, for
Lebesgue-a.e.\ \(u\in[0,1]\),
\[
g(L(u,r)) = g^*(u)
\qquad\text{a.s.\ in }r.
\]
Substituting \(u=1-F_i(v_i)\), we obtain that for \(F_i\)-almost every \(v_i\),
\[
g\!\left(L(1-F_i(v_i),r)\right)=g^*(1-F_i(v_i))
\qquad\text{a.s.\ in }r.
\]
By the definitions of \(g\) and \(\sigma_i\),
\[
g\!\left(L(1-F_i(v_i),r)\right)
=
\tilde y_i\!\left(
F_i^{-1}\!\left(1-L(1-F_i(v_i),r)\right)
\right)
=
\tilde y_i(\sigma_i(v_i,r)).
\]
Hence for \(F_i\)-almost every \(v_i\),
\[
\tilde y_i(\sigma_i(v_i,r)) = g^*(1-F_i(v_i))
\qquad\text{a.s.\ in }r,
\]
and therefore
\[
\E_r\!\left[\tilde y_i(\sigma_i(v_i,r))\right]
=
g^*(1-F_i(v_i))
\qquad\text{for }F_i\text{-a.e.\ }v_i.
\]

Let \(V_i\sim F_i\) and \(R\sim\Uni[0,1]\) be independent. By
Item~\ref{item:cor-random-relabeling-measure}, \(\sigma_i(V_i,R)\sim F_i\).
Since \(\tilde y_i=y_i\) \(F_i\)-almost everywhere,
\(\tilde y_i(\sigma_i(V_i,R))=y_i(\sigma_i(V_i,R))\) almost surely.
By Fubini's theorem, for \(F_i\)-almost every \(v_i\), this equality holds
almost surely in \(R\), and hence
\[
\E_r\!\left[\tilde y_i(\sigma_i(v_i,r))\right]
=
\E_r\!\left[y_i(\sigma_i(v_i,r))\right].
\]
Combining the last two equations gives
\[
 \E_r\!\left[y_i(\sigma_i(v_i,r))\right]
=
g^*(1-F_i(v_i))
\qquad\text{for }F_i\text{-almost every }v_i.
\]

We observe that since $F_i$ is nondecreasing and $g^*$ is decreasing, the map $v_i\mapsto g^*(1-F_i(v_i))$ is nondecreasing.
In addition,
since $g$ and $g^*$ are equimeasurable,
\begin{align*}
    \mathcal{L}(g^*(1-F_i(v_i)))&=\mathcal{L}(g(1-F_i(v_i)))=\mathcal{L}(\tilde y_i(v_i))=\mathcal{L}(y_i(v_i)).
\end{align*}
Therefore, \(v_i\mapsto g^*(1-F_i(v_i))\) is a  nondecreasing rearrangement of \(y_i\).
Since the nondecreasing rearrangement is unique up to \(F_i\)-almost-everywhere equality,
\[
g^*(1-F_i(v_i))=x_i^\uparrow(v_i)
\qquad\text{for }F_i\text{-a.e.\ }v_i.
\]
Finally,
\begin{align*}
    \E_r\!\left[y_i(\sigma_i(v_i,r))\right]=x_i^\uparrow(v_i)
\qquad\text{for }F_i\text{-a.e.\ }v_i,
\end{align*}
which proves Item~\ref{item:nondecreasing-equality}.
\end{proof}

\subsection{Partial Truthification}\label{app:partial-truthification}

This subsection turns the relabeling map from Corollary~\ref{cor:random-relabeling} into a local transformation of the auction. For a fixed buyer \(i\), the construction replaces the outcome induced by buyer \(i\)'s bid with the outcome induced by a randomly relabeled bid, and charges a Myerson-style per-unit payment for the rearranged interim allocation. The lemmas below show that this makes buyer \(i\) truthful and monotone almost everywhere, while the other buyers see the same distribution over outcomes and therefore keep the same incentives.

\begin{definition}[Partial Truthification]\label{def:truthification}
    Let $\A$ be an auction with UM-constrained buyers and BNE-$\aee^{(i)}$ $\vbid$,
    let $\sigma_i$ be the measurable map from
    Corollary~\ref{cor:random-relabeling}
    for buyer $i$ and $T_i$ be the set from the same corollary.

    Define the \emph{truthification of $\A$ for buyer $i$} to be the auction $\A'$
    such that for every bid profile $\vec{b}$ and randomness $r=(r_1,r_2)\in[0,1]^2$,
    the outcome of each bidder $j\neq i$ is given by
    \begin{align*}
        \A'_j(\vec{b},r):=\begin{cases}
        \A_j\big((\bid_i(\sigma_i(b_i,r_1)),\vec{b}_{-i}),r_2\big), & \text{if }b_i\in  T_i,\\
            \text{$0$ units and $0$ payment}, & \text{if }b_i\notin  T_i.
        \end{cases}
    \end{align*}

    For buyer $i$,
    the outcome is a random one-shot allocation and payment.
    Specifically, for every bid profile $\vec{b}$ and randomness $r=(r_1,r_2)\in[0,1]^2$,
    the one-shot allocation of buyer $i$ is given by
    \begin{align*}
        &x'_i(\vec{b},r):=\begin{cases}
        x_i\big((\bid_i(\sigma_i(b_i,r_1)),\vec{b}_{-i}),r_2\big), & \text{if }b_i\in  T_i,\\
        0, & \text{if }b_i\notin  T_i.
        \end{cases}
    \end{align*}
    The one-shot payment for buyer $i$ is proportional to the realized allocation,
    with a per-unit price that depends only on $b_i$. Using the notation
    \begin{align*}
        \hat{x}_i(b_i)=\expect{x}_i'(b_i;\vbid)=\E_{r,\bv_{-i}}\!\left[x'_i\big((b_i,\vbid_{-i}(\bv_{-i})),r\big)\right],
    \end{align*}
    define
    \begin{align}\label{eq:lem-proof-relabeling-gamma}
        \gamma_i(b_i):=\begin{cases}
                    b_i-\dfrac{\int_{\bottomV_i}^{b_i}\hat{x}_i(z)dz}{\hat{x}_i(b_i)} & \text{if }\hat{x}_i(b_i)>0,\\[6pt]
                    0 & \text{if }\hat{x}_i(b_i)=0.
                \end{cases}
    \end{align}
    and set
    \begin{align}\label{eq:lem-proof-relabeling-pay}
        p_i'(\vec{b},r):=\begin{cases}
            x'_i(\vec{b},r)\cdot\gamma_i(b_i), & \text{if }b_i\in  T_i,\\
            0, & \text{if }b_i\notin  T_i.
        \end{cases}
    \end{align}

\end{definition}

\begin{lemma}\label{lem:rearranging-preserves-bic}
    Let $\A$ be an auction with BNE-\aee $\vbid$, and let
    $\A'$ be the truthification of $\A$ for buyer $i$ with relabeling $\sigma_i$ from Definition~\ref{def:truthification}.
    Define
    \begin{align*}
        \bid'_j(v_j):= \begin{cases}
        v_j & j=i\\
        \bid_j(v_j) & j\neq i
        \end{cases}
    \end{align*}

    Then for every $j\neq i$, bid $b_j$, and set of outcomes $J$ for bidder $j$,
    \begin{align*}
        \Pr_{r,\bv_{-j}}[\A'_j((b_j,\vbid'_{-j}(\bv_{-j})),r)\in J]=\Pr_{r,\bv_{-j}}[\A_j((b_j,\vbid_{-j}(\bv_{-j})),r)\in J].
    \end{align*}

    In particular, for every valuation $v_j$,
    \begin{align*}
        \E_{\bv_{-j}\sim \vec{F}_{-j}}[\expect{\util}'_j(b_j,\vbid'_{-j}(\bv_{-j});v_j)]
        =\E_{\bv_{-j}\sim \vec{F}_{-j}}[\expect{\util}_j(b_j,\vbid_{-j}(\bv_{-j});v_j)].
    \end{align*}
    That is, the expected utility of bidder $j$ for any bid $b_j$
    and valuation $v_j$ is the same in $\A'$ and $\A$ (assuming the other bidders bid by $\vbid'$ or $\vbid$ respectively).

\end{lemma}

\begin{proof}
Fix a bidder $j\neq i$, valuation $v_j$ and bid $b_j$.
Now let $\bv_{-j}\sim \vec{F}_{-j}$ and $R=(R_1,R_2)\sim\Uni[0,1]^2$ be drawn independently.
Define, for every bid $b_j$, the individual outcomes of bidder $j$ in $\A'$ and $\A$ as follows:
\begin{align*}
    O_j(b_j)\coloneqq \A_j((b_j,\vbid_{-j}(\bv_{-j})),R_2)
    \quad\text{and}\quad
    O'_j(b_j)\coloneqq \A'_j((b_j,\vbid'_{-j}(\bv_{-j})),R).
\end{align*}

By definition of $\A'$ and since $\Pr[T_i]=1$, we have
\begin{align*}
    \mathcal{L}(O_j'(b_j))
    &=\mathcal{L}(\A'_j((b_j,\vbid'_{-j}(\bv_{-j})),R))
    =\mathcal{L}(\A'_j((b_j,\vbid'_{-j}(\bv_{-j})),R)\mid\,v_i\in T_i)
    \\
    &=\mathcal{L}(\A_j((\bid_i(\sigma_i(v_i,R_1)),b_j,\vbid'_{-(i,j)}(\bv_{-(i,j)})),R_2)\mid\,v_i\in T_i)
    \\
    &=\mathcal{L}(\A_j((\bid_i(\sigma_i(v_i,R_1)),b_j,\vbid'_{-(i,j)}(\bv_{-(i,j)})),R_2)),
\end{align*}
where we use the notation $(z_i,z_j,\vec{u}_{-(i,j)})$ to
denote the vector $\vec{u}$ with the $i$-th and $j$-th coordinates removed,
and then adding $z_i$ and $z_j$ back in the $i$-th and $j$-th coordinates respectively.

Since $\sigma_i$ is $F_i$-measure-preserving, for $v_i\sim F_i$ independently of $R_1$, we have that
$\mathcal{L}(\sigma_i(v_i,R_1))=\mathcal{L}(v_i)$ and thus
$\mathcal{L}(\bid_i(\sigma_i(v_i,R_1)))=\mathcal{L}(\bid_i(v_i))$.
Together with the fact that $\bid'_{\ell}(v_{\ell})=\bid_{\ell}(v_{\ell})$ for all $\ell\neq i$, we have
\begin{align*}
    &\mathcal{L}(\A_j((\bid_i(\sigma_i(v_i,R_1)),b_j,\vbid'_{-(i,j)}(\bv_{-(i,j)})),R_2))
    =\mathcal{L}(\A_j((\bid_i(v_i),b_j,\vbid'_{-(i,j)}(\bv_{-(i,j)})),R_2))\\
    &\quad\quad=\mathcal{L}(\A_j((b_j,\vbid_{-j}(\bv_{-j})),R_2))= \mathcal{L}(O_j(b_j)).
\end{align*}
Combining the above equalities, we get
\begin{align*}
    \mathcal{L}(O_j'(b_j))=\mathcal{L}(O_j(b_j)).
\end{align*}

Since the acceptability predicate $\pred_j$ depends only on the
 distribution of the individual outcome of bidder $j$ and the bid, it follows that a bid $b_j$
 and its random outcome are acceptable in $\A'$,
  given valuation $v_j$, if and only if they are acceptable in $\A$.
  Formally, $\pred_j(O_j'(b_j),b_j;v_j)=\pred_j(O_j(b_j),b_j;v_j)$.

  Since \(\mathcal L(O'_j(b_j))=\mathcal L(O_j(b_j))\), and the acceptability predicate \(\pred_j\)
depends only on bidder \(j\)'s bid \(b_j\), valuation \(v_j\), and the distribution of bidder \(j\)'s
individual outcome, we obtain
\[
\pred_j(O'_j(b_j),b_j;v_j)=\pred_j(O_j(b_j),b_j;v_j).
\]
Hence the constrained utility of bidder \(j\) when bidding \(b_j\)
is the same in \(\A'\) and in \(\A\): if this common
predicate is false, both sides equal \(-\infty\); if it is true, then both sides equal the same
unconstrained expected utility, since \(O'_j(b_j)\) and \(O_j(b_j)\) have the same probability law.
 Therefore,
\[
\E_{\bv_{-j}\sim \vec{F}_{-j}}\!\left[\expect{\util}'_j(b_j,\vbid'_{-j}(\bv_{-j});v_j)\right]
=
\E_{\bv_{-j}\sim \vec{F}_{-j}}\!\left[\expect{\util}_j(b_j,\vbid_{-j}(\bv_{-j});v_j)\right].
\]
\end{proof}

\begin{lemma}
    \label{lem:truthification-buyer-i-bic}
    Let $\A$ be an auction with BNE-$\aee^{(i)}$ $\vbid$ and UM-constrained buyers, and let
    $\A'$ be the truthification of $\A$ for buyer $i$ with relabeling $\sigma_i$
    and set $T_i$ from Definition~\ref{def:truthification}.

    Then:
    \begin{enumerate}
        \item
        $\expect{x}_i'(\cdot;\vbid)$
          is monotone-\aee, specifically monotone on $T_i$.
          \item For every $v_i\in T_i$ it holds that
          $\expect{x}_i'(v_i;\vbid)=x_i^\uparrow(v_i)$.
        \item Buyer $i$ gets nonnegative utility
    from truthfully bidding in $\A'$.
        \item $\A'$ is truthful-\aee for buyer $i$ w.r.t.\ $\vbid$.
    \end{enumerate}
\end{lemma}
\begin{proof}
    Fix $i\in[n]$, valuation $v_i$, and let $\A'$ be the truthification of $\A$
    for buyer $i$ with relabeling $\sigma_i$ from Definition~\ref{def:truthification}.
    Let $ T_i\subseteq[\bottomV_i,\topV_i]$ be the set from
     Corollary~\ref{cor:random-relabeling} with $\Pr_{F_i}[ T_i]=1$.

    For monotonicity-\aee, we show that for every $v_i\in T_i$,
    \begin{align*}
        \expect{x}_i'(v_i;\vbid)=\E_{r,\bv_{-i}}\!\left[x'_i\big((v_i,\vbid_{-i}(\bv_{-i})),r\big)\right]
        \underset{(1)}{=}\E_{r_1,r_2,\bv_{-i}}\!\left[x_i\big((\bid_i(\sigma_i(v_i,r_1)),\vbid_{-i}(\bv_{-i})),r_2\big)\right]
        \underset{(2)}{=}x_i^\uparrow(v_i),
    \end{align*}
    where (1) is by definition of $\A'$ for $v_i\in T_i$ and
     (2) is by Item~\ref{item:nondecreasing-equality} of
     Corollary~\ref{cor:random-relabeling} (for every $v_i\in T_i$).
    Thus,
     $\expect{x}_i'(\cdot;\vbid)$
     is monotone on $ T_i$, meaning that it is indeed monotone-\aee

    We continue by showing that bidding truthfully is \emph{acceptable} for buyer $i$
    in $\A'$.
    Recall that since the constraint $\pred_i$ is \emph{uncapped-payment}, by definition if the outcome $S_i$ is
    almost surely one-shot and almost surely results in nonnegative
    unconstrained utility, then $\pred_i^O(S_i;v_i)=\true$.
     By definition of $\A'$, the outcome for buyer $i$ is always one-shot,
     so it remains to verify that it results in
      nonnegative unconstrained utility almost surely.
      For $v_i\notin T_i$ it is immediate
      that
      buyer $i$ gets an unconstrained utility of $0$
      by bidding truthfully.
      Hence, we can focus on $v_i\in T_i$.
    For these valuations the unconstrained utility for
     bids $\vec{b}_{-i}$ and randomness $r$
     is given by
     \begin{align*}
       {\unconstrained}'_i((v_i,\vec{b}_{-i}),r;v_i)&=
        x_i'((v_i,\vec{b}_{-i}),r)\cdot v_i -p_i'((v_i,\vec{b}_{-i}),r)\\
        &=x_i'((v_i,\vec{b}_{-i}),r)\cdot v_i -x_i'((v_i,\vec{b}_{-i}),r)\cdot\gamma_i(v_i).
     \end{align*}
    If $\hat{x}_i(v_i)=0$ then the above is all zero and thus nonnegative.
     Otherwise, by definition of $\gamma_i(v_i)$ in~\eqref{eq:lem-proof-relabeling-gamma},
     we have $\gamma_i(v_i)\leq v_i$, and thus ${\unconstrained}'_i((v_i,\vec{b}_{-i}),r;v_i)\geq 0$.
        Therefore, the random outcome of buyer $i$ that results from buyer $i$
        bidding truthfully in $\A'$ is \emph{acceptable} and results in nonnegative utility.

     It remains to show $\A'$ is truthful-\aee
     for buyer $i$ w.r.t.\ $\vbid$.
     For any valuation $v_i\in  T_i$ and acceptable bid $b_i$, we have
     \begin{align}\label{eq:util-tag-is-unconstrained-tag}
        \E_{r,\bv_{-i}}\sqbr{\util_i'((b_i,\vbid_{-i}(\bv_{-i}));v_i)}
        &=
        \E_{r,\bv_{-i}}\sqbr{{\unconstrained_i}'((b_i,\vbid_{-i}(\bv_{-i}));v_i)},
     \end{align}
      by definition of the constrained utility. We now split into cases:
     \begin{itemize}
        \item If $b_i\in T_i$
     we have
     \begin{align*}
        \E_{r,\bv_{-i}}\sqbr{{\unconstrained_i}'((b_i,\vbid_{-i}(\bv_{-i}));v_i)}
        &\underset{(*)}{=}\E_{r,\bv_{-i}}\sqbr{x_i'((b_i,\vbid_{-i}(\bv_{-i})),r)\cdot v_i -x_i'((b_i,\vbid_{-i}(\bv_{-i})),r)\cdot\gamma_i(b_i)}\\
        &=\hat{x}_i(b_i)v_i-\hat{x}_i(b_i)\gamma_i(b_i),
     \end{align*}
     where $(*)$ is by definition of $\A'$ for buyer $i$.
     If $\hat{x}_i(b_i)=0$ then the above is zero and thus we get
	     $\E_{r,\bv_{-i}}\sqbr{\util_i'((b_i,\vbid_{-i}(\bv_{-i}));v_i)}\leq\E_{r,\bv_{-i}}\sqbr{\util_i'((v_i,\vbid_{-i}(\bv_{-i}));v_i)}$.

     So assume $\hat{x}_i(b_i)>0$. By definition of $\A'$ and $\gamma_i(b_i)$ in~\eqref{eq:lem-proof-relabeling-gamma}, we have
     \begin{align*}
        \E_{r_2,\bv_{-i}}\sqbr{p'_i((b_i,\vbid_{-i}(\bv_{-i})),r_2)}
        &=\hat{x}_i(b_i)\gamma_i(b_i)=\hat{x}_i(b_i)b_i-\int_{\bottomV_i}^{b_i}\hat{x}_i(z)dz,
     \end{align*}
     which is exactly Myerson's payment identity for the interim allocation
     $\hat{x}_i$.

     For the following analysis, we use the fact that
     $\hat{x}_i$ is nondecreasing over $T_i$, and since
     $F_i$ has a positive density function over
     $[\bottomV_i,\topV_i]$, it follows that the Lebesgue
     measure of $[\bottomV_i,\topV_i]\setminus T_i$
is $0$, and thus $\hat{x}_i$ is nondecreasing Lebesgue-\aee

     Therefore, if $b_i\leq v_i$ then
        \begin{align*}
            &\expect{\unconstrained}_i'((v_i,\vbid_{-i}(\bv_{-i}));v_i)-
            \expect{\unconstrained}_i'((b_i,\vbid_{-i}(\bv_{-i}));v_i)
            \\&\quad\quad
            =\hat{x}_i(v_i)v_i-\brackets{\hat{x}_i(v_i)v_i-\int_{\bottomV_i}^{v_i}\hat{x}_i(z)dz}
            -\brackets{\hat{x}_i(b_i)v_i-\brackets{\hat{x}_i(b_i)b_i-\int_{\bottomV_i}^{b_i}\hat{x}_i(z)dz}}
            \\&\quad\quad
            =\int_{b_i}^{v_i}\hat{x}_i(z)dz-\hat{x}_i(b_i)(v_i-b_i)
            =\int_{b_i}^{v_i}\hat{x}_i(z)-\hat{x}_i(b_i)dz
            \geq 0.
        \end{align*}
     Similarly, if $b_i>v_i$ then
     \begin{align*}
            \expect{\unconstrained}_i'((v_i,\vbid_{-i}(\bv_{-i}));v_i)-
            \expect{\unconstrained  }_i'((b_i,\vbid_{-i}(\bv_{-i}));v_i)
            &=-\int_{v_i}^{b_i}\hat{x}_i(z)dz-\hat{x}_i(b_i)(v_i-b_i)
            \\&=\int_{v_i}^{b_i}\hat{x}_i(b_i)-\hat{x}_i(z)dz
            \geq 0.
     \end{align*}

    Either way, together with Eq.~\eqref{eq:util-tag-is-unconstrained-tag} it follows that for every $v_i,b_i\in T_i$,
    \begin{align*}
       \E_{r,\bv_{-i}}\sqbr{\util_i'((v_i,\vbid_{-i}(\bv_{-i}));v_i)}\geq
       \E_{r,\bv_{-i}}\sqbr{{\unconstrained_i}'((b_i,\vbid_{-i}(\bv_{-i}));v_i)}\geq
       \E_{r,\bv_{-i}}\sqbr{{\util_i}'((b_i,\vbid_{-i}(\bv_{-i}));v_i)}.
    \end{align*}

        \item If $b_i\notin T_i$, by definition of $\A'$ we have
        $\E_{r,\bv_{-i}}\sqbr{{\unconstrained_i}'((b_i,\vbid_{-i}(\bv_{-i}));v_i)}
             = 0$,
        so together with Eq.~\eqref{eq:util-tag-is-unconstrained-tag}
        and since we showed that truthful bidding results in
        nonnegative utility for buyer $i$,
        \begin{align*}
            \E_{r,\bv_{-i}}\sqbr{\util_i'((v_i,\vbid_{-i}(\bv_{-i}));v_i)}
            \geq 0=\E_{r,\bv_{-i}}\sqbr{{\unconstrained_i}'((b_i,\vbid_{-i}(\bv_{-i}));v_i)}.
        \end{align*}

     \end{itemize}
     Hence,
     $\A'$ is indeed truthful-\aee for buyer $i$ w.r.t.\ $\vbid$, completing the proof.
\end{proof}

\begin{claim}\label{clm:truthification-nonnegative-pay-allocation-feasible}
    The auction $\A'$ from Definition~\ref{def:truthification}
    has nonnegative payments and is allocation-feasible.
\end{claim}
\begin{proof}
    Let $i$ for which $\A'$ is the truthification of $\A$.

    \emph{Payments are nonnegative:}
    For bid vector $\vec{b}$, if $b_i\notin T_i$ then all buyers pay $0$,
    which is nonnegative. So suppose $b_i\in T_i$.
    For $j\neq i$ the nonnegativity of payments comes
     from the nonnegativity guaranteed in $\A$:
     realized payments for $j\neq i$ were payments made
      in $\A$.
    For bidder $i$,
    the payment of buyer $i$ is given by
    $p_i'(\vec{b},r)=x_i'(\vec{b},r)\cdot\gamma_i(b_i)$.
    The allocation $x_i'(\vec{b},r)$ is nonnegative,
    and since $\hat{x}_i(\cdot)$ is monotone nondecreasing-\aee and specifically monotone on $T_i$ (by Lemma~\ref{lem:truthification-buyer-i-bic}),
    we have that for $F_i$-almost every $z\leq b_i$,
    $\hat{x}_i(z)\leq \hat{x}_i(b_i)$, and thus
    we have
    \begin{align*}
        \gamma_i(b_i)=b_i-\dfrac{\int_{\bottomV_i}^{b_i}\hat{x}_i(z)dz}{\hat{x}_i(b_i)}
        = b_i-\int_{\bottomV_i}^{b_i}\dfrac{\hat{x}_i(z)}{\hat{x}_i(b_i)}dz
        \geq b_i-(b_i-\bottomV_i)=\bottomV_i\geq 0.
    \end{align*}
    Therefore also when $b_i\in T_i$,
     the total payment for bidder $i$ with valuation $b_i$ is nonnegative.

    \emph{It is allocation-feasible:}
    The realized allocation in $\A'$ is
    either $(0,\dots,0)$, which we require in our model
    to be feasible, or
    an allocation that could be
    realized in $\A$,
    just for a different bid profile.

\end{proof}

 \section{Deferred Results and Proofs for Section~\ref{sec:revenue}: Revenue-Aligned Payoff}\label{app:rev-aligned}

 \subsection{Properties of the Ironed Virtual Welfare Maximizer (\revMaximizer)}\label{sec:proof-rev-aligned-optimal-among-tight}

\begin{proof}[Proof of Lemma~\ref{lem:rev-aligned-optimal-among-tight}]
    Let $\A$ be \revMaximizer from Definition~\ref{def:rev-aligned-ironed-maximizer}. First observe that $\A$ is allocation-feasible since all ex-post allocations are directly selected from $\Omega$.

    We now show that its allocation rule $x_i(\cdot,\bv_{-i})$ is nondecreasing for every $i$ and $\bv_{-i}$. Since $\irn_{i,\Pi}(\cdot)$ is nondecreasing (Observation~\ref{obs:irn-nondecreasing}), and because \revMaximizer is defined like in the conditions of Claim~\ref{clm:pointwise-maximizer-nondecreasing}, we can apply  Claim~\ref{clm:pointwise-maximizer-nondecreasing} proving that $\A$ is indeed monotone.

    Now, using the fact that $\A$ is deterministic,
    \begin{align*}
        \E[\A]&=\E_{\bv,r}\sqbr{\sum_{i=1}^n \brackets{\payoffValue(v_i)x_i(\bv)+\payoffPay p_i(\bv)}-c\brackets{\sum_{i=1}^n x_i(\bv)}}\\
        &\quad=\sum_{i=1}^n \brackets{\E_{\bv,r}\sqbr{\payoffValue(v_i)x_i(\bv)}+\payoffPay\E_{\bv,r}\sqbr{p_i(\bv)}}-\E_{\bv,r}\sqbr{c\brackets{\sum_{i=1}^n x_i(\bv)}}.
    \end{align*}
    By~\cite{myerson1981optimal}, $\E_{\bv}\sqbr{p_i(\bv)}=\E_{\bv}\sqbr{\varphi_i(v_i)x_i(\bv)}$, so
    \begin{align*}
        \E[\A]
        &=\sum_{i=1}^n \brackets{\E_{\bv,r}\sqbr{\payoffValue(v_i)x_i(\bv)}+\payoffPay\E_{\bv,r}\sqbr{\varphi_i(v_i)x_i(\bv)}}-\E_{\bv,r}\sqbr{c\brackets{\sum_{i=1}^n x_i(\bv)}}\\
        &=\sum_{i=1}^n \E_{\bv,r}\sqbr{\varphi_{i,\Pi}(v_i)x_i(\bv)}-\E_{\bv,r}\sqbr{c\brackets{\sum_{i=1}^n x_i(\bv)}}.
    \end{align*}

    By Lemma~\ref{lemma:ironing-properties}, $\E_{\bv,r}\sqbr{\varphi_{i,\Pi}(v_i)x_i(\bv)}\leq \E_{\bv,r}\sqbr{\irn_{i,\Pi}(v_i)x_i(\bv)}$. Define $\Psi$ and $H$ from the ironing procedure used to define $\irn_{i,\Pi}$. At points $v_i$ where $\Psi(F_i(v_i))<H(F_i(v_i))$, $\Psi$ is locally linear as the convex hull of $H$, and thus $\irn_{i,\Pi}$ is locally constant. From the way $\vec{x}$ is defined, with a fixed deterministic tie-breaking rule, it will then be constant over these intervals.
    Thus, $\E_{\bv,r}\sqbr{\varphi_{i,\Pi}(v_i)x_i(\bv)}= \E_{\bv,r}\sqbr{\irn_{i,\Pi}(v_i)x_i(\bv)}$ for every $v_i$.

    Continuing the above analysis, we get
    \begin{align}\label{eq:proof-of-argmax-1}
        \E[\A]&=\sum_{i=1}^n \E_{\bv,r}\sqbr{\irn_{i,\Pi}(v_i)x_i(\bv)}-\E_{\bv,r}\sqbr{c\brackets{\sum_{i=1}^n x_i(\bv)}}.
    \end{align}

    For the remaining properties left to prove:
    \begin{itemize}
        \item \textbf{It is allocation-feasible} since all ex-post allocations are directly selected from $\Omega$.
        \item \textbf{It has nonnegative payments}, since $x_i(\cdot,\bv_{-i})$ is nondecreasing for every $i$ and $\bv_{-i}$, and by definition $p_i(\bv)=x_i(\bv)v_i-\int_{\bottomV_i}^{v_i}x_i(z,\bv_{-i})dz\geq x_i(\bv)v_i-x_i(\bv)(v_i-\bottomV_i)\geq 0$.
        \item \textbf{It is IR} since bidding truthfully does not violate a bid constraint, and since it is deterministic and one-shot, no outcome constraint is violated either, so the truthful utility is equal to the unconstrained truthful utility. By definition, and since $x_i(\cdot,\bv_{-i})$ is nondecreasing, $p_i(\bv)\leq x_i(\bv)v_i$, so $\expect{u}_i(v_i;v_i)\geq0$.
        \item \textbf{It is U-DSIC} because for every $i$, $v_i,b_i$ and $\bv_{-i}$,
    \begin{align*}
        &\util_i(\bv;v_i)=\int_{\bottomV_i}^{v_i}x_i(z,\bv_{-i})dz,\quad\text{and}\\
        &\util_i((b_i,\bv_{-i});v_i)=x_i(b_i,\bv_{-i})v_i-x_i(b_i,\bv_{-i})b_i+\int_{\bottomV_i}^{b_i}x_i(z,\bv_{-i})dz\\
        &\quad\quad\quad\quad\quad\quad=x_i(b_i,\bv_{-i})(v_i-b_i)+\int_{\bottomV_i}^{b_i}x_i(z,\bv_{-i})dz.
    \end{align*}

    Using the monotonicity of $x_i(\cdot,\bv_{-i})$, for every $b_i\leq v_i$,
    \begin{align*}
        \util_i(\bv;v_i)-\util_i((b_i,\bv_{-i});v_i)&=\int_{b_i}^{v_i}x_i(z,\bv_{-i})dz
        -x_i(b_i,\bv_{-i})(v_i-b_i)
        \\&\geq (v_i-b_i)x_i(b_i,\bv_{-i})-x_i(b_i,\bv_{-i})(v_i-b_i)
        = 0.
    \end{align*}
    For $b_i>v_i$,
    \begin{align*}
        \util_i(\bv;v_i)-\util_i((b_i,\bv_{-i});v_i)&=-\int_{v_i}^{b_i}x_i(z,\bv_{-i})dz
        +x_i(b_i,\bv_{-i})(b_i-v_i)
        \\&\geq -(b_i-v_i)x_i(b_i,\bv_{-i})+x_i(b_i,\bv_{-i})(b_i-v_i)=0.
    \end{align*}

    Thus bidding truthfully is a dominant strategy with respect to the unconstrained utility, so $\A$ is indeed U-DSIC.
    \end{itemize}
\end{proof}

 \subsection{Payoff Comparison Claims}\label{sec:proof-rev-mono-support}
\begin{claim}\label{clm:not-nondecreasing-means-strict}
    Let \(h : [a,b] \to \mathbb{R}_{\ge 0}\) be measurable (where \(b\) may be \(\infty\)), and let
\(\mu\) be a probability distribution over \([a,b]\) with density \(f\) satisfying
\(f(v) > 0\) for Lebesgue-a.e. \(v \in [a,b]\).
For every \(t \in \mathbb{R}_{\ge 0}\), define
\[
\tau_h(t) := \inf\{ z \in [a,b] : h(z) \ge t\},
\]
with the convention \(\inf \emptyset = \infty\).
Let
\[
M_h := \operatorname*{ess\,sup}_{v\sim\mu} h(v).
\]

If \(h\) is not nondecreasing \(\mu\)-a.e., then there exists a measurable set
\(X \subseteq [0,M_h)\) with positive Lebesgue measure such that for every \(t \in X\),
\[
\Pr_{v\sim\mu}[h(v)\ge t] \;<\; \Pr_{v\sim\mu}[v\ge \tau_h(t)].
\]
In particular, \(\Pr_{v\sim\mu}[h(v)\ge t] > 0\) for every \(t\in X\).
\end{claim}
\begin{proof}
For \(t \in \mathbb{R}_{\ge 0}\), let
\[
V_t := \{v\in[a,b] : h(v)\ge t\}.
\]
Since \(V_t \subseteq [\tau_h(t),b]\) for every \(t\), define
\[
B := \Bigl\{ t\in[0,M_h) : \mu(V_t) < \mu([\tau_h(t),b]) \Bigr\}.
\]
We prove the contrapositive by assuming \(\lambda(B)=0\), where \(\lambda\) denotes Lebesgue
measure, and showing that \(h\) is nondecreasing \(\mu\)-a.e.

Since \(\lambda(B)=0\), the set \([0,M_h)\setminus B\) is dense in \([0,M_h)\). Choose a countable
dense set
\[
D \subseteq [0,M_h)\setminus B.
\]
Define
\[
N_0 := \{v\in[a,b] : h(v) > M_h\},
\qquad
N := N_0 \cup \bigcup_{t\in D}\bigl([\tau_h(t),b]\setminus V_t\bigr).
\]
By definition of \(M_h\), we have \(\mu(N_0)=0\). Also, for every \(t\in D\), since \(t\notin B\),
\[
\mu(V_t)=\mu([\tau_h(t),b]).
\]
As \(V_t\subseteq [\tau_h(t),b]\), this implies
\[
\mu([\tau_h(t),b]\setminus V_t)=0.
\]
Since \(D\) is countable, \(\mu(N)=0\).

We now show that \(h\) is nondecreasing on \([a,b]\setminus N\).
Take \(x<y\) with \(x,y\notin N\). Suppose toward contradiction that \(h(x)>h(y)\).
Because \(x,y\notin N_0\), we have \(h(x),h(y)\le M_h\). Hence we may choose
\(t\in D\) such that
\[
h(y) < t < h(x).
\]
Since \(h(x)\ge t\), we have \(x\in V_t\), and therefore \(x\ge \tau_h(t)\). As \(y>x\), it follows
that \(y\ge \tau_h(t)\), i.e.
\[
y\in [\tau_h(t),b].
\]
But \(y\notin N\), so in particular
\[
y\notin [\tau_h(t),b]\setminus V_t.
\]
Hence \(y\in V_t\), which means \(h(y)\ge t\), contradicting \(h(y)<t\).

Thus \(h\) is nondecreasing on \([a,b]\setminus N\). Since \(\mu(N)=0\), it follows that \(h\) is
nondecreasing \(\mu\)-a.e. This proves the contrapositive, and therefore the claim follows.
\end{proof}
\begin{claim}\label{clm:single-buyer-rearrangement-payoff}
    Let $\A$ be an auction with UM-constrained buyers and
     BNE-$\aee^{(i)}$ $\vbid$ for some $i\in[n]$.
Let $x^\uparrow(v_i)$ be a
     nondecreasing rearrangement of
     $v_i\mapsto\expect{x}_i(\bid_i(v_i);\vbid)$, and
     $p^\uparrow(v_i)$ be the Myerson payment rule
     defined by~\eqref{eq:unique-pay}:
    \begin{align*}
        p^\uparrow(v_i)=x^\uparrow(v_i)v_i-\int_{\bottomV_i}^{v_i}x^\uparrow(z)dz.
    \end{align*}

    Then for any revenue-aligned payoff $(\payoffValue,\payoffPay)$,
    \begin{align*}
        \E_{v_i\sim F_i}[\payoffValue(v_i)x^\uparrow(v_i)+\payoffPay p^\uparrow(v_i)]
        \geq \E_{\bv,r}[\payoffValue(v_i)x_i(\vbid(\bv),r)
        +\payoffPay p_i(\vbid(\bv),r)].
    \end{align*}

    Moreover, if $\expect{x}_i(\bid_i(\cdot);\vbid)$ is not monotone-\aee
    and either $\payoffPay>0$ or $\payoffValue(\cdot)$ is
    strictly increasing, then the inequality is strict.
\end{claim}
\begin{proof}

    First, we show that the welfare term weakly improves from the rearrangement.

Since $F_i$ is continuous and atomless, the random variable
$U := F_i(V_i)$ is uniformly distributed on $[0,1]$.
For any integrable function $g(V_i)$ we may therefore write
$g(V_i) = g(F_i^{-1}(U))$.
Define
\[
\tilde x(u) := \expect{x}_i\brackets{\bid_i\brackets{F_i^{-1}(u)};\vbid},
\qquad
\tilde \payoffValue(u) := \payoffValue(F_i^{-1}(u)),
\qquad u \in [0,1].
\]
Because both $\payoffValue(\cdot)$ and $F_i^{-1}(\cdot)$ are nondecreasing,
$\tilde \payoffValue(\cdot)$ is nondecreasing on $[0,1]$.
Moreover, $x^\uparrow$ being an $F_i$-measure-preserving nondecreasing
rearrangement of $\expect{x}_i(\bid_i(\cdot);\vbid)$ is equivalent to
\[
\tilde x^\uparrow(u) := x^\uparrow(F_i^{-1}(u))
\]
being the nondecreasing rearrangement of $\tilde x$ on $[0,1]$.
By the Hardy--Littlewood rearrangement inequality on $[0,1]$,
\[
\int_0^1 \tilde \payoffValue(u)\,\tilde x^\uparrow(u)\,du
\;\ge\;
\int_0^1 \tilde \payoffValue(u)\,\tilde x(u)\,du.
\]
Translating back via $v_i=F_i^{-1}(u)$ yields
\begin{align}\label{eq:clm-nondecreasing-better}
    \E_{v_i\sim F_i}\sqbr{\payoffValue(v_i)\,x^\uparrow(v_i)}
\geq
\E_{v_i\sim F_i}\sqbr{\payoffValue(v_i)\,\expect{x}_i(\bid_i(v_i);\vbid)}
=\E_{\bv,r}\sqbr{\payoffValue(v_i)\,x_i(\vbid(\bv),r)}.
\end{align}

In addition, if $\payoffValue(\cdot)$ is strictly increasing,
then if $\expect{x}_i(\bid_i(\cdot);\vbid)$ is not monotone-\aee, by Hardy-Littlewood
 the inequality is strict.

    Next, define
    \begin{align*}
        \tau_x(t)=\inf\set{z\in[\bottomV_i,\topV_i]:\ \expect{x}_i(\bid_i(z);\vbid)\geq t}
    \end{align*}
    and
    \begin{align*}
        \tau^\uparrow(t)=\inf\set{z\in[\bottomV_i,\topV_i]:\ x^\uparrow(z)\geq t},
    \end{align*}
    where by convention $\inf\emptyset=\bottomV_i$.

    We will show that $\tau^\uparrow(t)\geq \tau_x(t)$ for every $t\in[0,m]$.

Fix $t$. We have
\begin{align}\label{eq:clm-proof-x-then-tau}
    \Pr_{v_i\sim F_i}\sqbr{\expect{x}_i(\bid_i(v_i);\vbid)\geq t}\leq \Pr_{v_i\sim F_i}\sqbr{v_i\geq \tau_x(t)}.
\end{align}

If $\expect{x}_i(\bid_i(\cdot);\vbid)$ is not monotone $F_i$-\aee,
then by Claim~\ref{clm:not-nondecreasing-means-strict}
there exists a set $X\subseteq[0,m]$ with positive Lebesgue measure
such that for every $t\in X$,
inequality~\eqref{eq:clm-proof-x-then-tau} is strict
and $\Pr_{v_i\sim F_i}[\expect{x}_i(\bid_i(v_i);\vbid)\geq t]>0$.

Since $x^\uparrow$ is an $F_i$--measure-preserving rearrangement
 of $\expect{x}_i(\bid_i(\cdot);\vbid)$,
and using the fact that $x^\uparrow$ is nondecreasing and $F_i$ is atomless,
\begin{align*}
    \Pr_{v_i\sim F_i}\sqbr{\expect{x}_i(\bid_i(v_i);\vbid)\geq t}=\Pr_{v_i\sim F_i}\sqbr{x^\uparrow(v_i)\geq t}=\Pr_{v_i\sim F_i}\sqbr{v_i\geq \tau^\uparrow(t)}.
\end{align*}

Plugging this into~\eqref{eq:clm-proof-x-then-tau},
we get
$\Pr_{v_i\sim F_i}\sqbr{v_i\geq \tau^\uparrow(t)}\leq\Pr_{v_i\sim F_i}\sqbr{v_i\geq \tau_x(t)}$, and it immediately follows that
 $\tau_x(t)\leq\tau^\uparrow(t)$,
 with the inequality being strict for every $t\in X$.

    Now, due to Claim~\ref{clm:tau-bound},
    \begin{align}\label{eq:clm-pay-by-rearranged-tau}
        \E_{v_i\sim F_i}\sqbr{ \payoffPay\expect{p}_i(\bid_i(v_i);\vbid)}\leq
        \E_{v_i\sim F_i}\sqbr{\payoffPay\int_0^{\expect{x}_i(\bid_i(v_i);\vbid)} \tau_x(t)dt}
        \leq
        \E_{v_i\sim F_i}\sqbr{\payoffPay\int_0^{\expect{x}_i(\bid_i(v_i);\vbid)} \tau^\uparrow(t)dt}.
    \end{align}

    If $\expect{x}_i(\bid_i(\cdot);\vbid)$ is not monotone-\aee
    then using $X$ defined above,

    \begin{align*}
       & \E_{v_i\sim F_i}\sqbr{\int_0^{\expect{x}_i(\bid_i(v_i);\vbid)} \tau^\uparrow(t)dt}
        -
        \E_{v_i\sim F_i}\sqbr{\int_0^{\expect{x}_i(\bid_i(v_i);\vbid)} \tau_x(t)dt}
        \\&\qquad=
         \int_0^m \brackets{ \tau^\uparrow(t)-\tau_x(t)}\Pr_{v_i}[\expect{x}_i(\bid_i(v_i);\vbid)\geq t]dt
        \\&\qquad
        \geq\int_{t\in X} \brackets{ \tau^\uparrow(t)-\tau_x(t)}\Pr_{v_i}[\expect{x}_i(\bid_i(v_i);\vbid)\geq t]dt.
        \numberthis\label{eq:tau-proof-positive-X}
    \end{align*}
Let
\[
M := \operatorname*{ess\,sup}_{v_i\sim F_i} \expect{x}_i(\bid_i(v_i);\vbid).
\]
Since \(\expect{x}_i(\bid_i(\cdot);\vbid)\in[0,m]\), we have \(M\le m\).

If \(\expect{x}_i(\bid_i(\cdot);\vbid)\) is not monotone-a.e., then by Claim~\ref{clm:not-nondecreasing-means-strict} applied to
\[
h(v_i) := \expect{x}_i(\bid_i(v_i);\vbid),
\]
there exists a measurable set \(X\subseteq[0,M)\) with positive Lebesgue measure such that for
every \(t\in X\),
\[
q(t) := \Pr_{v_i\sim F_i}[\expect{x}_i(\bid_i(v_i);\vbid)\ge t]
\;<\;
\Pr_{v_i\sim F_i}[v_i\ge \tau_x(t)].
\]
Because \(x^\uparrow\) is an \(F_i\)-measure-preserving nondecreasing rearrangement of
\(\expect{x}_i(\bid_i(\cdot);\vbid)\),
\[
q(t)
=
\Pr_{v_i\sim F_i}[x^\uparrow(v_i)\ge t]
=
\Pr_{v_i\sim F_i}[v_i\ge \tau^\uparrow(t)].
\]
Also, since \(t<M\), by the definition of essential supremum we have \(q(t)>0\).

Therefore, for every \(t\in X\),
\[
\Pr_{v_i\sim F_i}[v_i\ge \tau^\uparrow(t)]
=
q(t)
<
\Pr_{v_i\sim F_i}[v_i\ge \tau_x(t)].
\]
Because \(F_i\) is atomless and has density positive a.e., the tail function
\(s\mapsto \Pr_{v_i\sim F_i}[v_i\ge s]\) is strictly decreasing wherever it is positive. Hence,
for every \(t\in X\),
\[
\tau^\uparrow(t) > \tau_x(t).
\]
Combining this with \(q(t)>0\), we obtain that for every \(t\in X\),
\[
\bigl(\tau^\uparrow(t)-\tau_x(t)\bigr)\Pr_{v_i\sim F_i}[\expect{x}_i(\bid_i(v_i);\vbid)\ge t]
=
\bigl(\tau^\uparrow(t)-\tau_x(t)\bigr)q(t)
> 0.
\]
Thus the integrand in~\eqref{eq:tau-proof-positive-X} is measurable, nonnegative everywhere, and strictly positive on a
set \(X\) of positive Lebesgue measure. Therefore,
\[
\int_{t\in X}
\bigl(\tau^\uparrow(t)-\tau_x(t)\bigr)
\Pr_{v_i\sim F_i}[\expect{x}_i(\bid_i(v_i);\vbid)\ge t]
\,dt
>0.
\]
So the second inequality in~\eqref{eq:clm-pay-by-rearranged-tau} is strict whenever \(\payoffPay>0\) and
\(\expect{x}_i(\bid_i(\cdot);\vbid)\) is not monotone-a.e.

    Thus, when $\payoffPay>0$ and $\expect{x}_i(\bid_i(\cdot);\vbid)$
    is not monotone-\aee, the second inequality of~\eqref{eq:clm-pay-by-rearranged-tau}
    is strict.

    By Tonelli, since $\tau^\uparrow(t)$ is always nonnegative,
    \begin{align*}
        \E_{v_i\sim F_i}\sqbr{\int_0^{\expect{x}_i(\bid_i(v_i);\vbid)} \tau^\uparrow(t)dt}
        &=\E_{v_i\sim F_i}\sqbr{\int_0^{\infty} \tau^\uparrow(t)\mathbb{I}_{t\leq\expect{x}_i(\bid_i(v_i);\vbid)}dt}
        =\int_0^m \tau^\uparrow(t)\Pr_{v_i\sim F_i}[\expect{x}_i(\bid_i(v_i);\vbid)\geq t]dt\\
        &=\int_0^m \tau^\uparrow(t)\Pr_{v_i\sim F_i}[x^\uparrow(v_i)\geq t]dt\numberthis\label{eq:clm-based-on-alpha}
        =\E_{v_i\sim F_i}\sqbr{\int_0^{x^\uparrow(v_i)} \tau^\uparrow(t)dt}.
   \end{align*}
    Transition~\eqref{eq:clm-based-on-alpha} is due to $x^\uparrow$ being $F_i$--measure-preserving.

    Plugging this into~\eqref{eq:clm-pay-by-rearranged-tau},
    and using
    Claim~\ref{clm:nondecreasing-payment-threshold-int}
    with $y(\cdot):=x^\uparrow(\cdot)$,
    \begin{align*}
        \payoffPay\E_{v_i\sim F_i}\sqbr{ \expect{p}_i(\bid_i(v_i);\vbid)}
        &\leq\payoffPay\E_{v_i\sim F_i}\sqbr{\int_0^{x^\uparrow(v_i)} \tau^\uparrow(t)dt}
  \\&\underset{(1)}{=}\payoffPay\E_{v_i\sim F_i}\sqbr{ x^\uparrow(v_i)v_i-\int_{\bottomV_i}^{v_i}x^\uparrow(z)dz}
  =\payoffPay\E_{v_i\sim F_i}\sqbr{p^\uparrow(v_i)}.
    \end{align*}
    The last transition comes from how we defined $p^\uparrow(\cdot)$.

    Together with Eq.~\eqref{eq:clm-nondecreasing-better},
     $\payoffPay\geq0$ and the linearity of expectation,
      we are done showing the weak inequality.

      In the case that $\payoffPay>0$ and $\expect{x}_i(\bid_i(\cdot);\vbid)$ is not monotone-\aee,
      we get a strict inequality, as we showed above.

\end{proof}

\end{document}